\documentclass[twoside,11pt]{article}

\usepackage[margin=1.0in]{geometry}
\usepackage{xcolor}
\usepackage{amsmath}
\usepackage{ mathtools, thmtools}
\usepackage{pdfsync}
\usepackage[linesnumbered,ruled,vlined]{algorithm2e}

\usepackage{hyperref}       %
 \hypersetup{ hidelinks }
\usepackage{url}            %
\usepackage{booktabs}       %
\usepackage{amsfonts}       %
\usepackage{nicefrac}       %
\usepackage{microtype}      %
\usepackage{xcolor}         %
\usepackage{enumitem}
\usepackage{cite}
\usepackage{multirow}
\usepackage{graphicx}
\usepackage{subfig}
\usepackage{amsthm, amsmath, amssymb,mathtools, thmtools}

\newtheorem{vignette}{Vignette}
\newtheorem{assumption}{Assumption}
\newtheorem{lemma}{Lemma}
\newtheorem{theorem}{Theorem}
\newtheorem{corollary}{Corollary}
\newtheorem{proposition}{Proposition}
\newtheorem{remark}{Remark}
\newtheorem{definition}{Definition}
\newtheorem{example}{Example}
\usepackage{thm-restate}
\usepackage[ruled,vlined]{algorithm2e}
\usepackage{graphicx}
\usepackage[sort]{natbib}
\usepackage{subfig}
\usepackage[export]{adjustbox}

\newcommand{\LL}{\mc{L}}
\newcommand{\PP}{\mc{P}}
\newcommand{\bz}{\zeta}
\let\cite\citep

\DeclareMathOperator*{\ee}{\mathbb{E}}

\newcommand{\argmin}{\operatornamewithlimits{argmin}}

\newcommand{\R}{\ensuremath{\mathbb{R}}}
\newcommand{\mc}{\mathcal}
\newcommand{\mb}{\mathbb}

\newcommand{\zj}{z_i}
\newcommand{\lj}{\ell_i}
\newcommand{\query}{\delta}

\newcommand{\Dz}[1]{\mc{D}_{#1}}
\newcommand{\Dzj}{\Dz{i}}

\DeclareMathOperator{\proj}{proj}

\DeclareMathOperator*{\amin}{argmin}

\DeclareMathOperator*{\E}{\mathbb{E}}

\newcommand{\Blower}{c_l}

\newcommand{\Bl}{c_l}
\newcommand{\Bui}{c_{u,i}}

\newcommand{\players}{players}
\newcommand{\player}{player}
\newcommand{\X}{\mathcal{X}}

\newcommand{\Z}{\mathcal{Z}}

\newcommand{\cD}{\mathcal{D}}

\newcommand{\Na}{n}
\newcommand{\tr}{\textrm{tr}}

\begin{document}

\title{Multiplayer Performative Prediction:\\ Learning in Decision-Dependent
Games\thanks{Drusvyatskiy's research was supported by NSF
    DMS-1651851 and CCF-2023166 awards. Ratliff's research was supported by NSF
CNS-1844729 and Office of Naval Research YIP Award N000142012571. Fazel's
research was supported in part by awards NSF TRIPODS II-DMS 2023166, NSF
TRIPODS-CCF 1740551, and NSF CCF 2007036.}}
\author{Adhyyan Narang\thanks{Department of Electrical and Computer Engineering,
University of Washington, Seattle}, Evan Faulkner$^\ast$, Dmitriy
Drusvyatskiy\thanks{Department of Mathematics, University of Washington,
Seattle},
Maryam Fazel$^\ast$,
Lillian J. Ratliff$^\ast$}

\maketitle

\begin{abstract}
Learning problems commonly exhibit an interesting feedback mechanism wherein the population data  reacts to competing decision makers' actions.  This paper formulates a new game theoretic framework for this phenomenon, called {\em multi-player performative prediction}. We focus on two distinct solution concepts, namely $(i)$ performatively stable equilibria and $(ii)$ Nash equilibria of the game. The latter equilibria are arguably more informative, but can be found efficiently only when the game is monotone.
We show that under mild assumptions, the performatively stable equilibria can be found efficiently by a variety of algorithms, including repeated retraining and the repeated (stochastic) gradient method. We then establish transparent sufficient conditions for strong monotonicity of the game and use them to develop algorithms for finding Nash equilibria. We investigate derivative free methods and adaptive gradient algorithms wherein each player alternates between learning a parametric description of their distribution and gradient steps on the empirical risk. Synthetic and semi-synthetic numerical experiments illustrate the results.

\end{abstract}

\section{Introduction}
\label{sec:introduction}

Supervised  learning theory and algorithms crucially rely on the training and testing data being generated from the same distribution. This assumption, however, is often violated in contemporary applications because data distributions may ``shift" in reaction to the decision maker's actions. Indeed, supervised 
learning algorithms are increasingly being trained on data that is generated by strategic or even adversarial agents, and deployed in environments that react to the decisions that the algorithm makes. In such settings, the model learned on the training data may be unsuitable for downstream inference and prediction tasks. 

The method most commonly used in machine learning practice to address such distributional shifts is to periodically retrain the model to adapt to the changing distribution \citep{diethe2019continual,wu2020deltagrad}. Consequently, it is important to understand when such retraining heuristics converge and what types of solutions they find.
Despite the ubiquity of retraining heuristics in practice, one should be aware that training without consideration of strategic effects or decision-dependence can lead to unintended consequences, including reinforcing bias. This is a concern for applications with 
potentially significant social impact, such as predictive policing~\citep{lum2016predict}, criminal sentencing~\citep{angwin2016propub,courtland2018bias}, pricing equity in ride-share markets~\citep{chen2015peek}, and loan or job procurement \citep{bartlett2019consumer}.

Optimization over decision-dependent probabilities has classical roots in operations research; 
see for example the review article of \cite{hellemo2018decision} and  references therein. The more recent work of \citep{perdomo2020performative}, motivated by the strategic classification literature \citep{dong2018strategic,hardt2016strategic,miller2020strategic}, sets forth an elegant framework---aptly named \emph{performative prediction}---for modeling decision-dependent data distributions in machine learning settings.
There is now extensive research that develops  algorithms for performative prediction by leveraging  advances in convex  optimization \citep{drusvyatskiy2020stochastic, miller2021outside, mendler2020stochastic,perdomo2020performative,brown2020performative}. 

The existing strategic classification and performative prediction literature
focuses solely on the interplay between a single learner and the population that
reacts to the learner's actions. However, learning algorithms in
practice are often deployed alongside other algorithms which may even be
competing with one another. Concrete examples to keep in mind are those of
college admissions and loan procurement, wherein the applicants may tailor their
profile to make them more desirable for the college of their choice, or the loan with the terms (such as interest rate) that 
match the
applicant's current socio-economic and fiscal situation. In these cases, there are multiple competing learners (colleges, banks) and the population reacts based on the admissions policies of all the colleges (or banks) simultaneously. Examples of this type are widespread in applications; we provide further motivating vignettes in Section~\ref{sec:setup}.

\subsection{Contributions} 
\label{sec:contributions}

We formulate the first game theoretic model for decision-dependent learning in
the presence of competition,  called {\em multi-player performative
prediction}.\footnote{A preliminary version of this paper appeared in the
Proceedings of the 25th International Conference on Artificial Intelligence and
Statistics, 2022.}
This is a new class of stochastic games that model a variety of machine learning 
problems arising in many practical applications. 
The model captures, as a special case, 
important problems including strategic classification in settings with
multiple decision-making entities that model learning when each entity's data
distribution depends on the action taken. 
It opens up an entire new class of
problems that can be studied to determine important socio-economic implications
of using machine learning algorithms (including classifiers and predictors) in
settings where the ``data'' generated for training is produced by strategic
users that react based on their own internal preferences. 

We focus on two solution concepts for such games: $(i)$ performatively stable
equilibria and $(ii)$ Nash equilibria. The former arises naturally when
decision-makers employ na\"ive repeated retraining algorithms. This is very
common practice, and hence it is important to understand the equilibrium 
to which such algorithms converge and precisely when they do so.
We show that performatively stable
equilibria are sure to exist and to be unique under reasonable smoothness,
convexity, and Lipschitz assumptions (Section~\ref{sec:algorithms}). Moreover, repeated training and the
(stochastic) gradient methods succeed at finding such equilibrium strategies. The
finite time efficiency estimates (or iteration complexity) we obtain reduce to state-of-the art guarantees
 in the single player setting.

In some applications, a performatively stable equilibrium may be a poor solution
concept and instead a  
Nash equilibrium may be desirable. In particular, as machine
learning algorithms become more sophisticated, in the sense that at the time of
learning decision-dependence is taken into consideration, a more natural
equilibrium concept is a Nash equilibrium. Aiming towards algorithms for finding
Nash equilibria, we develop transparent conditions ensuring strong monotonicity
of the game (Section~\ref{sec:monotonicity}). Assuming that the game is strongly monotone, we then discuss a number of algorithms for finding Nash equilibria (Section~\ref{sec:algorithms_nash}). In particular, derivative-free methods are immediately applicable but have a high sample complexity $\mathcal{O}(\frac{d^2}{\epsilon^2})$ \cite{bravo2018bandit,drusvyatskiy2021improved}.
Seeking faster algorithms, we introduce an additional assumption that the data distribution depends linearly on the performative effects of all the players.
When the players know explicitly how the distribution depends on their own performative effects, but not those of their competitors, a simple stochastic gradient method is applicable and comes equipped with an efficiency guarantee of $\mathcal{O}(\frac{d}{\varepsilon})$. 
Allowing players to know their own performative effects may be unrealistic in some settings.  Consequently, we propose an \emph{adaptive} algorithm in the setting when the data distribution has an amenable parametric description. In the algorithm, the players alternate between estimating the parameters of the distribution and optimizing their loss, again with only empirical samples of their individual gradients given the estimated parameters. The sample complexity for this algorithm, up to variance terms, matches the rate $\mathcal{O}(\frac{d}{\varepsilon})$ of the stochastic gradient method.

Finally, we present illustrative numerical experiments using both a synthetic
example to validate the theoretical bounds, and a semi-synthetic example
generated using data from multiple ride-share platforms (Section~\ref{sec:numerics}). 

\subsection{Related Work}
\label{sec:lit_review}

\paragraph{Performative Prediction.} The multiplayer setting in the present paper is inspired by the single player  performative prediction  framework introduced by \cite{perdomo2020performative}, and further refined by \cite{mendler2020stochastic} and \cite{miller2021outside}. These works introduce the notions of performative optimality and stability, and show that repeated retraining and stochastic gradient methods convergence to a stable point. Subsequently, \citep{drusvyatskiy2020stochastic} showed that a variety of popular gradient-based algorithms in the decision-dependent setting can be understood as the analogous algorithms applied to a certain static problem corrupted by a vanishing bias. 
 In general, performative stability does not imply performative optimality. Seeking to develop algorithms for finding performatively optimal points, \cite{miller2021outside} provide sufficient conditions for the  prediction problem to be convex. 
For decision-dependent distributions described as location families, \cite{miller2021outside} additionally introduce a two-stage algorithm for finding performatively optimal points.   
The paper \citep{izzo2021learn} instead focuses on algorithms that estimate gradients with finite differences.
The performative prediction framework is largely motivated by the
problem of strategic classification \cite{hardt2016strategic}. This problem has been studied extensively from the perspective of
causal inference \cite{bechavod2020causal,miller2020strategic} and convex optimization
\cite{dong2018strategic}.

Another line of work in performative prediction has focused on the setting in
which the environment evolves dynamically in time or experiences time drift.
This line of work is more closely related to reinforcement learning wherein a
decision maker attempts to maximize their reward over time given that the
stochastic environment depends on their decision. In particular, \citep{brown2020performative} formulate a time-dependent performative prediction problem such that the decision-maker seeks to optimize the
stationary reward---i.e., the reward under the fixed point distribution induced
by the player's decision. Repeated retraining algorithms seeking the performatively stable solution
to this problem are studied. In contrast, \citep{ray2021dynamic} study
performative prediction in 
geometrically decaying environments, and provide conditions and algorithms that lead to the 
performatively optimal solution. 
The papers \citep{cutler2021stochastic} and \citep{wood2021online} study performative prediction problems wherein the environment is drifting not only due to the action of the decision maker
but also in time. These two papers analyze the tracking efficiency of the proximal stochastic gradient method and projected gradient descent under time drift. %

\paragraph{Gradient-Based Learning in Continuous Games.} There is a broad and growing literature on learning in games. We focus 
 on the most relevant subset: gradient-based learning in continuous games. 
In his seminal work, \citep{rosen1965existence} showed that convex games which are \emph{diagonal strictly
convex} 
admit a unique Nash equilibrium and the gradient method converges to it. 
There is a large literature extending this work to more general games.
For instance, \cite{ratliff2016characterization} provide a characterization of Nash equilibria in non-convex continuous games, and show that continuous time gradient dynamics locally converge to Nash; building on this work, \cite{chasnov20a} provide local convergence rates that extend to global rates when the game admits a potential function or is strongly monotone.

Under the assumption of strong monotonicity, 
 the iteration complexity of  stochastic  and derivative-free gradient
 methods has also been obtained \citep{mertikopoulos2019learning,
 bravo2018bandit,drusvyatskiy2021improved}. Relaxing strong monotonicity to
 monotonicity,   \citet{tatarenko2019learning,tatarenko2020bandit}  
show that the stochastic gradient and derivative free gradient methods---i.e.,
where players use a single-point query of the loss to construct an estimate of
their individual gradient
 of a smoothed version of their  loss function---converge asymptotically. The
 approach to deal with the lack of strong monotonicity is to add a
 regularization term that decays to zero asymptotically. The update players
 employ in this regularized game is then analyzed 
 as a stochastic gradient method with an additional bias term. We take a similar
 perspective to \citet{tatarenko2019learning,tatarenko2020bandit} and \citep{drusvyatskiy2020stochastic} in the analysis of all the algorithms we study---namely, we view
 the updates as a stochastic gradient method with additional bias.

\paragraph{Stochastic programming.} Stochastic optimization problems with decision-dependent uncertainties have  appeared in the classical stochastic programming literature, such as \citep{ahmed2000strategic,dupacova2006optimization,jonsbraaten1998class,rubinstein1993discrete,varaiya1988stochastic}. We  refer the reader to the recent paper \cite{hellemo2018decision}, which discusses taxonomy and various models of decision dependent uncertainties. An important theme of these works is to utilize structural assumptions on how the decision variables impact the distributions. Consequently, these works sharply deviate from the framework explored in  \cite{perdomo2020performative,mendler2020stochastic,miller2021outside}  and  from  our paper. 
Time-varying problems have also been studied under the title ``nonstationary optimization problems'' in, e.g., \cite{gaivoronskii1978,ermoliev1988}, where it is assumed that the time varying functions converges to a limit but there is no explicit assumption on decision or state-feedback.

\section{Notation and Preliminaries} \label{sec:preliminaries}
This section records basic notation that we will use. A reader that is familiar with convex games and the Wasserstein-1 distance between probability measures may safely skip this section.
Throughout, we let $\R^d$ denote a $d$-dimensional Euclidean space, with inner produce $\langle \cdot,\cdot\rangle$ and the induced norm $\|x\|=\sqrt{\langle x,x\rangle}$. The projection of a point $y\in\R^d$ onto a set $\X\subset\R^d$ will be denoted by
$$\proj_{\X}(y)=\argmin_{x\in\X} \|x-y\|.$$
The normal cone to a convex set $\X$ at $x\in\X$, denoted by $N_{\X}(x)$ is the set
$$N_{\X}(x)=\{v\in\R^d: \langle v,y-x\rangle\leq 0~~\forall y\in \X\}. $$

\subsection{Convex Games and Monotonicity}
Fix an index set $[n]=\{1,\ldots, n\}$, integers $d_i$ for $i\in[n]$, and set $d=\sum_{i=1}^n d_i$. We will always decompose vectors $x\in \R^d$ as $x=(x_1,\ldots, x_n)$ with $x_i\in \R^{d_i}$. Given an index $i$, we abuse notation and write $x = (x_i, x_{-i})$, where $x_{-i}$ denotes the vector of all coordinates except
$x_i$.
A collection of functions $\LL_i\colon\R^{d}\to \R$ and sets $\X_i\subset\R^{d_i}$, for $i\in [n]$, induces a game between $n$ players, wherein each player $i$ seeks to solve the problem
\begin{equation}\label{eqn:game_prelim}
\min_{x_i\in \X_i}~\LL_{i}(x_i,x_{-i}).
\end{equation}
Define the joint action space $\mc{X}=\mc{X}_1\times \cdots \times\mc{X}_n$.
A vector $x^{\star}\in\R^d$ is called a {\em Nash equilibrium} of the game \eqref{eqn:game_prelim} if the condition
$$x_i^{\star}\in \argmin_{x_i\in \X_i}~\LL_i(x_i,x_{-i}^{\star})\qquad \textrm{holds for each }i\in [n].$$
Thus $x^{\star}$ is a Nash equilibrium  if each
player $i$ has no incentive to deviate from $x_{i}^\star$ when the strategies of all
other players remain fixed at $x_{-i}^\star$.

We use $\nabla_i \LL(x)$ to denote the
derivative of $\LL(\cdot)$ with respect to $x_i$. With this notation, we define
the vector of individual gradients \[H(x):=(\nabla_1 \LL_1(x),\ldots,\nabla_n
\LL_n(x)).\]
This map arises naturally from writing down first-order optimality conditions
corresponding to \eqref{eqn:game_prelim}
for each player. Namely, we say that \eqref{eqn:game_prelim} is a  {\em
$C^1$-smooth convex game} if the sets $\X_i$ are closed and convex, the
functions $\LL_{i}(\cdot,x_{-i})$ are convex (i.e., $\LL_i$ is convex in $x_i$ when $x_{-i}$ are fixed), and the partial gradients $\nabla_i
\LL_i(x)$ exist and are continuous. Thus, the Nash equilibria $x^{\star}$ are characterized by the inclusion
$$0\in H(x^{\star})+N_{\X}(x^{\star}).$$

Generally speaking, finding global Nash equilibria is only possible for ``monotone''
games. A $C^1$-smooth convex game is called {\em $\alpha$-strongly monotone}
(for $\alpha\geq 0$) if it satisfies
$$\langle H(x)-H(x'),x-x'\rangle\geq \alpha\|x-x'\|^2\qquad \textrm{for all }x,x'\in\R^d.$$
If this condition holds with $\alpha=0$, the game is simply called {\em
monotone}. It is well-known from \citep{rosen1965existence}  that
$\alpha$-strongly monotone games (with $\alpha>0$)
over convex, closed and bounded strategy sets $\X$ admit a {\em unique} Nash equilibrium $x^{\star}$, which satisfies
$$\langle H(x),x-x^{\star}\rangle\geq \alpha\|x-x^{\star}\|^2\qquad \textrm{for all }x\in \X.$$

\subsection{Probability Measures and Gradient Deviation}
To simplify notation, we will always assume that when taking expectations with respect to a measure that the expectation exists and that integration and differentiation operations may be swapped whenever we encounter them. These assumptions are completely standard to justify under uniform integrability conditions.

We will be interested in random variables taking values in a  metric space.
Given a metric space $\Z$ with metric  $d(\cdot,\cdot)$, the symbol $\mathbb{P}(\Z)$  will denote the set of Radon probability measures $\mu$ on $\Z$ with a finite first moment $\ee_{z\sim \mu} [d(z,z_0)]<\infty$ for some $z_0\in \Z$.  We  measure the deviation between two measures $\mu,\nu\in \mathcal{Z}$ using the Wasserstein-1 distance:
\begin{equation}\label{eqn:KR}
W_1(\mu,\nu)=\sup_{h\in {\rm Lip}_1}\,\left\{\E_{X\sim \mu}[h(X)]-\E_{Y\sim \nu}[h(Y)]\right\},
\end{equation}
where ${\rm Lip}_1$ denotes the set of $1$-Lipschitz continuous functions $h\colon\Z\to\R$.
Fix a function $g\colon\R^d\times \Z\to \R$ and a measure $\mu\in \mathbb{P}(\Z)$, and define the expected loss
$$f_{\mu}(x)=\ee_{z\sim \mu}g(x,z).$$
The following standard result shows that the Wasserstein-1 distance controls how
the gradient $\nabla f_{\mu}(x)$ varies with respect to $\mu$; see, for
example,~\citet[Lemmas 1.1, 2.1]{drusvyatskiy2020stochastic} or \cite[Lemma C.4]{perdomo2020performative} for a short proof.

\begin{lemma}[Gradient deviation]\label{lem:exchange_integral_diff}
Fix a function $g\colon\R^d\times \Z\to \R$ such that $g(\cdot,z)$  is $C^1$-smooth  for all $z\in \Z$ and the map $z\mapsto\nabla_x g(x,z)$ is $\beta$-Lipschitz continuous for any $x\in\R^d$.
Then for any measures $\mu,\nu\in \mathbb{P}(\mathcal{Z})$,  the estimate holds:
$$\sup_x \|\nabla f_{\mu}(x)-\nabla f_{\nu}(x)\|\leq \beta\cdot W_1(\mu,\nu).$$
 \end{lemma}

\section{Decision-Dependent Games}
\label{sec:setup}

We model the problem of $n$ decision-makers, each facing a decision-dependent learning problem, as an $n$-{\player} game.
Each player $i\in [n]$ seeks to solve the decision-dependent optimization problem %
\begin{equation}\label{eqn:main_problem}
   \min_{x_i \in \X_i} \LL_i(x_i,x_{-i})\qquad \textrm{where}\qquad \LL_i(x):=\E_{\zj \sim \Dz{i}(x)} \lj(x,z_i).
\end{equation}
Throughout, we suppose that each set $\X_i$ lies in the Euclidean space $ \R^{d_i}$
and we set $d=\sum_{i=1}^{\Na} d_i$.
The loss function for the $i$'th player is denoted as $\lj\colon \R^{d} \times
\Z_i \to \R$, where $\Z_i$ is some metric
space and $\mathcal{D}_i(x)\in
\mathcal{P}(\Z_i)$ is a probability measure that depends on the joint decision $x\in \X$ and the player $i\in [n]$.  Observe that the  
random variable $\zj$ in the objective function of player $i$ is governed by the
distribution $\Dzj(x)$, which crucially depends on the strategies
$x=(x_1,\ldots, x_n)$ chosen by \emph{all} players. This is worth emphasizing:
the parameters chosen by one player have an influence on the data seen by all
other players. This is one of the critical ways in which the problems for the
different players are strategically coupled. The other is directly through the
loss function $\ell_i$ which also depends on the joint decision $x$. These two
sources of strategic coupling are why the game theoretic abstraction naturally arises. 
It is worth keeping in mind that in most practical settings (see the upcoming Vignettes \ref{vin:same} and \ref{vin:different}), the loss functions $\lj(x,z_i)$ depend only on $x_i$, that is $\lj(x,z_i)\equiv \lj(x_i,z_i)$. If this is the case, we will call the game 
{\em separable} (which refers to separable losses, not distributions). %
Thus, for separable games, the coupling among the players is due entirely to the distribution $\Dzj(x)$ that depends on the actions of all the players.

\begin{remark}{\rm 
The decision-dependence in the distribution map may encode the reaction of strategic users in a population to the 
announced joint decision $x$; hence, in these cases there is also a ``game" between the decision-makers  and the strategic users in the environment---a game %
with a different interaction structure known as a Stackelberg game \citep{von2010market}. 
This level of strategic interaction between decision-maker and strategic user is
abstracted away to an aggregate level in $\mc{D}_i(x_i,x_{-i})$. The game
between a single decision maker and the strategic user population has been
studied widely (cf.~Section~\ref{sec:lit_review}). We leave it to
future work to investigate both layers of strategic interaction simultaneously.
}
\end{remark}

We assume that each {\player} has full information of the other {\players}'
parameters. This is a reasonable assumption in our setup: if the data population
(e.g., strategic users) are able to respond to the {\players}' deployed
decisions
$x_i$, the other {\players} must be able to respond to these decisions as well.
In essense, these decisions are publicly announced. The following vignettes based on practical applications motivate different types of strategic coupling. 
\begin{vignette}[\textbf{Multiplayer forecasting}] 
{\rm Players have the same decision-dependent data distribution---namely, 
 $\mc{D}\equiv \Dzj \equiv \mc{D}_j$ for all $i, j\in[n]$.
   Multiple mapping applications forecast the travel time between different locations, yet the realized travel time is collectively influenced by all their forecasts. The decision-dependent {\players} are the mapping applications (firms).
    The decision $x_i$ each {\player} makes is the rule for recommending routes. Users choose routes, which then impact the realized travel time $z_i\equiv z\sim \mc{D}$ on the $m$ different road segments in the network observed by all firms.
   }
\label{vin:same}
\end{vignette}

\begin{vignette}
{\rm Players have different distributions---i.e., $\mc{D}_i\not\equiv\mc{D}_j$ for all $i,j\in [n]$. 
\begin{enumerate}[label=(\textbf{\alph*}), topsep=0pt, itemsep=-2pt]
    \item \textbf{Multiplayer Strategic Classification.} Multiple universities classify students as accepted or rejected using applicant data, where each applicant designs their application to fit desiderata of multiple universities.  The data $z_i\sim \mc{D}_i(x)$ is an application that university $i$ receives, and as a decision-dependent player, each university $i$  designs a classification rule $x_i$ to determine which applicants are accepted.
The goal of a university is to accept qualified candidates, and different types
of universities predominently cater  to different populations  (e.g., liberal
arts versus science and engineering), yet students may apply to multiple
programs across many universities thereby resulting in distinct distributions
$\mc{D}_i$ that depend on the joint decision rule $x$. 
\item \textbf{Revenue Maximization via Demand Forecasting.} In a ride-share
    market, multiple platforms
    forecast demand for rides (respectively, supply of drivers) at different
    locations in order to optimize their revenue by using the forecast to set
    prices. In most ride-share markets, drivers and passengers participate in multiple
    platforms by, e.g., ``price shopping''. Hence, the forecasted demand $z_i\sim\mc{D}_i(x)$  for platform $i$ depends on their
    own decision $x_i$ as well as their competitors' decisions $x_{-i}$. 
    \end{enumerate}
 }

\label{vin:different}
\end{vignette}

Prior formulations of decision dependent learning do not readily extend to the settings described in the vignettes without a game theoretic model. There are two natural solution concepts for the game \eqref{eqn:main_problem}. The first is the classical notion of Nash equilibrium; we repeat the definition here for ease of reference.
\begin{definition}[Nash Equilibrium]   \label{def:perf_opt}
{\rm
A vector $x^{\star}\in \X$ is a {\em Nash equilibrium} of the game \eqref{eqn:main_problem} if 
the inclusion 
$$x_i^{\star}\in \argmin_{x_i\in \X_i}~\LL_i(x_i,x_{-i}^\star)\qquad \textrm{holds for each }i\in [n].$$
}
\end{definition}
As previously observed, generally speaking, finding Nash equilibria is only possible for monotone games. The game \eqref{eqn:main_problem} can easily fail to be monotone even if the game is separable and the loss functions $\ell_i(\cdot,z)$ are strongly convex. In Section~\ref{sec:monotonicity}, we develop sufficient conditions for (strong) monotonicity and use them to analyze algorithms for finding Nash equilibria. The sufficient conditions we develop, which are in line with those in the single player setting \citep{miller2021outside}, are necessarily quite restrictive.

 On the other hand, there is an alternative solution concept, which is more amenable to numerical methods, and reduces to performatively stable points of \cite{perdomo2020performative} in the single player setting. The idea is to decouple the effects of a decision $x$ on the integrand $\ell(x,z)$ and on the distribution $\mathcal{D}(x)$. Setting notation, any vector $y\in \X$ induces a static game (without performative effects) wherein the distribution for player $i$ is fixed at $\mathcal{D}_i(y)$, that is:
\begin{equation}\label{eqn:param_fam_games}
\min_{x_i\in \X_i}~\LL_i^y(x_i,x_{-i})\qquad \textrm{where}\qquad \LL_i^y(x_i,x_{-i}):=\E_{z_i\sim \mathcal{D}_i(y)}\ell_i(x_i,x_{-i},z_i).\tag*{$\mathcal{G}(y)$}
\end{equation}
Notice that this is a parametric family of static games, indexed by $y\in \X$.

\begin{definition}[Performatively stable equilibria]
{\em
A strategy vector $x^{\star}\in \X$ is a {\em performatively stable equilibrium} of the game \eqref{eqn:main_problem} if it is a Nash equilibrium of the static game $\mathcal{G}(x^{\star})$ (with game $\mathcal{G}(\cdot)$ as defined above). 
}
\end{definition}

The difference between performatively stable equilibria and Nash equilibria of
is that the governing distribution for player $i$ is
fixed at $\mathcal{D}_i(x^{\star})$. Performatively stable equilibria have a
clear intuitive meaning: each player $i$ has no incentive to deviate from
$x^{\star}$ having access only to data drawn from $\mathcal{D}(x^{\star})$.
Notice that if the game \eqref{eqn:main_problem} is separable---the typical
setting---the static game $\mathcal{G}(y)$ is entirely decoupled for any $y$ in
the sense that  the problem that player $i$ aims to solve depends only on
$x_i$ and not on $x_{-i}$. This enables a variety of single player optimization
techniques to extend to the computation of performatively stable equilibria. 

Nash and performatively stable equilibria are typically distinct, even in the single player setting as explained in \cite{perdomo2020performative}. This distinction is worth highlighting. 
Under mild smoothness assumptions, the chain rule directly implies the following expression for the gradient of player $i$'th objective:
\begin{equation}\label{eqn:grad_exp}
\nabla_i  \LL_i(x_i,x_{-i})=
\underbrace{\E_{z_i\sim\mc{D}_i(x)}[\nabla_i\ell_i(x_i,x_{-i},z_i)]}_{P_i}+
\underbrace{\frac{d}{du_i}\E_{z_i\sim\mc{D}_i(u_i,x_{-i})}[\ell_i(x_i,x_{-i},z_i)]\Big|_{u_i=x_i}}_{Q_i}.
\end{equation}
If $x$ is a Nash equilibrium of the game \eqref{eqn:main_problem}, then equality
$0=\nabla_i \LL_i(x_i,x_{-i})=P_i+Q_i$ holds for all $i\in [n]$. On the other hand, provided the loss functions $\ell_i$ are convex, $x$ is a performatively stable equilibrium precisely when $P_i=0$ for all $i\in[n]$. This clearly shows that the two solution concepts are typically distinct, since performative stability in essence ignores the term $Q_i$ on the right side of \eqref{eqn:grad_exp}. 
It is an open question as to how these equilibrium concepts
        compare in terms of their efficiency relative to the social optimum.  We explore this empirically in Section~\ref{sec:numerics}.

In the rest of the paper we impose the following assumption that is in line with the performative prediction literature.

\begin{assumption}[Convexity and smoothness]\label{assump:convex}
There exist constants $\alpha>0$ and $\beta_i> 0$ such that for each $i\in [n]$, the following hold:
\begin{enumerate}[itemsep=-1pt, topsep=5pt]
    \item  For any $y\in\X$, the game $\mathcal{G}(y)$ is $\alpha$-strongly monotone.     
    \item Each loss $\ell_i(x_i,x_{-i},z_i)$  is $C^1$-smooth in $x_i$ and the map $z_i\mapsto\nabla_i \ell_i(x,z_i)$ is $\beta_i$-Lipschitz
        continuous for any $x\in\X$. 
  \end{enumerate}
\end{assumption}

It is worth noting that in the setting where the losses are seperable, the
game $\mathcal{G}(y)$ is $\alpha$-strongly monotone as long as each expected
loss $\E_{z\sim \mathcal{D}_i(y)}\ell_i(x_i,z_i)$ is $\alpha$-strongly convex in
$x_i$. Assumption~\ref{assump:convex} alone {\em does not} imply convexity of
the objective functions $\LL_i(x_i,x_{-i})$ in $x_i$ nor monotonicity of the game
\eqref{eqn:main_problem} itself. Sufficient conditions for convexity and strong
monotonicity of the game are given in Section~\ref{sec:monotonicity}.

Next, we require the distributions $\cD_i(x)$ to vary in a Lipschitz way with respect to $x$. 
\begin{assumption}[Lipschitz distributions]\label{assum:perm_pred}
	For each $i\in[n]$, there exists $\gamma_i>0$ satisfying 	
 $$W_1(\cD_i(x),\cD_i(y))\leq \gamma_i\cdot\|x-y\|\qquad \textrm{ for all }x,y\in \X.$$
 In this case, we define the constant $\rho:=\sqrt{\sum_{i=1}^n (\frac{\beta_i\gamma_i}{\alpha})^2}$.
\end{assumption}

We will see in Theorem~\ref{thm:rrm_converge} that the constant $\rho$ is fundamentally important for algorithms, since it characterizes settings in which algorithms can be expected to work.
We end this section with some convenient notation that will be used throughout.
To this end, fix two vectors $x=(x_1,\ldots, x_n)\in \X$ and $z=(z_1,\ldots, z_n)\in \Z_1\times \ldots\times \Z_n$. We then set 
$$g_i(x,z_i):=\nabla_i \ell_i(x,z_i)\qquad \textrm{and}\qquad g(x,z):=(g_1(x,z_1),\ldots, g_n(x,z_n)).$$
Taking expectations define 
\begin{equation}\label{eqn:game_jacobian}
G_{i,y}(x):=\E_{z_i\sim \mathcal{D}_i(y)}g_i(x,z_i)\qquad \textrm{and}\qquad G_{y}( x):=(G_{1,y}(x),\ldots, G_{n,y}( x)).
 \end{equation}
Thus $G_{y}(\cdot)$ is the vector of individual gradients corresponding to the game \ref{eqn:param_fam_games}. Notice that we may write
$$G_{y}( x):=\E_{z\sim \mathcal{D}_{\pi}(y)} g(x,z)$$
where $\mathcal{D}_{\pi}(y):=\mathcal{D}_1(y)\times\ldots\times\mathcal{D}_n(y)$ is the product measure.
 The following is a direct consequence of Lemma~\ref{lem:exchange_integral_diff}.

\begin{lemma}[Deviation in the vector of individual gradients]\label{lem:dev_game_jac}
Suppose Assumptions~\ref{assump:convex} and \ref{assum:perm_pred} hold. Then for
every $x,y,y'\in \X$ and index $i\in[n]$, the estimates hold:
\begin{align*}
\|G_{i,y}(x)-G_{i,y'}(x)\|&\leq \beta_i\gamma_i \cdot\|y-y'\|,\\
\|G_y(x)-G_{y'}(x)\|&\leq \sqrt{\sum_{i=1}^n \beta_i^2\gamma_i^2} \cdot\|y-y'\|.
\end{align*}
\end{lemma}
\begin{proof}
Using Lemma~\ref{lem:exchange_integral_diff} and the standing Assumptions~\ref{assump:convex} and \ref{assum:perm_pred} we compute
\begin{align*}
\|G_{i,y}(x)-G_{i,y'}(x)\|&=\left\|\E_{z_i\sim \mathcal{D}_i(y)}\nabla_i \ell_i(x,z_i)-\E_{z_i\sim \mathcal{D}_i(y')}\nabla_i \ell_i(x,z_i)\right\|\\
&\leq  \beta_i \cdot W_1(\mathcal{D}_i(y),\mathcal{D}_i(y'))\\
&\leq  \beta_i\gamma_i \cdot \|y-y'\|.
\end{align*}
Therefore, we deduce 
\begin{align*}
\|G_y(x)-G_{y'}(x)\|^2=\sum_{i=1}^n \|G_{i,y}(x)-G_{i,y'}(x)\|^2&\leq \sum_{i=1}^n \beta_i^2\gamma_i^2 \cdot \|y-y'\|^2.
\end{align*}
The proof is complete.
\end{proof}

\section{Algorithms for Finding Performatively Stable Equilibria}
\label{sec:algorithms}

The previous section isolated two solution concepts for decision dependent
games: Nash equilibria and performatively stable equilibria. In this section, we
discuss existence of the latter and algorithms for finding these. We
discuss three algorithms: repeated retraining, the repeated gradient method, and the
repeated stochastic  gradient method. While the first two are largely conceptual, the
repeated stochastic  gradient method is entirely implementable.

\subsection{Repeated Retraining}
Observe that performatively stable equilibria of \eqref{eqn:main_problem} are precisely the fixed points of the map
$$\mathtt{Nash}(x):=\{x'\in\X:  x'\textrm{ is a Nash equilibrium of the game }\mathcal{G}(x)\}.$$
A natural algorithm is therefore {\em repeated retraining}, which is simply the fixed point iteration
\begin{equation}\label{eqn:repeated_nash}
x^{t+1}=\mathtt{Nash}(x^t).
\end{equation}
In the single player settings, this algorithm was studied in
\citet{perdomo2020performative}. Unrolling notation, given a current decision
vector $x^t$, the updated decision vector $x^{t+1}$ is the Nash equilibrium of the game wherein each player $i\in [n]$ seeks to solve
\begin{equation}\label{eqn:monotono_game_exp}
\min_{y_i\in \X_i} \E_{z_i\sim \mathcal{D}_i(x^t)}\ell_i(y_i,y_{-i},z_i).
\end{equation}
It is important to notice that since $x^t$ is fixed, the game \eqref{eqn:monotono_game_exp} is strongly monotone in light of Assumption~\ref{assump:convex}.
Thus, in iteration $t$, the players play a Nash equilibrium (i.e., a best response) in this
game induced by the prior joint decision $x^t$.  Importantly, despite the fact that $x^{t+1}$ is a
Nash equilibrium of a game in iteration $t$,
the collective decision $x^{t+1}$ is \textbf{not} a Nash equilibrium for
the multiplayer performative prediction problem \eqref{eqn:main_problem}. In fact, players have an
incentive to change their action since it is possible that by changing $x_i^t$,
the change it induces in the distribution $\mc{D}_i(\cdot)$ reduces their expected loss.

The following theorem shows that in the regime $\rho<1$, the game \eqref{eqn:main_problem} admits a unique performatively stable equilibrium and repeated retraining converges linearly.

\begin{theorem}[Repeated retraining]
\label{thm:rrm_converge}
Suppose Assumptions~\ref{assump:convex}-\ref{assum:perm_pred} hold and that we are in the regime $\rho<1$.
Then the game \eqref{eqn:main_problem} admits a unique performatively stable equilibrium $x^{\star}$ and the repeated retraining process \eqref{eqn:repeated_nash} converges linearly:
 \[\|x^{t+1}-x^{\star} \|\leq \rho\cdot \|x^t-x^{\star}\|\qquad \textrm{for all}
 ~t\geq 0.\]
\end{theorem}
\begin{proof}
We will show that the map $\mathtt{Nash}(\cdot)$ is Lipschitz continuous with parameter $\rho$. To this end, consider two points $x$ and $x'$ and set
$y:=\mathtt{Nash}(x)$ and $y':=\mathtt{Nash}(x')$. Note that first order optimality conditions for $y$ and $y'$ guarantee
$$\langle G_{x}(y),y-y'\rangle\leq 0\qquad \textrm{and}\qquad \langle G_{x'}(y'),y'-y\rangle\leq 0.$$
Strong monotonicity of the game $\mathcal{G}(x)$ therefore ensures
\begin{align*}
\alpha\cdot\|y-y'\|^2&\leq  \langle G_{x}(y)-G_{x}(y'),y-y'\rangle\\
&\leq  \langle G_{x'}(y')-G_{x}(y'),y-y'\rangle\\
&\leq \| G_{x'}(y')-G_{x}(y'))\|\cdot\|y-y'\|\ \\
&\leq \sqrt{\sum_{i=1}^n \gamma_i^2\beta_i^2}
\cdot\|x-x'\|\cdot \|y-y'\|,
\end{align*}
where the last inequality follows from Lemma~\ref{lem:dev_game_jac}.
Dividing through by $\|y-y'\|$ guarantees that $\mathtt{Nash}(\cdot)$ is indeed a
contraction with parameter $\rho$. The result follows immediately from the Banach fixed point theorem.
\end{proof}

Theorem~\ref{thm:rrm_converge} shows that the interesting parameter regime is
$\rho<1$, since outside of this setting performative equilibria may even fail to
exist. 
It is worth noting that when the game \eqref{eqn:main_problem} is separable, each iteration of repeated retraining \eqref{eqn:repeated_nash}  becomes
\begin{equation}\label{eqn:monotono_game_exp}
\min_{y_i\in \X_i} \E_{z_i\sim \mathcal{D}_i(x^t)}\ell_i(y_i,z_i).
\end{equation}
That is, the optimization problems faced by the players are entirely independent of each other. In the separable case, the regime when repeated retraining succeeds can be slightly enlarged from $\rho<1$ to $\sum_{i=1}^n\left(\frac{\beta_i\gamma_i}{\alpha_i}\right)^2<1$, where $\alpha_i$ is the strong convexity constant of the loss for player $i$. This is the content of the following theorem, whose proof is a slight modification of the proof of Theorem~\ref{thm:rrm_converge}.

\begin{theorem}[Improved contraction for separable games]
Suppose that the game \eqref{eqn:main_problem} is separable and that each loss
function $\LL_i^y(x_i)=\E_{z_i\sim \mathcal{D}_i(y)}\ell_i(x_i,z_i)$ is
$\alpha_i$-strongly convex in $x_i$ for every $y\in \X$. Suppose moreover that
Assumptions~\ref{assump:convex} and \ref{assum:perm_pred} hold, and that we are in the regime $\sum_{i=1}^n\left(\frac{\beta_i\gamma_i}{\alpha_i}\right)^2<1$. The game \eqref{eqn:main_problem} admits a unique performatively stable equilibrium $x^{\star}$ and the repeated retraining process converges linearly:
 \[\|x^{t+1}-x^{\star} \|\leq
 \sqrt{\sum_{i=1}^n\left(\tfrac{\beta_i\gamma_i}{\alpha_i}\right)^2}\cdot
 \|x^t-x^{\star}\|\qquad \textrm{for all}~t\geq 0.\]
 \end{theorem}
\begin{proof}
We will show that the map $\mathtt{Nash}(\cdot)$ is Lipschitz continuous with parameter  $\sqrt{\sum_{i=1}^n\left(\frac{\beta_i\gamma_i}{\alpha_i}\right)^2}$. To this end, consider two points $x$ and $x'$ and set
$y:=\mathtt{Nash}(x)$ and $y':=\mathtt{Nash}(x')$. Note that first order optimality conditions for $y$ and $y'$ guarantee
$$\langle G_{i,x}(y),y_i-y'_i\rangle\leq 0\qquad \textrm{and}\qquad \langle G_{i,x'}(y'),y'_i-y_i\rangle\leq 0\qquad \textrm{for all }i\in [n].$$
Set $v=(\alpha_1^{-1},\ldots, \alpha_n^{-1})$ and let $\odot$ denote the Hadamard product between two vectors.
 Strong convexity of the loss functions therefore ensures
\begin{align*}
\|y-y'\|^2&\leq \sum_i \alpha_i^{-1}\langle G_{i,x}(y)-G_{i,x}(y'),y_i-y'_i\rangle\\
&\leq  \sum_i \alpha_i^{-1}\langle G_{i,x'}(y')-G_{i,x}(y'),y_i-y'_i\rangle\\
&=  \langle v\odot ( G_{x'}(y')-G_{x}(y')),y-y'\rangle\\
&\leq \|v\odot ( G_{x'}(y')-G_{x}(y'))\|\cdot\|y-y'\|\ \\
&\leq \sqrt{\sum_{i=1}^n\frac{\beta_i^2\gamma_i^2}{\alpha_i^2}}
\cdot\|x-x'\|\cdot \|y-y'\|,
\end{align*}
where the last inequality follows from Lemma~\ref{lem:exchange_integral_diff} and the standing Assumptions~\ref{assump:convex} and \ref{assum:perm_pred}.
Dividing through by $\|y-y'\|$ guarantees that $\mathtt{Nash}(\cdot)$ is indeed a
contraction with parameter $\sqrt{\sum_{i=1}^n\frac{\beta_i^2\gamma_i^2}{\alpha_i^2}}$. The result follows immediately from the Banach fixed point theorem.
\end{proof}

\subsection{Repeated Gradient Method}
\label{sec:rgd}
Repeated retraining is largely a conceptual algorithm since in each iteration it
requires computation of the exact Nash equilibrium of a stochastic game
\eqref{eqn:monotono_game_exp}. A more realistic algorithm would instead take a
single gradient step on the game \eqref{eqn:monotono_game_exp}. With this in
mind, given a step-size parameter $\eta> 0$, the {\em repeated gradient method} repeats the updates:
$$x^{t+1}=\proj_{\X}(x^t-\eta G_{x^t}(x^t)).$$
More explicitly, in iteration $t$, each player $i\in [n]$ performs the update
$$x^{t+1}_i=\proj_{\X_i}\left(x^t-\eta_t \E_{z_i\sim \mathcal{D}_i(x^t)}\nabla_i \ell_i(x_i^t,x_{-i}^t,z_i)\right).$$
This algorithm is largely conceptual since each player needs to compute an expectation; nonetheless we next show that this process converges linearly under the following additional smoothness assumption.

\begin{assumption}[Smoothness]\label{assump:smoothnessL}
For every $y\in\X$, the vector of individual gradients $G_{y}(x)$ is $L$-Lipschitz continuous in $x$.
\end{assumption}

The following is the main result of this section.

\begin{theorem}[Repeated gradient method]
Suppose that Assumptions~\ref{assump:convex}-\ref{assump:smoothnessL} hold and
that we are in the regime $\rho<\frac{1}{2+\sqrt{2}}$. Then the iterates $x^t$
generated by the repeated gradient method with $\eta=\frac{\alpha}{L^2}$ converge
linearly to the %
performatively stable equilibrium $x^{\star}$---that is, the following
estimate holds:
\begin{equation}\label{eqn:grad_discnt_dcrs}
\|x^{t+1}-x^{\star}\|\leq
\left(\frac{1}{\sqrt{1+\frac{\alpha^2}{L^2}}}+\frac{\alpha^2\rho}{L^2}\right)\|x^t-x^{\star}\|
\qquad \textrm{for all } t\geq 0.
\end{equation}
\end{theorem}
\begin{proof}
Using the triangle inequality, we estimate
\begin{equation}\label{eqn:start_here}
\begin{aligned}
\|x^{t+1}-x^{\star}\|&=\|\proj_{\X}(x^t-\eta G_{x^t}(x^t))-x^{\star}\|\\
		     &\leq \|\proj_{\X}(x^t-\eta G_{x^{\star}}(x^t))-x^{\star}\|+\|\proj_{\X}(x^t-\eta G_{x^t}(x^t))-\proj_{\X}(x^t-\eta G_{x^{\star}}(x^t))\|\\
&\leq \|\proj_{\X}(x^t-\eta G_{x^{\star}}(x^t))-x^{\star}\|+\eta\| G_{x^t}(x^t))- G_{x^{\star}}(x^t))\|\\
&\leq  \|\proj_{\X}(x^t-\eta G_{x^{\star}}(x^t))-x^{\star}\|+\eta\sqrt{\sum_i \beta_i^2\gamma_i^2}\cdot \|x^t-x^{\star}\|,
\end{aligned}
\end{equation}
where the last inequality follows from Lemma~\ref{lem:dev_game_jac}.
The rest of the argument is standard. We will simply show that the first term on
the on right-side is a fraction of $\|x^t-x^*\|$. To this end, set
$y^{t+1}=\proj_{\X}(x^t-\eta G_{x^{\star}}(x^t))$. Since the function $x\mapsto
\frac{1}{2}\|x^t-\eta G_{x^{\star}}(x^t)-x\|^2$ is 1-strongly convex and $y^{t+1}$ is its minimizer over $\X$, we deduce
\begin{align*}
\frac{1}{2}\|y^{t+1}-x^{\star}\|^2
\leq \frac{1}{2}\|x^t-\eta G_{x^{\star}}(x^t)-x^{\star}\|^2-\frac{1}{2}\|x^t-\eta G_{x^{\star}}(x^t)-y^{t+1}\|^2.\end{align*}
Expanding the right hand side yields
\begin{equation}\label{eqn:young}
\begin{aligned}
\frac{1}{2}\|y_{t+1}-x^{\star}\|^2&\leq \frac{1}{2}\|x^t-x^{\star}\|^2-\eta\langle   G_{x^{\star}}(x^t), y^{t+1}-x^{\star}\rangle -\frac{1}{2}\|y^{t+1}-x^t\|^2\\
&=\frac{1}{2}\|x^t-x^{\star}\|^2-\eta\langle  G_{x^{\star}}(y^{t+1}), y^{t+1}-x^{\star} \rangle\\
&\quad -\eta\langle   G_{x^{\star}}(x^{t})-G_{x^{\star}}(y^{t+1}), y^{t+1}-x^{\star}\rangle-\frac{1}{2}\|y^{t+1}-x^t\|^2.
\end{aligned}
\end{equation}
Strong convexity of the loss functions ensures
\begin{equation}\label{eqn:contract}
\eta\langle G_{x^{\star}}(y^{t+1}), y^{t+1}-x^{\star}\rangle\geq \eta\langle G_{x^{\star}}(y^{t+1})-G_{x^{\star}}(x^{\star}), y^{t+1}-x^{\star}\rangle\geq \alpha\eta \|y^{t+1}-x^{\star}\|^2.
\end{equation}
Young's inequality in turn implies
\begin{equation}\label{eqn:lastone}
\begin{aligned}
\eta |\langle G_{x^{\star}}(x^{t})-G_{x^{\star}}(y^{t+1}), y^{t+1}-x^{\star} \rangle| &\leq  \frac{\|G_{x^{\star}}(x^{t})-G_{x^{\star}}(y^{t+1})\|^2}{2L^2}+ \frac{\eta^2 L^2\|y^{t+1}-x^{\star} \|^2}{2}\\
&\leq \frac{\|x^t-y^{t+1}\|^2}{2}+\frac{\eta^2 L^2 \|y^{t+1}-x^{\star}\|^2}{2},
\end{aligned}
\end{equation}
where the last inequality follows from Assumption~\ref{assump:smoothnessL}.
Putting the estimates \eqref{eqn:young}-\eqref{eqn:lastone} together yields
$$\frac{1}{2}\|y^{t+1}-x^{\star}\|^2 \leq \frac{1}{2}\|x^{t}-x^{\star}\|^2-\frac{2\alpha\eta-\eta^2L^2}{2}\|y^{t+1}-x^{\star}\|^2.$$
Rearranging gives
$\|y^{t+1}-x^{\star}\|^2\leq \tfrac{1}{1+2\alpha\eta-\eta^2L^2}\|x^{t}-x^{\star}\|^2.$
Returning to \eqref{eqn:start_here} we therefore conclude
$$\|x^{t+1}-x^{\star}\|\leq \left(\frac{1}{\sqrt{1+2\alpha\eta-\eta^2L^2}}+\eta\sqrt{\sum_i \beta_i^2\gamma_i^2}\right)\|x^t-x^{\star}\|.$$
Plugging in $\eta=\frac{\alpha}{L^2}$  yields the claimed estimate \eqref{eqn:grad_discnt_dcrs}. An elementary argument shows that in the assumed regime $\rho<\frac{1}{2+\sqrt{2}}$, the term in the parentheses in \eqref{eqn:grad_discnt_dcrs} is indeed smaller than one.
\end{proof}

\subsection{Repeated Stochastic  Gradient Method}
\label{sec:stochastic_RGP}
As observed earlier, the repeated gradient method is still largely a conceptual
algorithm since an expectation has to be computed in every iteration. We next analyze an implementable algorithm that approximates the expectation in each step of gradient with an unbiased estimator. Namely, in each iteration $t$ of the
{\em  repeated stochastic gradient method}, each player $i\in [n]$ performs the update:
\begin{equation*}\left\{
\begin{aligned}
&{\rm Sample }~z^t_i\sim \mathcal{D}_i(x^t) \\
&{\rm Set } ~x^{t+1}_{i}=\proj_{\X_i} \left(x^t_i-\eta \nabla_i \ell_i(x_i^t,x_{-i}^t, z_i^t)\right)
\end{aligned}\right\}.
\end{equation*}
We will analyze the method under the following standard variance assumption.

\begin{assumption}[Finite variance]\label{assump_fin_var}
There exists a constant  $\sigma\geq 0$ satisfying
 $$\E_{z\sim \mathcal{D}_{\pi}(x)}\|G_{x}(x)-g(x,z)\|^2\leq \sigma^2\qquad \textrm{for all }x\in \X.$$
 \end{assumption}

Convergence analysis for the repeated stochastic gradient method will follow from the
following simple observation. Noticing the equality $G_x(x)=\E_{z\sim
    \mathcal{D}_{\pi}(x)} g(x,z)$, Lemma~\ref{lem:dev_game_jac} directly implies
    that
$$\|G_x(x)-G_{x^{\star}}(x)\|\leq \alpha\rho \|x-x^{\star}\|.$$
That is, we may interpret the repeated stochastic gradient method as a standard stochastic gradient algorithm applied to the static problem $\mathcal{G}(x^{\star})$ with a bias that is linearly bounded by the distance to  $x^{\star}$. With this realization, we can forget about dynamics and simply analyze the stochastic gradient method on a static problem with this special form of bias. Appendix~\ref{sec:append:generic_res} does exactly that. In particular, the following is a direct consequence of Theorem~\ref{thm:sgd} in Appendix~\ref{sec:append:generic_res}.

In the following theorem, let $\mb{E}_t$ denote the conditional expectation with respect to the $\sigma$-algebra generated by $(x^l)_{l=1,\ldots,t}$.
\begin{theorem}[One-step improvement]\label{thm:one_step_sgd}
Suppose that Assumptions~\ref{assump:convex}-\ref{assump_fin_var} hold and that we are in the regime $\rho<1$. Then with any $\eta<\frac{\alpha(1-\rho)}{8 L^2}$, the repeated stochastic gradient method generates a sequence $x^t$ satisfying
$$\mathbb{E}_t\|x^{t+1}-x^{\star}\|^2\leq \frac{1+2\eta\alpha \rho +2\eta^2\alpha^2 \rho^2}{1+2\eta\alpha(\frac{1+\rho}{2})}\|x^t-x^{\star}\|^2+\frac{4\eta^2\sigma^2}{1+2\eta\alpha(\frac{1+\rho}{2})},$$
where $x^{\star}$ is the %
performatively stable equilibrium of the game
\eqref{eqn:main_problem}.
\end{theorem}
\begin{proof}
This follows directly from applying Theorem~\ref{thm:sgd} in Appendix~\ref{sec:append:generic_res} with $G(x)=G_{x^{\star}}(x)$, $g^t=g(x^t,z^t)$, $C_t=D=0$, and $B=\alpha\rho$.
\end{proof}

With Theorem~\ref{thm:one_step_sgd} at hand, it is straightforward to obtain global efficiency estimates under a variety of step-size choices. One particular choice, highlighted by \cite{ghadimi2013optimal}, is the step-decay schedule that periodically cuts $\eta$ by a fraction. The resulting algorithm and its convergence guarantees are summarized in the following corollary.

\begin{corollary}[Repeated stochastic gradient method with a step-decay schedule]\label{cor:main_cor_sgd}
Suppose that Assumptions~\ref{assump:convex}-\ref{assump_fin_var} hold and we are in the regime $\rho<\frac{1}{2}$. Set $\eta_0:=\frac{\alpha(1-\rho)}{4}\cdot\min\{1,\frac{1}{2L^2} \}$. 
Consider running the repeated stochastic gradient method in $k=0,\ldots, K$ epochs,
for $T_k$ iterations each, with constant step-size $\eta_k=2^{-k}\eta_0$, and
such that the last iterate of epoch $k$ is used as the first iterate in epoch
$k+1$. Fix a target accuracy $\varepsilon>0$ and suppose we have available a
constant $R\geq \|x^0-x^{\star}\|^2$. Set
$$T_0=\left\lceil\tfrac{10}{(1-\rho)\alpha\eta_0}\log(\tfrac{2R}{\varepsilon})\right\rceil, ~~T_k=\left\lceil\tfrac{10\log(4)}{(1-\rho) \alpha\eta_k} \right\rceil~~\textrm{ for }~~k\geq 1, \qquad\textrm{and}\qquad K=\left\lceil1+\log_2\left(\tfrac{40\eta_0\sigma^2}{(1-\rho)\alpha\varepsilon}\right)\right\rceil.$$
The final iterate $x$ produced satisfies $\E\|x-x^\star\|^2\leq \varepsilon$, while
the total number of iterations of the repeated stochastic gradient method is at most
$$\mathcal{O}\left(\frac{L^2}{(1-\rho)\alpha^2}\cdot\log\left(\frac{2R}{\varepsilon}\right)+\frac{\sigma^2}{(1-\rho)^2\alpha^2\varepsilon}\right).$$
\end{corollary}
\begin{proof}
Consider a sequence $x^0, x^1,\ldots, x^t$ generated by the stochastic gradient method
with a fixed step-size $\eta\leq \eta_0$. Using Theorem~\ref{thm:one_step_sgd} together with the tower rule for conditional expectations, we deduce
\begin{equation}\label{eqn:one_step_equivaleb}
\mathbb{E}\|x^{t+1}-x^{\star}\|^2\leq \frac{1+2\eta\alpha \rho +2\eta^2\alpha^2 \rho^2}{1+2\eta\alpha(\frac{1+\rho}{2})}\E\|x^t-x^{\star}\|^2+\frac{4\eta^2\sigma^2}{1+2\eta\alpha(\frac{1+\rho}{2})}.
\end{equation}
Our choice of $\eta_0$ ensures
$$\frac{1+2\eta\alpha \rho +2\eta^2\alpha^2
    \rho^2}{1+2\eta\alpha(\frac{1+\rho}{2})}\leq
    \frac{1+2\eta\alpha\cdot\frac{1+3\rho}{4}}{1+2\eta\alpha\cdot\frac{1+\rho}{2}}=1-\frac{2\eta\alpha(\frac{1-\rho}{4})}{1+2\eta\alpha(\frac{1+\rho}{2})}\leq
    1-\tfrac{1-\rho}{10}\eta\alpha.$$
Therefore iterating \eqref{eqn:one_step_equivaleb} we obtain the estimate
$$\mathbb{E}\|x_{t+1}-x^{\star}\|^2\leq \left(1-\psi(\eta)\right)^{t+1}\|x_0-x^{\star}\|^2+\Gamma\eta,$$
where we set $\psi(\eta)=c\eta$ with $c=\tfrac{1-\rho}{10}\alpha$ and $\Gamma=\frac{40\sigma^2}{\alpha(1-\rho)}$.
 The result now follows directly from \cite[Lemma B.2]{drusvyatskiy2020stochastic}.
\end{proof}

The efficiency estimate in Corollary~\ref{cor:main_cor_sgd} coincides with the standard efficiency estimate for the stochastic gradient method on static problems, up to multiplication by $(1-\rho)^{-2}$.

\section{Monotonicity of Decision-Dependent Games}\label{sec:monotonicity}
Our next goal is to develop algorithms for finding true Nash equilibria of the
 game \eqref{eqn:main_problem}. As the first step, this section presents
 sufficient conditions for the game to be monotone along with some examples. We
 note, however, that the sufficient conditions we present are strong, and
 necessarily so because the game  \eqref{eqn:main_problem} is typically not
 monotone. When specialized to the single player setting $n=1$, the sufficient
 conditions we derive are identical to those in \cite{miller2021outside} although the proofs are entirely different.

First, we impose the following mild smoothness assumption.

\begin{assumption}[Smoothness of the distribution]\label{assump:smooth_distr}
For each index $i\in [n]$ and point $x\in \X$, the map $u_i\mapsto\E_{z_i\sim \mathcal{D}(u_i,x_{-i})} \ell_i(x,z_i)$ is differentiable at $u_i=x_i$ and its derivative is continuous in $x$.
\end{assumption}

Under Assumption~\ref{assump:smooth_distr}, the chain rule implies the following expression for the derivative of player $i$'s loss function
$$\nabla_i  \LL_i(x_i,x_{-i})= \frac{d}{du_i}\E_{z_i\sim\mc{D}_i(x_i,x_{-i})}[\ell_i(u_i,x_{-i},z_i)]\Big|_{u_i=x_i}+\frac{d}{du_i}\E_{z_i\sim\mc{D}_i(u_i,x_{-i})}[\ell_i(x_i,x_{-i},z_i)]\Big|_{u_i=x_i}.$$
To simplify notation, define
\[H_{i,x}(y):=\frac{d}{du_i}\E_{z_i\sim \mathcal{D}(u_i,x_{-i})} \ell_i(y,z_i)\Big|_{u_i=x_i}.\]
Therefore, we may equivalently write
$$\nabla_i  \LL_i(x_i,x_{-i})= G_{i,x}(x)+H_{i,x}(x)$$
where $G_{i,x}(x)$ is defined in \eqref{eqn:game_jacobian}. Stacking together the individual partial gradients $H_{i,x}(y)$ for each player,
we set $H_x(y)=(H_{1,x}(y),\ldots,H_{n,x}(y))$. Therefore the vector of
individual gradients for \eqref{eqn:main_problem} is simply the map $D(x):=G_x(x)+H_x(x).$ Thus the game  \eqref{eqn:main_problem} is monotone, as long as $D(x)$ is a monotone mapping.

The sufficient conditions we present in Theorem~\ref{thm:monotone} are simply that we are in the regime $\rho<\frac{1}{2}$ and that the map $x\mapsto H_x(y)$ is monotone for any $y$. The latter can be understood as requiring that for any $y\in \X$, the auxiliary game wherein each player aims to solve
\[\min_{x_i\in \X} ~\E_{z_i\sim \mc{D}_i(x_i,x_{-i})} \ell_i(y,z_i).\]
is monotone. In the single player setting $n=1$, this simply means that the
function $x\mapsto\E_{z_i\sim \mc{D}_i(x)} \ell(y,z_i)$ is convex for any fixed $y\in \X$, thereby reducing exactly to the requirement in \cite[Theorem 3.1]{miller2021outside}.

The proof of Theorem~\ref{thm:monotone} crucially relies on the following useful lemma.

\begin{lemma}
Suppose that Assumptions~\ref{assump:convex}, \ref{assum:perm_pred}, \ref{assump:smooth_distr} hold.
 For any $x\in \X$, the map $H_x(y)$ is Lipschitz continuous in $y$ with parameter $\sqrt{\sum_{i=1}^n\beta_i^2\gamma_i^2}$.
       \label{lem:Hyx_Lipschitz}
\end{lemma}
\begin{proof}
Fix three points $x,x',y\in \X$. Player $i$'s coordinate of $H_{x'}(x)-H_{x'}(y)$ is simply
$$H_{i,x'}(x)-H_{i,x'}(y)=\frac{d}{du_i} \E_{z_i\sim \mathcal{D}_i(u_i,x_{-i}')}\left(\ell_i(x,z_i)- \ell_i(y,z_i)\right)\Big|_{u_i=x_i'}.$$
 Setting $\gamma(s)=y+s(x-y)$ for any $s\in (0,1)$, the fundamental theorem of calculus ensures
\[
\ell_i(x,z_i)-\ell_i(y,z_i)=\int_{s=0}^1\langle \nabla_i
\ell_i(\gamma(s),z_i),x-y\rangle~ ds.
\]
Therefore differentiating, taking an expectation, and using the Cauchy-Schwarz inequality we deduce
\begin{equation}\label{eqn:intergralmean}
\|H_{i,x'}(x)-H_{i,x'}(y)\|\leq \int_{s=0}^1\left\|\frac{d}{du_i} \E_{z_i\sim
    \mathcal{D}_i(u_i,x_{-i}')}\nabla_i \ell_i(\gamma(s),z_i)\Big|_{u_i=x_i'} \right\|\cdot \|x-y\|~ ds.
\end{equation}
Now for any $s\in (0,1)$, Lemma~\ref{lem:dev_game_jac} guarantees that the map 
$u_i\mapsto \E_{z_i\sim \mathcal{D}_i(u_i,x_{-i}')}\nabla_i \ell_i(\gamma(s),z_i)$ is Lipschitz continuous with parameter $\beta_i\gamma_i$ and therefore its derivative is upper-bounded in norm by $\beta_i\gamma_i$.
We therefore deduce that the right hand side of \eqref{eqn:intergralmean} is
upper bounded by $\beta_i\gamma_i\|x-y\|$. Applying this argument to each
player leads to the claimed Lipschitz constant on $H_x(y)$ 
with respect to $x$.
\end{proof}

\begin{theorem}[Monotonicity of the decision-dependent game]\label{thm:monotone}
  Suppose that Assumptions~\ref{assump:convex}, \ref{assum:perm_pred}, and
  \ref{assump:smooth_distr} hold, and that we are in the regime $\rho<\frac{1}{2}$ and
    the map $x\mapsto H_x(y)$ is monotone for any $y$. The game
     \eqref{eqn:main_problem} is strongly monotone with parameter $\left(1-2\rho\right)\alpha$.
\end{theorem}
\begin{proof}
    Fix an arbitrary pair of points $x, x'\in \X$. Expanding the following
    inner product, we have
    \begin{align}
\langle D(x)-D(x'),x-x'\rangle=\langle G_x(x)-G_{x'}(x'),x-x'\rangle +\langle H_x(x)-H_{x'}(x'),x-x'\rangle.\label{eqn:key_decomp_inner}
\end{align}
We estimate the first term as follows:
\begin{align}
\langle G_x(x)-G_{x'}(x'),x-x'\rangle&=\langle G_{x'}(x)-G_{x'}(x'),x-x'\rangle+\langle G_x(x)-G_{x'}(x),x-x'\rangle\notag\\
&\geq \alpha\|x-x'\|^2-\left(\sum_{i=1}^n\beta_i^2\gamma_i^2\right)^{1/2}\cdot \|x-x'\|^2\label{line:comput1}\\
&=\left(1-\rho\right)\alpha \|x-x'\|^2,\label{eqn:last}
\end{align}
where \eqref{line:comput1} follows from Assumption~\ref{assump:convex} and Lemma~\ref{lem:dev_game_jac}.
Next, we estimate the second term on the right side of \eqref{eqn:key_decomp_inner} as follows:
\begin{align}
\langle H_x(x)-H_{x'}(x'),x-x'\rangle&=\langle H_{x'}(x)-H_{x'}(x'),x-x'\rangle+\langle H_x(x)-H_{x'}(x),x-x'\rangle\notag\\
&\geq \langle H_{x'}(x)-H_{x'}(x'),x-x'\rangle\label{eqn:mon_proof}\\
&\geq -\|H_{x'}(x)-H_{x'}(x')\|\cdot \|x-x'\|\\
&\geq -\left(\sum_{i=1}^n\beta_i^2\gamma_i^2\right)^{1/2}\|x-x'\|^2,\label{eqn:key_step}
\end{align}
where \eqref{eqn:mon_proof} follows from the assumption that the map $x\mapsto H_x(y)$ is monotone and \eqref{eqn:key_step} follows from Lemma~\ref{lem:Hyx_Lipschitz}.
Combining \eqref{eqn:key_decomp_inner}, \eqref{eqn:last}, and \eqref{eqn:key_step} completes the proof.
\end{proof}

The following two examples of multiplayer performative prediction problems illustrate settings where the mapping $x\mapsto H_x(y)$ is indeed monotone and therefore Theorem~\ref{thm:monotone} may be applied to deduce monotonicity of the game.
\begin{example}[Revenue Maximization in Ride-Share Markets]\label{example:revenue_max}
{\rm
    Consider a ride share market with two firms that
each would like to maximize their revenue by setting the price $x_i$. The demand $z_i$ that
each ride share firm sees is influenced not only by the price they set but also
the price that their competitor sets. Suppose that firm $i$'s loss is given by
\[\ell_i(x_i,z_i)=-z_i^{\top}x_i+\frac{\lambda_i}{2}\|x_i\|^2\]
where $\lambda_i\geq 0$ is some
regularization parameter. Moreover, let us suppose  that the random demand $z_i$
takes the semi-parametric form
$z_i=\bz_{i}+A_ix_i+A_{-i}x_{-i}$, where $\bz_{i}$ follows some base
distribution $\PP_{i}$ and
the parameters $A_i$ and $A_{-i}$ capture price
elasticities to player $i$'s and its competitor's change in price, respectively.
The mapping $x\mapsto H_x(y)$ is monotone. Indeed, observe that $i$-th component of $H_x(y)$ is given by
\[H_{i,x}(y)=\E_{\bz_{i}\sim \PP_{i}}A_{i}^\top
    \nabla_{z_i}\ell_i(y_i,\bz_{i}+A_ix_i+A_{-i}x_{-i})=-A^\top_{i}y_i.\]
Hence, the map $x\mapsto H_x(y)$ is constant and is therefore trivially monotone.
}
\end{example}

The next example is a multiplayer extension of Example 3.2 in
\cite{miller2021outside} which models the single player decision-dependent
problem of predicting the final vote margin in an election contest.

\begin{example}[Strategic Prediction]\label{ex:strategic_prediction} %
    {\rm Consider
two election prediction platforms. Each platform seeks to predict the vote
margin. Not only can predicting a large margin in either direction dissuade
people from voting, but people may look at multiple  platforms  as a source for
information. Features $\theta$ such as past polling averages are drawn i.i.d.~from a
static distribution $\PP_\theta$. Each platform observes a sample drawn from the conditional distribution
\[z_i|\theta\sim \varphi_i(\theta)+A_ix_i+A_{-i}x_{-i}+w_i,\]
where $\varphi_i:\mb{R}^{d_i}\to \mb{R}$ is an arbitrary map, the parameter matrices
$A_i\in \mb{R}^{m_i\times d_i}$ and $A_{-i}\in \mb{R}^{m_i\times d_{-i}}$ are
fixed, and $w_i$ is a zero-mean random variable.  Each player seeks to predict $z_i$ by minimizing the loss
\[\ell_i(x_i,z_i)=\frac{1}{2}\|z_i-\theta^\top x_i\|^2.\]
We claim that  the map $x\mapsto H_x(y)$ is monotone as long as
$$\sqrt{n-1}\cdot\max_{i\in [n]}\|A_{-i}^\top A_i\|_{\rm op}\leq \min_{i\in [n]}\lambda_{\rm{min}}(A_i^{\top}A_i),$$
where $\lambda_{\rm{min}}$ denotes the minimal eigenvalue.
 The
interpretation of this condition is that the performative effects due to
interaction with competitors are small relative to any player's own performative
effects. To see
the claim, set $\bar{A}_i$ to be the
matrix satisfying $\bar{A}_ix=A_ix_i+A_{-i}x_{-i}$ and observe that the $i$-th component of $H_x(y)$ is given by
\begin{align*}
H_{i,x}(y)&=\E_{\theta, w_i}A_{i}^\top
    \nabla_{z_i}\ell_i(y_i,\varphi_i(\theta)+\bar{A}_ix+w_i)\\
    &=\E_{\theta,
    w_i}A^\top_{i}(\varphi_i(\theta)+\bar{A}_ix-\theta^\top y_i+w_i)\\
    &=A_i^{\top}A_i x_i+A_i^{\top}A_{-i} x_{-i}+\E_{\theta,
    w_i}A^\top_{i}(\varphi_i(\theta)-\theta^\top y_i+w_i).
    \end{align*}
Therefore, the map  $H_{i,x}(y)$ is affine in $x$. Consequently, monotonicity of $x\mapsto H_{i,x}(y)$ is equivalent to monotonicity of  the linear map
$x\mapsto V(x)+W(x)$,
where $V$ is the block diagonal matrix $V(x)=\textrm{Diag}(A_1^{\top}A_1,\ldots, A_n^{\top}A_n)x$ and we define the linear map $W(x)=(A_1^{\top}A_{-1}x_{-1}, \ldots,A_n^{\top}A_{-n}x_{-n})$.
The minimal eigenvalue of $V$ is simply $\min_{i\in [n]}\lambda_{\min}( A_i^{\top}A_i)$. Let us estimate the operator norm of $W$. To this end, set $s:=\max_{i}\|A_i^{\top}A_{-i}\|_{\rm op}$ and for any $x$ we compute
$$\|W(x)\|^2=\sum_{i=1}^n \|A_i^{\top}A_{-i}x_{-i}\|^2\leq \sum_{i=1}^n s^2\|x_{-i}\|^2= (n-1) s^2\|x\|^2.$$
Thus, under the stated assumptions, the operator norm of $W$ is smaller than the minimal eigenvalue of $V$, and therefore the sum $V+W$ is monotone.
 }
\end{example}

\section{Algorithms for Finding Nash Equilibria}
\label{sec:algorithms_nash}

In contrast to Section~\ref{sec:algorithms}, in this section we analyze algorithms that
converge to the Nash equilibrium of the $n$--player performative prediction game \eqref{eqn:main_problem} when the game is strongly monotone. Recall
that the Nash equilibrium
    $x^\star$ of this game is characterized by the relation
\[x^\star_i\in\amin_{x_i\in
\X_i}\E_{z_i\sim \mc{D}_i(x_i,x_{-i}^\star)}\ell_i(x_i,x_{-i}^\star,z_i)\qquad \forall i\in [n].\]
It is important to stress the distinction between the
performatively stable equilibria studied in Section~\ref{sec:algorithms} and the Nash equilibrium
$x^\star$ of the game \eqref{eqn:main_problem}: namely, for the latter concept the distribution
$\mc{D}_i$ explicitly depends on the optimization variable $x_i$ versus being
fixed at $x^\star=(x_i^\star,x_{-i}^\star)$.
Theorem~\ref{thm:monotone} (cf.~Section~\ref{sec:monotonicity}) gives  sufficient conditions under which the multiplayer
performative prediction game  \eqref{eqn:main_problem} is strongly monotone and hence, admits a unique Nash equilibrium.

In the following subsections, we
study natural learning dynamics---namely, variants of \emph{gradient play} as it
is referred to in the literature on learning and games---for continuous games seeking Nash equilibrium in different information settings. Specifically, we study
gradient-based learning methods where players update using an estimate of
their individual gradient
 consistent with the information available to them. It is important to contrast
 the gradient updates in Section~\ref{sec:algorithms} with the updates considered in this
 section: the Nash-seeking algorithms studied in this section all use gradient
 estimates of the individual gradient
 \begin{equation}
    \nabla_i  \LL_i(x_i,x_{-i})= \E_{z_i\sim\mc{D}_i(x)}[\nabla_i\ell_i(x,z_i)]+\frac{d}{du_i}\E_{z_i\sim\mc{D}_i(u_i,x_{-i})}[\ell_i(x,z_i)]\Big|_{u_i=x_i}
     \label{eq:individual_gradient}
 \end{equation}
 for each player $i\in[n]$, whereas performatively stable equilibrium seeking
 algorithms of Section~\ref{sec:algorithms} are defined such that the gradient
 update only uses the first term on the right hand side of
 \eqref{eq:individual_gradient}.

The main difficulty with applying gradient-based methods is that estimation of the second term on the right hand side of
 \eqref{eq:individual_gradient}, without some parametric assumptions on the distributions $\mc{D}_i$. Consequently, we start in Section~\ref{sec:dfo} with derivative free methods, wherein
 each player only has access to loss
 function queries. This does not require players to have any information on the
 distribution $\mc{D}_i$, but results in a slow algorithm, roughly with
 complexity $\mc{O}(\frac{d^2}{\varepsilon^2})$. In practice, the players may have some
 information on $\mc{D}_i$, and it's reasonable that they would exploit this
 information during learning. Hence, in Section~\ref{sec:stochastic_gradplay_pp}
 we impose a specific parametric assumption of the distributions and study
 \emph{stochastic gradient play}\footnote{This method is known as stochastic
     gradient play in the game theory literature; here we refer to as the \emph{stochastic gradient method} to be consistent with
 the naming convention of other methods in the paper.} under the assumption that each player knows their own
 ``influence'' parameter on the distribution. The resulting algorithm enjoys
 efficiency on the order of $\mc{O}(\frac{1}{\varepsilon})$.
 Section~\ref{sec:adapt_grad_play} instead develops a variant of a stochastic gradient method
 wherein each player adaptively learns their influence parameters, and uses
 their current estimate of those parameters to
 optimizing their loss function by taking a step along the direction of their
 individual gradient; the resulting process has efficiency on the
 order of $\mc{O}(\frac{d}{\varepsilon})$.

 \subsection{Derivative Free  Method for Performative Prediction Games }
 \label{sec:dfo}
As just alluded to, the first information setting we consider for multiplayer performative
prediction is the ``bandit feedback'' setting, where players have oracle
access to queries of their loss function only, and therefore are faced with the
problem of  creating an estimate
of their gradient from such queries.
This setting requires the least assumptions on what information
is available to players. In the optimization literature, when a first order oracle is not available, derivative free or zeroth order methods are typically applied.
Derivative free methods have been extended to games 
\citep{bravo2018bandit,drusvyatskiy2021improved}. The results in this section are direct consequences of the results in these papers. We concisely spell them out here in order to compare them with the convergence guarantees discussed in the following two sections.

 The {\em derivative free (gradient) method} we consider proceeds as follows. Fix a parameter $\delta>0$.
 In each iteration $t$, each player $i\in [n]$ performs the update:
\begin{equation}\label{eqn:DFO}\left\{
\begin{aligned}
    &{\rm Sample}~v_i^t\in \mb{S}_i\\
    &{\rm Sample }~z^t_{i}\sim \mathcal{D}_{i}(x^t+\delta v^t) \\
    &{\rm Set } ~x^{t+1}_{i}=\proj_{(1-\query)\X_i}\left(x_{i}^t-\eta_t \frac{d_i}{\query}\ell_i(x^t+\query
v^t,z_i^t)v_{i}^t\right)
\end{aligned}\right\}.
\end{equation}
Recall that $\mb{S}_i$ denotes the unit sphere with dimension $d_i$. The reason for
 projecting onto the set $(1-\delta)\X_i$ is simply to ensure that in the next
 iteration $t+1$, the strategy played by player $i$  remains in $\X_i$. We state
 the convergence guarantees of the method informally here because they are meant
 only as a baseline result. The formal statement for derivative free methods in general games can be found in \citet{drusvyatskiy2021improved}.\footnote{Though Theorem 2 in \citet{drusvyatskiy2021improved} is stated for deterministic games, it applies verbatim whenever the value of the loss function for each each player is replaced by an unbiased estimator of their individual loss functions---our setting.}

\begin{proposition}[Convergence rate of the derivative free method]\label{corzeroordernew} Consider an $n$--player
    decision-dependent game \eqref{eqn:main_problem}. Under reasonable smoothness and bounded variance assumptions, algorithm \eqref{eqn:DFO} with appropriately chosen parameters $\delta$ and $\eta_t$ will find a point $x$ satisfying $\mathbb{E}[\|x-x^{\star}\|^2]\leq \varepsilon$ after at most $O(\frac{d^2}{\varepsilon^2})$ iterations.
\end{proposition}

The rate $O(\frac{d^2}{\varepsilon^2})$ can be extremely slow in practice.
Therefore  in the rest of the paper, we focus on stochastic gradient based methods, which enjoy significantly better efficiency guarantees (at cost of access to a richer oracle).

\subsection{Stochastic Gradient Method in Performative Prediction Games}
\label{sec:stochastic_gradplay_pp}
In practice, often players have some information regarding their
data distribution $\mc{D}_i$ and can leverage this during learning. Stochastic
gradient play---which we refer to as the stochastic gradient method to be
consistent with the rest of the paper---is
a natural learning algorithm commonly adopted in the literature on learning in
games for settings where players have an unbiased
estimate of their individual gradient. To apply the stochastic gradient method to multiplayer
performative prediction, players need oracle access to the gradient of their
loss with respect to their choice variable, which requires some knowledge of how the distribution $\mc{D}_i$ depends on
the joint action profile $x$. To this end, let us impose the following parametric assumption, which we have already encountered in Example~\ref{example:revenue_max}.

 \begin{assumption}[Parametric assumption]\label{assum:param}
    For each index $i\in[n]$, there exists a probability measure $\PP_{i}$ and matrices $A_i$ and $A_{-i}$ satisfying
  \[z_i\sim \mc{D}_i(x) \qquad \Longleftrightarrow\qquad
  z_i=\bz_{i}+A_ix_i+A_{-i}x_{-i} \quad \textrm{ for } \bz_{i}\sim \PP_{i}.\]
 The mean and covariance of $\bz_{i}$ are defined as $\mu_i:=\E_{\bz_{i}\sim
 \PP_{i}}[\bz_{i}]$ and $\Sigma_i:=\E_{\bz_{i}\sim
 \PP_{i}}[(\bz_{i}-\mu_i)(\bz_{i}-\mu_i)^{\top}]$, respectively.  
 \end{assumption}

 Assumption~\ref{assum:param} is very natural and generalizes an
 analogous assumption used in the single player setting in
 \cite{miller2021outside}. It asserts that the distribution used by player $i$
 is a ``linear perturbation'' of some base distribution $\PP_{i}$. We can interpret the matrices $A_i$ and $A_{-i}$ as quantifying the performative effects of player $i$'s decisions and the rest of the players' decisions, respectively, on the  distribution $\mathcal{D}_i$ governing player $i$'s data.

Under Assumption~\ref{assum:param}, we may write player $i$-th  loss function as
\begin{equation}\label{eqn:biglos_madeez}
\LL_i(x)=\E_{\bz_{i}\sim \PP_{i}} \ell_i(x,\bz_{i}+A_ix_i+A_{-i}x_{-i}).
\end{equation}
Under mild smoothness assumptions, differentiating \eqref{eqn:biglos_madeez} using the chain rule, we see that the gradient of the
$i$-th player's loss is simply
\begin{equation}
    \nabla_i  \LL_i(x)=
    \E_{z_{i}\sim\mc{D}_{i}(x)}[\nabla_i\ell_i(x,z_i)+A_i^\top
    \nabla_{z_i}\ell_i(x,z_i)].   \label{eq:individual_gradient_parametric}
 \end{equation}
Therefore, given a point $x$, player $i$ may draw $z_i\sim \mc{D}_{i}(x)$ and form the vector
$$w_i(x,z_i)=\nabla_i\ell_i(x,z_i)+A_i^\top
    \nabla_{z_i}\ell_i(x,z_{i}).$$
By definition, $w_i(x,z_i)$ is an unbiased estimator of $\nabla_i \LL_i(x)$, that is
$$\E_{z_i\sim\mathcal{D}_i(x)}w_i(x,z)=\nabla_i \LL_i(x).$$
With this notation, the stochastic gradient method proceeds as follows: in each iteration $t\geq 0$ each player $i\in [n]$ performs the update:
\begin{equation}\label{eqn:performSGD}\left\{
\begin{aligned}
    &{\rm Sample }~z^t_{i}\sim \mathcal{D}_i(x^t) \\
&{\rm Set }~x^{t+1}_{i}=\proj_{\X_i} \left(x^t_i-\eta_t \cdot w_i(x^t,z^t_i) \right)
\end{aligned}\right\}.
\end{equation}

Let us look at the computation that is required in each iteration. Evaluation of
the vector $w_i(x,z_i)$ requires evaluation of both $\nabla_i \ell_i(x,z_i)$ and
$\nabla_{z_i} \ell_i(x,z_i)$, and knowledge of the matrix $A_i$. When the game
is separable, it is very reasonable that each player can explicitly compute
$\nabla_i \ell_i(x_i,z_i)$ and $\nabla_{z_i} \ell_i(x_i,z_i)$ assuming oracle
access to queries $z_i$ from the environment which does depend on $x_{-i}$
and $x_i$.
Moreover, the matrix $A_i$ depends only on the performative effects of player
$i$, and in this section we will suppose that it is indeed known to each player.
In the next section, we will develop an adaptive algorithm wherein each player
$i\in [n]$ simultaneously learns $A_i$ and $A_{-i}$ while optimizing their loss.

In order to apply standard convergence guarantees for stochastic gradient play,
we need to assume that $(i)$ the vector of individual gradients is Lipschitz continuous and (ii) that the variance of $w(x,z_i)$ is bounded.
Let us begin with the former.

 \begin{assumption}[Smoothness]\label{assum:smoothness_elli}
 Suppose that the map $(\nabla_1\LL_1(x),\nabla_2 \LL_2(x),\ldots, \nabla_n \LL_n(x))$ is $L$-Lipschitz continuous.
 \end{assumption}

 The constant $L$ may be easily estimated from the smoothness parameters of each
 individual loss function $\ell_i(x,z)$ and the magnitude of the matrices $A_i$
 and $A_{-i}$; this is the content of the following lemma. In what follows, we
 define the mixed partial derivative
 $\nabla_{i,z_i}\ell_i(x,z_i)=(\nabla_{i}\ell_i(x,z_i),\nabla_{z_i}\ell_i(x,z_i))$. Recall that $\nabla_i\ell_i(x_i,x_{-i},z_i)$ denotes the partial derivative of $\ell_i$
 with respect to the $x_i$ argument and $\nabla_{z_i}\ell_i(x_i,x_{-i},z_i)$
 denotes the partial derivative with respect to $z_i$.

\begin{lemma}[Sufficient conditions for Assumption~\ref{assum:smoothness_elli}.]
 Suppose that Assumption~\ref{assum:param} holds and that there exist constants $\xi_i\geq 0$ such that for each index $i$ the map
 $(x,z_i)\mapsto \nabla_{i,z_i}\ell_i(x,z_i)$ is $\xi_i$-Lipschitz continuous. Then Assumption~\ref{assum:smoothness_elli} holds with
 $$L=\sqrt{\sum_{i=1}^n \xi_i^2\max\{1,\|A_i\|_{\rm op}^2\}\cdot (1+\|\bar A_i\|_{\rm op}^2)}.$$
\end{lemma}
\begin{proof}
Let $\bar A_i$ be a matrix satisfying $\bar A_i x=A_i x_i+A_{-i}x_{-i}$.
Observe
that we may write
$$\nabla_i \LL_i(x)=\E_{\bz_{i,0}\sim
\PP_{i}}V^{\top}\nabla_{i,z_i}\ell_i(x,\bz_{i}+\bar A_i x)\qquad  \textrm{where}\qquad V=\begin{bmatrix} I &0\\
0 &A_i
\end{bmatrix}.$$
Therefore, we deduce
\begin{align*}
\|\nabla_i\LL_i(x)-\nabla_i \LL_i(x')\|&\leq \|V\|_{\rm op}\E_{\bz_{i}\sim
\PP_{i}}\|\nabla_{i,z_i}\ell_i(x,\bz_{i}+\bar A_i
x)-\nabla_{i,z_i}\ell_i(x',\bz_{i}+\bar A_i x') \|\\
&\leq \max\{1,\|A_i\|_{\rm op}\}\cdot \xi_i\cdot \E_{\bz_{i}\sim
\PP_{i}}\|(x,\bz_{i}+\bar A_i x)-(x',\bz_{i}+\bar A_i x')\|\\
&=  \max\{1,\|A_i\|_{\rm op}\}\cdot \xi_i\cdot \sqrt{\|x-x'\|^2+\|\bar A_i(x-x')\|^2}\\
&\leq  \max\{1,\|A_i\|_{\rm op}\}\cdot \xi_i\cdot \sqrt{1+\|\bar A_i\|_{\rm op}^2}\cdot \|x-x'\|.
\end{align*}
This completes the proof.
\end{proof}

Next we assume a finite variance bound.

\begin{assumption}[Finite variance]
Suppose that there exists a constant $\sigma>0$ satisfying
$$\E_{z\sim \mathcal{D}_{\pi}(x)}\|w(x,z)-\E_{z'\sim \mathcal{D}_{\pi}(x)}w(x,z')\|^2\leq \sigma^2 \qquad \forall x\in \X.$$
    \label{a:bddvariance_lip}
\end{assumption}

Let us again present a sufficient condition for Assumption~\ref{a:bddvariance_lip} to hold in in terms of the variance of the individual gradients $\nabla_{i,z_i}\ell(x,z_i)$. The proof is immediate and we omit it.

\begin{lemma}[Sufficient conditions for Assumption~\ref{a:bddvariance_lip}]
Suppose that there exist constants $s_1,s_2\geq 0$ such that for all $x\in \X$ and $i\in[n]$ the estimates hold:
   \begin{align*}
\E_{z_i'\sim\mathcal{D}_i(x)}\|\nabla_i\ell_i(x,z_i')-\E_{z_i\sim\mathcal{D}_i(x)}\nabla_i\ell_i(x,z_i)\|^2 &\leq s_1^2\\
\E_{z_i'\sim\mathcal{D}_i(x)}\|\nabla_{z_i}\ell_i(x,z_i')-\E_{z_i\sim\mathcal{D}_i(x)}\nabla_{z_i}\ell_i(x,z_i)\|^2 &\leq s_2^2
\end{align*}
Then Assumption~\ref{a:bddvariance_lip} holds with $\sigma^2=\sum_{i=1}^n 2(s_1^2+\|A_i\|_{\rm op}^2s_2^2)$.
\end{lemma}

The following is now a direct consequence of standard convergence guarantees for stochastic gradient methods.

\begin{theorem}[Stochastic gradient play]\label{thm:perfsgd_on_step}
    Consider an $n$-player performative prediction game
    \eqref{eqn:main_problem}. Suppose that
    Assumptions~\ref{assum:param}-\ref{a:bddvariance_lip} hold and that the game
    is $\alpha$-strongly monotone with $\alpha>0$.  Then a single step of the stochastic gradient method \eqref{eqn:performSGD}  with any constant $\eta\leq
\frac{\alpha}{2 L^2}$ satisfies
\begin{equation}\label{eqn:recursion_sgd}
\E[\|x^{t+1}-x^\star\|^2]\leq
\frac{1}{1+\alpha\eta}\E[\|x^{t}-x^\star\|^2]+\frac{2\eta^2\sigma^2}{1+\eta\alpha},
\end{equation}
where $x^\star$ is the Nash equilibrium of the game \eqref{eqn:main_problem}.
\label{thm:stochasticgradientplay}
\end{theorem}
\begin{proof}
This is immediate from Theorem~\ref{thm:sgd} in Appendix~\ref{sec:append:generic_res} with $B\equiv C_t\equiv D\equiv 0$.
\end{proof}
Analogous to the analysis of the stochastic repeated gradient method, applying a step-decay schedule on $\eta$ yields the
following corollary. The proof follows directly from the recursion \eqref{eqn:recursion_sgd} and the generic results on step decay schedules; e.g. \cite[Lemma B.2]{drusvyatskiy2020stochastic}.

\begin{restatable}[Stochastic gradient method with a step-decay
    schedule]{corollary}{corstochasticgradplay}\label{cor:main_cor_sgd_pp} Suppose that the assumptions of
    Theorem~\ref{thm:stochasticgradientplay} hold.
 Consider running stochastic gradient method in $k=0,\ldots, K$ epochs, for $T_k$
 iterations each, with constant step-size
 $\eta_k=\frac{\alpha}{2{L}^2}\cdot 2^{-k}$, and such that the last iterate
 of epoch $k$ is used as the first iterate in epoch $k+1$. Fix a target accuracy
 $\varepsilon>0$ and suppose we have available a constant $R\geq \|x^0-x^{\star}\|^2$. Set
\[T_0=\left\lceil\tfrac{2}{\alpha\eta_0}\log(\tfrac{2R}{\varepsilon})\right\rceil,
~~T_k=\left\lceil\tfrac{2\log(4)}{ \alpha\eta_k} \right\rceil~~\textrm{ for
}~~k\geq 1, \qquad\textrm{and}\qquad
K=\left\lceil1+\log_2\left(\tfrac{2\eta_0 \sigma^2}{\alpha\varepsilon}\right)\right\rceil.\]
The final iterate $x$ produced satisfies $\E\|x-x^*\|^2\leq \varepsilon$, while
the total number of iterations of stochastic gradient play  called is at most
\[\mathcal{O}\left(\frac{{L}^2}{\alpha^2}\cdot\log\left(\frac{2R}{\varepsilon}\right)+\frac{\sigma^2}{\alpha^2\varepsilon}\right).\]
\end{restatable}

\subsection{Adaptive Gradient Method in Performative Prediction Games}\label{sec:adapt_grad_play}

Throughout this section, we continue working under the parametric Assumption~\ref{assum:param}.
An apparent deficiency of the stochastic gradient method discussed in Section~\ref{sec:stochastic_gradplay_pp} is that each player $i$ needs to know the matrix $A_i$ that governs the performative effect of the player on the distribution.
In typical settings, the matrix $A_i$ may be unknown to the player, but it might be possible to estimate it from data. Inspired by methods in adaptive control to simultaneously estimate the
parameters of the system and optimize the control input, we propose the
\emph{adaptive gradient method} outlined in
Algorithm~\ref{alg:adaptive_updates}.\footnote{We remark that the word ``adaptive" here refers to adaptively estimating the model parameters, and is different from its meaning in methods like AdaGrad, where it is the algorithm's stepsize that is being adapted.}
In each iteration, each player simultaneously
estimates their distribution
parameters and  myopically optimizes their individual loss via stochastic gradient method on the current
estimated loss. More precisely, the algorithm maintains two sequences: $(i)$
$x^t$ that eventually converges to the Nash equilibrium $x^{\star}$, and $(ii)$
 estimates
$\hat A^t_i$ that dynamically estimates the unknown matrix $\bar A_i$. In each iteration $t$, the algorithm draws samples $z_i^t\sim \mathcal{D}_i(x^t)$, and each player $i$ takes the gradient step
$$x^{t+1}_i =\proj_{\X_i}\left(x_i^t- \eta_t
  ((\nabla_{i}\ell_i(x^t,z_i^t)+(\hat{A}_{ii}^t)^{\top}\nabla_{z_i}\ell_i(x^t,z_i^t))\right),$$
where $\hat{A}_{ii}^t$ denotes the submatrix of $\hat{A}_i^t$ whose columns are indexed by player $i$'s action space. Next, in order to update $\hat{A}^t$, the algorithm draws a sample $q^t_i\sim \mathcal{D}_{i}(x^t+u^t)$ where $u^t$ is a user-specified noise sequence.
Observe that conditioned on $u^t$, the equality holds:$$\E[q_i^t-z_i^t\mid u^t]=\bar A_i u^t.$$
Therefore, a good strategy for forming a new estimate $\hat A_i^{t+1}$ of $\bar{A}_i$ from $\hat A_i^{t}$ is to take a gradient step on the least squares objective
$$\min_{B_i}\frac{1}{2}\|q_i^t-z^t_i-B_iu^t\|^2.$$
Explicitly, this gives the update
$$\hat{A}_{i}^{t+1} =  \hat{A}_{i}^t + \nu_t (q_i^t-z_{i}^t-\hat{A}_{i}^t u^{t} )
(u^{t})^\top.$$
Analogous to estimation in adaptive control or machine learning, we exploit noise injection $u_t$ to
ensure sufficient exploration of the parameter space. In particular, the noise vector needs to be sufficiently isotropic. We impose the following assumption.

\begin{assumption}[Injected Noise]
\label{a:injected_noise}
{\rm
The injected noise vector $u^t=(u_{1}^t,\ldots, u_{n}^t)\in \mb{R}^d$
is a zero-mean random vector that is independent of $x^t$, and independent of
the injected noise at any previous queries to the environment by any player.
Moreover, there exists constants $\Bl, R>0$ and $\Bui>0$ for each $i\in [n]$ such that for all $t\geq 0$ and $i\in [n]$ the random vector $v_i:=u_i^t$ satisfies
\[0\prec\Bl\cdot I\preceq \mb{E}[v_iv^\top_i],\qquad \E\|v_i\|^2\leq \Bui,\qquad \text{and} \qquad
\mb{E}[\|v_i\|^2v_iv_i^{\top}]\preceq R^2 \mathbb{E}[v_iv_i^{\top}].\]
}
\end{assumption}

In the simple Gaussian case where $u_{t} \sim \mc{N}(0, I_{d})$, we may set\footnote{For the justification of the expression for $R^2$, see \cite[Section 2.1]{dieuleveut2017harder}.
}
$$\Bl=1,\qquad  \Bui = d_i, \qquad \textrm{and}\qquad R^2 = 3\max_{i\in [n]} d_i.$$
Analyzing the convergence of Algorithm~\ref{alg:adaptive_updates} amounts to
decomposing the analysis into
convergence of the stochastic gradient method  on the estimated losses induced by the
sequence of
$\hat{A}_i^t$, and convergence of the estimation
error $\E\|\hat{A}_i^t-\bar{A}_i\|^2$. The former analysis proceeds in  an
analogous fashion to that of
Theorem~\ref{thm:stochasticgradientplay} in
Section~\ref{sec:stochastic_gradplay_pp}. For the latter, we leverage the
injected noise to ensure there is sufficient \emph{exploration}. The following lemma establishes a one-step improvement guarantee on estimation of $\bar A_i$. Throughout, we set $\hat A^t:=(\hat A_1^t,\ldots, \hat A_n^t)$ and let $\|\cdot\|_F$ denote the Frobenius norm.
We also let $\mathbb{E}_t$ be the conditional expectation with respect to the $\sigma$-algebra generated by $(x^l,u^l)_{l=1,\ldots,t}$.

\begin{algorithm}[t!]
\caption{Adaptive Gradient Method}\label{alg:adaptive_updates}
\SetAlgoLined
\textbf{Input}: %
Stepsizes $\{\eta_t\}_{t\geq 1}$, $\{\nu_t\}_{t\geq 1}$; initial $x^1\in\R^d$,
$\hat{A}_{i}^1\in \R^{m\times d}$\;

\For{$t=1, \ldots, t$}{

\For{$i\in [n]$}{
\textbf{Query the environment}: Draw samples $z_{i}^t\sim \mathcal{D}_i(x^t)$ and $q_{i}^t\sim \mathcal{D}_i(x^t+u^t)$\;
  \textbf{Individual gradient update:}
  $x^{t+1}_i =\proj_{\X_i}\left(x_i^t- \eta_t
  (\nabla_{i}\ell_i(x^t,z_i^t)+(\hat{A}_{ii}^t)^{\top}\nabla_{z_i}\ell_i(x^t,z_i^t))\right)$,\\
  where $\hat{A}_{ii}^t$ denotes the submatrix of $\hat{A}_i^t$ whose columns are indexed by player $i$.

\noindent \textbf{Estimation update:}
$\hat{A}_{i}^{t+1} =  \hat{A}_{i}^t + \nu_t (q_i^t-z_{i}^t-\hat{A}_{i}^tu^t_i)(u^t_i)^\top$

}}
\end{algorithm}

\begin{lemma}[Estimation error]\label{lem:estimation_error}
Suppose that Assumptions~\ref{assum:param} and~\ref{a:injected_noise} hold and choose $\nu_t\in (0, \frac{2}{R^2})$. Then the matrices $\hat A_i^{t}$ generated by Algorithm~\ref{alg:adaptive_updates} satisfy the estimate:
\begin{equation}\label{eqn:one_step_estima}
\frac{1}{2}\mathbb{E}_t\|\hat A_i^{t+1}-\bar{A}_i\|^2_F\leq\frac{1-\Blower
\nu_t(2-\nu_t R^2)}{2}\|\hat A_i^t-\bar A_i\|^2_F+\nu^2_t{\normalfont
\tr}(\Sigma_i)\Bui.
\end{equation}
Therefore when setting $\nu_t=\frac{2}{\left(\Bl(t+\frac{2 R^2}{\Bl})\right)}$ for all $t\geq 0$, the estimate holds:
$$\mathbb{E}\|\hat A^{t}-\bar{A}\|^2_F\leq \frac{\max\left\{(1+\frac{2
R^2}{\Bl})\|\hat{A}_1-\bar A\|^2_F,\frac{8\sum_{i=1}^n
\normalfont{\tr}(\Sigma_{i})\Bui}{\Bl^2}\right\}}{t+\frac{2 R^2}{\Bl}}.$$
\end{lemma}
\begin{proof}
This follows from a standard estimate for online least squares, which appears as
Lemma~\ref{lem:least_square_step} in Appendix~\ref{app:proof_adaptive}. Namely,
let $\mathcal{G}_1$ be the $\sigma$-algebra generated by $x^1,\ldots, x^t$
and
let $\mathcal{G}_2$ be the $\sigma$-algebra generated by
$\mathcal{G}_1\cup\{u^t\}$. Set $b=q_i^t-z_i^t$, $y=u^t_i$,  $B=\hat A_i^t$,
$V=\bar A_i$, $v=u_i^t$, $\lambda_1=\Bl$, and $\lambda_2=\Bui$.

Let us upper bound the variance $\E[\|Vy-b\|^2\mid \mathcal{G}_2]$.
To this end, let $w$ and $w'$ be drawn i.i.d from $\PP_{i}$. Observe that conditioned on $u_i^t$, the random vector $\bar A_i u_i^t-(q_i^t-z_i^t)$ has the same distribution as
 $w-w'$. Let us compute
$$\E \|w-w'\|^2=\tr(\E((w-w')(w-w')^{\top})=2\tr(\Sigma_i).$$
Therefore, we may set $\sigma^2=2\tr(\Sigma_i)$.
An application of Lemma~\ref{lem:least_square_step} completes the proof of \eqref{eqn:one_step_estima}.
Summing up \eqref{eqn:one_step_estima} for $i=1,\ldots, n$ and using the tower rule for for conditional expectations yields:
$$\mathbb{E}\|\hat
A^{t+1}-\bar{A}\|^2_F\leq(1-\nu_t\Bl(2-\nu_t^2R^2)) \mathbb{E}\|\hat A^t-\bar
A\|^2_F+2\nu^2_t\sum_{i=1}^n{\normalfont \tr}(\Sigma_i)\Bui.$$
Noting $\nu_t\leq \frac{1}{R^2}$, we deduce $1-\nu_t\Bl(2-\nu_t^2R^2)\leq 1-\nu_t\Bl$.
The result follows directly from plugging in the value of  $\nu_t$ and using Lemma~\ref{lem:basic_recurs_lemma} in Appendix~\ref{sec:technical_sequences}.
\end{proof}

Next we show that the direction of motion of Algorithm~\ref{sec:adapt_grad_play}
is well-aligned with the direction of motion of the stochastic gradient method.
To this end, define the true (stochastic) vector of individual gradients
\[
v^t:=(\nabla_{i}\ell_i(x^t,z_{i}^t)+A_{i}^{\top}\nabla_{z_i}\ell_i(x^t,z_{i}^t))_{i\in[n]},
\]
and its estimator that is used by the algorithm
\[\hat{v}^t:=(\nabla_{i}\ell_i(x^t,z_{i}^t)+(\hat{A}_{ii}^t)^{\top}\nabla_{z_i}\ell_i(x^t,z_{i}^t))_{i\in [n]}.\]
We make the following Lipschitzness assumption on the loss $\ell_i(x,z_i)$ in the variable $z_i$.

\begin{assumption}[Lipschitz continuity in $z$]\label{assump:lip_func}
Suppose that there exists a constant $\delta>0$ such that for all $x\in \X$, the
estimate holds:
$$\E_{z\sim \mathcal{D}_{\pi}(x)}\sqrt{\sum_{i=1}^n\|\nabla \ell_i(x,z_i)\|^2}\leq \delta.$$
\end{assumption}

\begin{lemma}\label{lem:basic_thm_grad_est}
Suppose that Assumptions~\ref{assum:param} and~\ref{assump:lip_func} hold. Then for each $t\geq 1$ and $i\in [n]$, the estimate holds:
$$\mathbb{E}_t\|\hat{v}^t-v^t\|\leq \delta \|\hat{A}^t- \bar A\|^2_F.$$
\end{lemma}
\begin{proof}
Notice that we may write $\hat{v}^t-v^t=B^tw^t$, where $B^t$ is the block diagonal matrix with blocks $\hat{A}_{ii}^t-A_i$ and we set $w^t=(\nabla_{z_i}\ell_i(x^t,z_{i}^t))_{i=1}^n$. Using H\"{o}lder's inequality we estimate:
\begin{align*}
    \mathbb{E}_t\|\hat{v}^t-v^t\|&= \mathbb{E}_t\|B^tw^t\|\leq \|B^t\|_F\cdot \mathbb{E}_t \|w^t\|\leq \delta \|\hat A^t-\bar A\|^2_F,
\end{align*}
as claimed.
\end{proof}

In light of Lemmas~\ref{lem:estimation_error} and \ref{lem:basic_thm_grad_est}, we may interpret Algorithm~\ref{sec:adapt_grad_play} as an approximation to the stochastic gradient method with a bias that tends to zero; we may then simply invoke generic convergence guarantees for biased stochastic gradient methods, which we record in Theorem~\ref{thm:sgd} of Appendix~\ref{sec:append:generic_res}. We will make use of the following assumption.

\begin{assumption}[Finite variance]
Suppose that there exists $\sigma>0$ such that for all $x\in \X$, the variance bound holds:
$$\E_{z_i\sim \mathcal{D}_i(x^t)}\|\nabla_{i,z_i} \ell_i(x^t,z_i^t)-\E_{z_i'\sim
    \mathcal{D}_i(x^t)}\nabla_{i,z_i} \ell_i(x^t,z_i')\|^2\leq \sigma^2.$$
\end{assumption}

The end result is the following theorem, which in particular implies a $\mathcal{O}(d/t)$ rate of convergence when $u^t$ are standard Gaussian. See the discussion after the theorem.

\begin{theorem}[Convergence of the adaptive method]\label{thm:cconv_adapt}
    Suppose that Assumptions~\ref{assum:param}, \ref{assum:smoothness_elli},
    \ref{a:injected_noise}, and~\ref{assump:lip_func} hold and that the game
    \eqref{eqn:main_problem} is $\alpha$-strongly monotone. Define the constant
    $k_0=1+\frac{8 L^2}{\alpha^2}$ and $q_0=\frac{2R^2}{\Bl}$ and set
    $\eta_t=\frac{2}{\alpha(t+k_0-2)}$ and $\nu_t=\frac{2}{\Bl(t+q_0)}$ for all $t\geq 0$. Then for all $t\geq 1$, the iterates generated by Algorithm~\ref{alg:adaptive_updates} satisfy
\begin{align*}
\mathbb{E}\|x^{t}-x^{\star}\|^2 &\leq \frac{\max\left\{\frac{1}{2}\alpha^2(1+k_0)\|x_1-x^{\star}\|^2, ~32(1+2\|\bar A\|^2_F)\sigma^2 +8\delta^2 Z \max\{\frac{1+k_0}{1+q_0},1\} \right\}}{\alpha^2(t+k_0)}\\
&~~~+\frac{\max\left\{\frac{1}{2}\alpha^2(1+k_0)^{3/2}\|x_1-x^{\star}\|^2,~64\sigma^2 Z \max\{\frac{1+k_0}{1+q_0},1\}\right\}}{\alpha^2(t+k_0)^{3/2}}.
\end{align*}
where we set $Z=\max\left\{(1+\frac{2 R^2}{\Bl})\|\hat{A}^1-\bar
A\|^2_F,\frac{8\sum_{i=1}^n \normalfont{\tr}(\Sigma_{i})\Bui}{\Bl^2}\right\}$.
\end{theorem}
\begin{proof}
We will apply the standard convergence guarantees in Theorem~\ref{thm:sgd} of Appendix~\ref{sec:append:generic_res} for biased stochastic gradient methods. Using Lemma~\ref{lem:basic_thm_grad_est} we estimate the gradient bias:
$$\|\mathbb{E}_t[\hat{v}^t]- \mathbb{E}_t [v^t]\|=\mathbb{E}_t\|\hat{v}^t-{v}^t\|\leq  \delta \|\hat{A}^t- \bar A\|^2_F.$$
Next, we estimate the variance:
\begin{align*}
\mathbb{E}_t[\|\hat v_i^t-\E\hat v_i^t\|^2]&= \mathbb{E}_t\left\|\begin{bmatrix}I & 0\\
0 & \hat A_{ii}^t
\end{bmatrix}(\nabla_{i,z_i} \ell_i(x^t,z_i^t)-\E_{z_i'\sim
    \mathcal{D}_i(x^t)}\nabla_{i,z_i} \ell_i(x^t,z_i'))\right\|^2.
\end{align*}
Summing these inequalities over $i\in [n]$, we deduce
$$\E[\|\hat v_i^t-\E\hat v_i^t\|^2]\leq \max\{1,\|\hat A^t\|^2_{\rm op}\}\sigma^2.$$
Recalling the definition of $Z$ and $q_0$ and applying Theorem~\ref{thm:sgd} in Appendix~\ref{sec:append:generic_res} we deduce
\begin{align*}
\mathbb{E}_t\|x^{t+1}-x^{\star}\|^2&\leq \frac{1}{1+\eta^t\alpha} \|x^t-x^{\star}\|^2+\frac{2\eta^2_t(\max\{1,\|\hat A^t\|^2_{\rm op}\})\sigma^2}{1+\eta_t\alpha}+\frac{2\eta_t\delta^2 \|\hat{A}^t-\bar A\|^2_F}{\alpha}\\
&\leq \frac{1}{1+\eta_t\alpha} \|x^t-x^{\star}\|^2+2\eta^2_t(1+\|\hat A^t\|^2_F)\sigma^2+\frac{2\eta_t\delta^2 \|\hat{A}^t-\bar A\|^2_F}{\alpha}\\
&\leq \frac{1}{1+\eta_t\alpha} \|x^t-x^{\star}\|^2+2\eta^2_t(1+2\|\bar
A\|^2_F)\sigma^2+\frac{2\eta_t\delta^2 \|\hat{A}^t-\bar A\|^2_F}{\alpha}\\
&\quad+4\eta_t^2\sigma^2\|\hat{A}^t-\bar A\|^2_F,
\end{align*}
where the last  inequality follows from the algebraic expression $\|\hat A^t\|^2\leq 2\|\bar A\|^2+2\|\hat A^t-\bar A\|^2$. Taking expectations and using the tower rule, we compute
\begin{align*}
\mathbb{E}\|x^{t+1}-x^{\star}\|^2&\leq \frac{1}{1+\eta_t\alpha}
\E\|x^t-x^{\star}\|^2+2\eta^2_t(1+2\|\bar A\|^2_F)\sigma^2+\frac{2\eta_t\delta^2
    \E\|\hat{A}^t-\bar A\|^2_F}{\alpha}\\
    &\quad+4\eta_t^2\sigma^2\E\|\hat{A}^t-\bar A\|^2_F\\
&\leq \frac{1}{1+\eta_t\alpha} \E\|x^t-x^{\star}\|^2+2\eta^2_t(1+2\|\bar A\|^2_F)\sigma^2+\frac{2\eta_t\delta^2  Z}{\alpha(t +q_0)}+\frac{4\eta_t^2\sigma^2 Z}{t +q_0},
\end{align*}
where  the last inequality follows from Lemma~\ref{lem:estimation_error}.

Now our choice $\eta_t=\frac{2}{\alpha(t+k_0-2)}$ ensures $\frac{1}{1+\eta_t \alpha}=1-\frac{2}{t+k_0}$.
Therefore we deduce
\begin{equation}\label{eqn:key_recursions_adaptive}
\begin{aligned}
\mathbb{E}\|x^{t+1}-x^{\star}\|^2&\leq \left(1-\frac{2}{t+k_0}\right)\E\|x^t-x^{\star}\|^2+\frac{8(1+2\|\bar A\|^2_F)\sigma^2}{\alpha^2(t+k_0-2)^2}\\
&\quad+\frac{16\sigma^2 Z}{\alpha^2(t+q_0)(t+k_0-2)^2}+\frac{4\delta^2 Z}{\alpha^2(t+q_0)(t+k_0-2)}.
\end{aligned}
\end{equation}
We now aim to apply  Lemma~\ref{lem:lemma_on_seq} in Section~\ref{sec:technical_sequences}. To this end, we need to upper bound the last three terms in \eqref{eqn:key_recursions_adaptive} so that the denominators are of the form $(t+k_0)^p$ for some power $p$.
To this end, note the following elementary estimates:
$$\frac{t+k_0}{t+k_0-2}\leq \frac{k_0+1}{k_0-1}\leq 2$$
and
$$\frac{(t+k_0)^2}{(t+q_0)(t+k_0-2)}\leq \frac{k_0+1}{k_0-1}\cdot \frac{t+k_0}{t+ q_0}\leq \frac{c(k_0+1)}{k_0-1}\leq 2 c$$
where $c=\max_{t\geq 1}\{ \frac{t+k_0}{t +q_0}\}=\max\{\frac{1+k_0}{1+q_0},1\}$. Combining these estimates with \eqref{eqn:key_recursions_adaptive}, we obtain
\begin{align*}
\mathbb{E}\|x^{t+1}-x^{\star}\|^2&\leq \left(1-\frac{2}{t+k_0}\right)\|x^t-x^{\star}\|^2+\frac{32(1+2\|\bar A\|^2_F)\sigma^2 +8\delta^2 Z c}{\alpha^2(t+k_0)^2}+\frac{64\sigma^2 Z\cdot c}{(\alpha^2(t+k_0)^3)}.
\end{align*}
Applying Lemma~\ref{lem:lemma_on_seq} in Section~\ref{sec:technical_sequences}, we conclude:
\begin{align*}
\mathbb{E}\|x^{t}-x^{\star}\|^2 &\leq \frac{\max\left\{\frac{1}{2}\alpha^2(1+k_0)\|x_1-x^{\star}\|^2, ~32(1+2\|\bar A\|^2_F)\sigma^2 +8\delta^2 Z c \right\}}{\alpha^2(t+k_0)}\\
&~~~+\frac{\max\left\{\frac{1}{2}\alpha^2(1+k_0)^{3/2}\|x_1-x^{\star}\|^2,~64\sigma^2 Z c\right\}}{\alpha^2(t+k_0)^{3/2}}.
\end{align*}
This completes the proof.
\end{proof}

In particular, consider the Gaussian case $u^t\sim \mathcal{N}(0,I_d)$ in the
setting when $d_i=C_i d$ 
for some numerical constants $C_i$, and when the traces $\tr(\Sigma_i)\equiv \tr(\Sigma)$ are equal for all $i\in [n]$. Then the efficiency estimate in Theorem~\ref{thm:cconv_adapt} becomes
\begin{align*}
&\E\|x^t-x^{\star}\|^2\\
&= \mathcal{O}\left(
\frac{\max\left\{L^2\|x_1-x^{\star}\|^2, ~\|\bar A\|^2_F\sigma^2 +\delta^2
\max\left\{d\|\hat{A}^1-\bar A\|^2_F,\normalfont{\tr}(\Sigma)\sum_{i=1}^n
\Bui\right\} \max\{\frac{L^2}{d\alpha^2},1\} \right\}}{\alpha^2t+L^2}\right.\\
&~~~\left.+\frac{\max\left\{L^3\|x_1-x^{\star}\|^2,~\alpha\sigma^2
\max\left\{d\|\hat{A}^1-\bar A\|^2_F,\normalfont{\tr}(\Sigma)\sum_{i=1}^n
\Bui\right\} \max\{\frac{L^2}{d\alpha^2},1\}\right\}}{(\alpha^2
    t+L^2)^{3/2}}\right).
\end{align*}
Consequently, treating all terms besides $d$ and $t$ as constants, yields the rate $\mathcal{O}(\frac{d}{t})$.

\section{Numerical Examples}
\label{sec:numerics}
Section~\ref{sec:monotonicity} introduced two examples motivated by
practical problems including revenue maximization in ride-share markets and
interactions between election prediction platforms. In this section, we explore
each of these examples. For the former, we create \emph{semi-synthetic}
 experiments using real data from a ride-share market in Boston, MA that includes two
well-known platforms (Uber and Lyft). We demonstrate how the effects of modeling competition impacts
revenue generation and demand.  For the latter, we create synthetic
 experiments that explore the gap between the social cost at the social
optimum, performatively stable equilibrium and the Nash equilibrium. 
We start with the purely synthetic example in
Section~\ref{sec:strategic_prediction}, and then move on to the
semi-synthetic example in Section~\ref{sec:rideshare}.

\subsection{Multiplayer Performative Prediction with Strategic Data Sources}
\label{sec:strategic_prediction}
\begin{figure}[t!]
    \centering 
    \subfloat[][Error to Nash
    Equilibrium\label{fig:complexity_nash_synthetic}]{

        \includegraphics[height=0.3\textwidth]{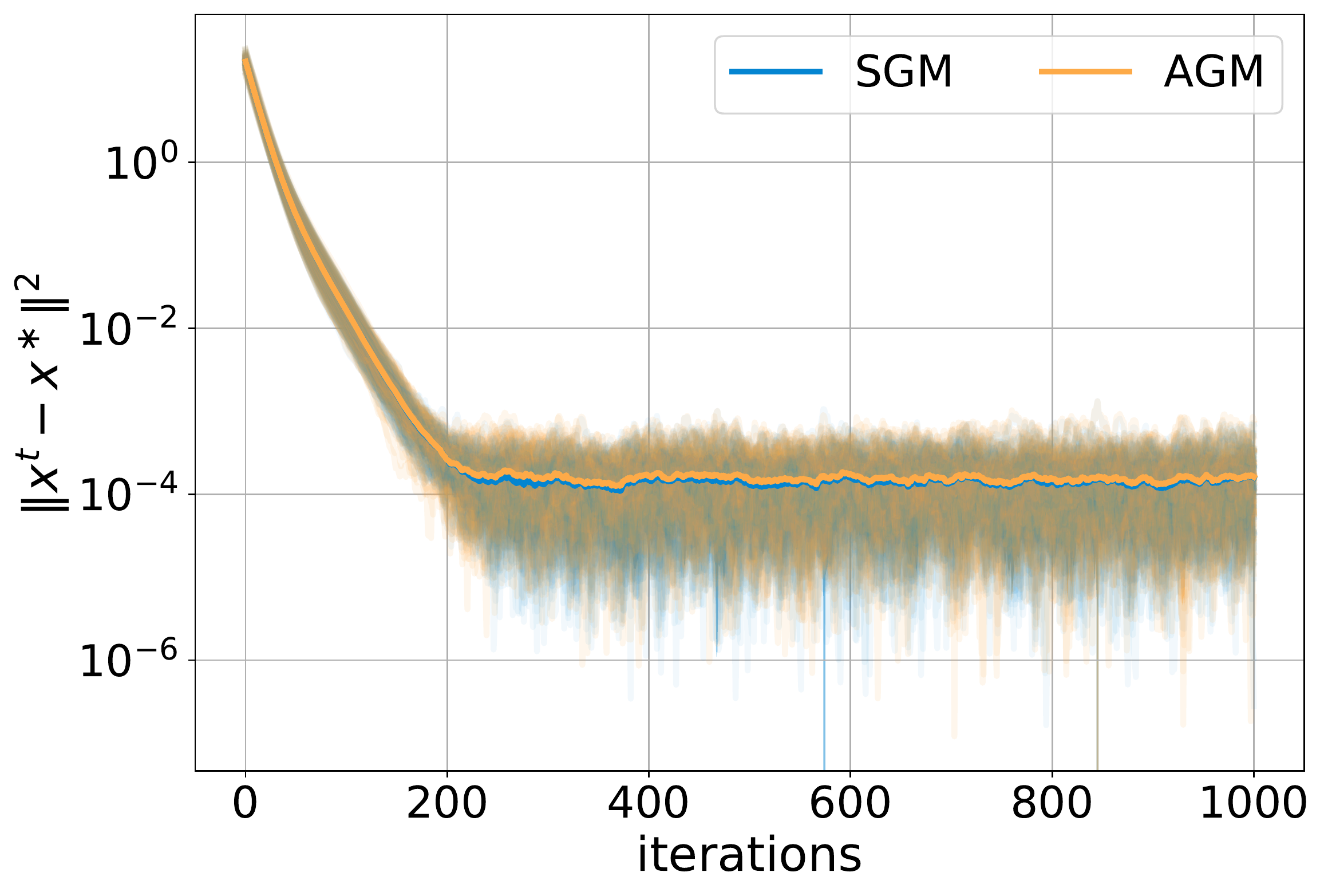}
    }\hspace*{0.25in}\subfloat[][\label{fig:complexity_ps_synthetic} Error to Performatively Stable
        Equilibrium
    ]{\includegraphics[height=0.3\textwidth]{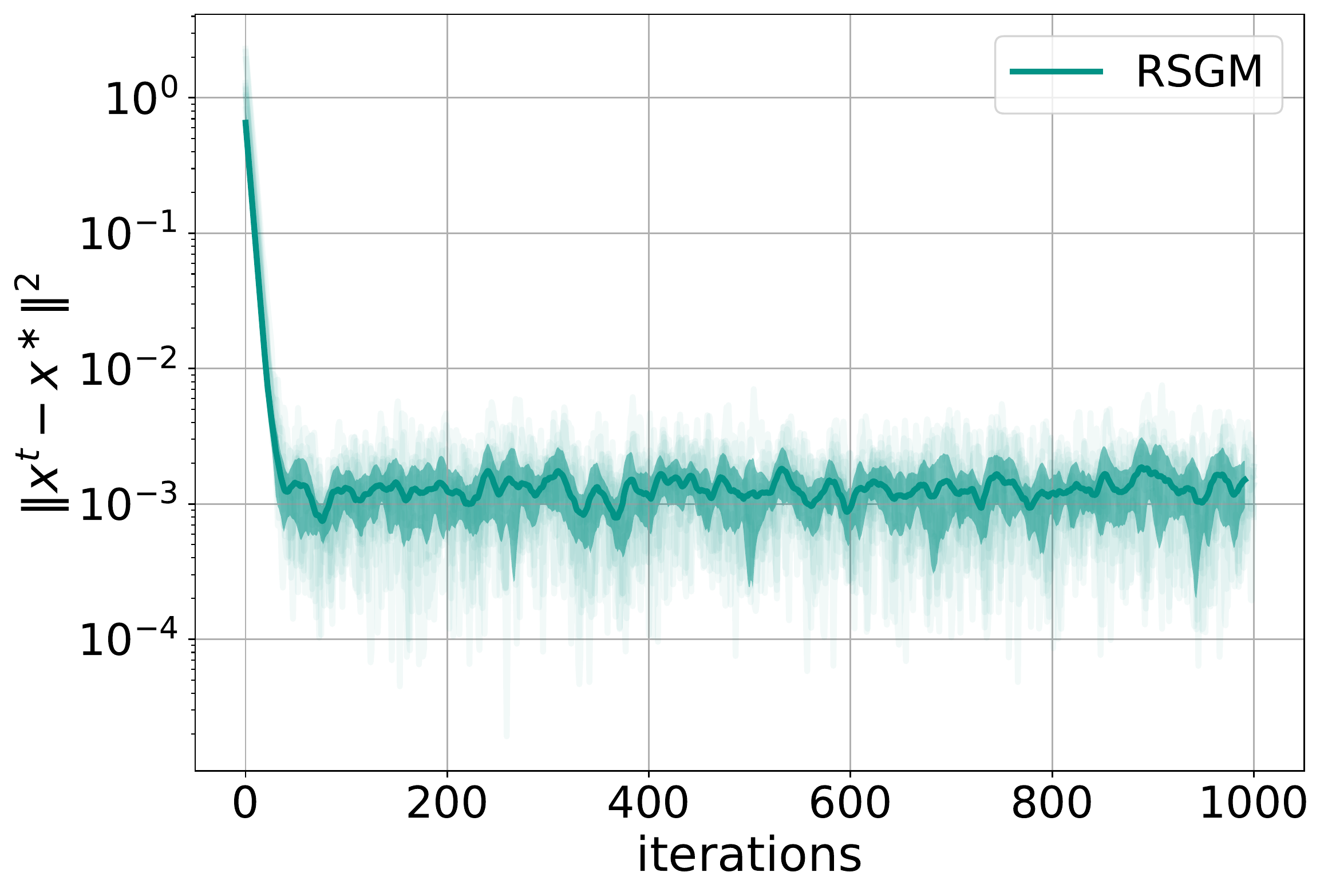}}

        \caption{\textbf{Strategic Prediction.} (a) Iteration complexity of the
        norm-squared error to the Nash equilibrium for both
    the stochastic gradient method (SGM) and the adaptive gradient method (AGM)
for a randomly generated problem instance where $m_i=10$ and $d_i=2$  for each
$i\in\{1,2\}$. The sample complexity of AGM empirically matches that of SGM even
for fairly large problem instances. (b) Iteration complexity of the norm-squared
error to the performatively stable equilibrium for the stochastic
repeated gradient method (SRGM) for a randomly generated problem instance where
$m_i=100$ and $d_i=2$  for each
$i\in\{1,2\}$. The empirical rates matches the theory in
Section~\ref{sec:algorithms} and Section~\ref{sec:algorithms_nash}. }
    \label{fig:convergence_synthetic}
\end{figure}

Recall in  Section~\ref{sec:monotonicity} in
Example~\ref{ex:strategic_prediction} we introduced a performative
prediction game motivated by multiple election platforms.
The decision problem that player $i$ faces is such that a set of features 
are drawn i.i.d.~from a
static distribution $\PP_\theta$, and player $i$ observes a sample drawn from the conditional distribution
\[z_i|\theta\sim \varphi_i(\theta)+A_ix_i+A_{-i}x_{-i}+w_i,\]
where $\varphi_i:\mb{R}^{d_i}\to \mb{R}$ is an arbitrary map, the parameter matrices
$A_i\in \mb{R}^{m_i\times d_i}$ and $A_{-i}\in \mb{R}^{m_i\times d_{-i}}$ are
fixed, and $w_i$ is a zero-mean random variable.  Each player seeks to predict $z_i$ by minimizing the loss
\[\ell_i(x_i,z_i)=\frac{1}{2}\|z_i-\theta^\top x_i\|^2.\]
We showed in Example~\ref{ex:strategic_prediction}
that  the map $x\mapsto H_x(y)$ is monotone as long as 
\[\sqrt{n-1}\cdot\max_{i\in [n]}\|A_{-i}^\top A_i\|_{\rm op}\leq \min_{i\in
[n]}\lambda_{\rm{min}}(A_i^{\top}A_i),.\]
As noted in Example~\ref{ex:strategic_prediction}, the
interpretation is that the performative effects due to
interaction with competitors are small relative to any player's own performative
effects. 
We enforce this condition on the parameters selected for the examples we explore
numerically.

We randomly generate problem instances---namely, the parameters $A_i,A_{-i}$ for
$i\in [n]$---by using {\texttt{scipy.sparse.random}}
which allows for the sparsity of the matrix to be set in addition to randomly
generating the matrix parameters. We set $d=5$ and $m=100$ for the experiments
in Figure~\ref{fig:convergence_synthetic}. These values can be changed in the
provided code, resulting in similar conclusions regarding the convergence rate.  In Figure~\ref{fig:complexity_nash_synthetic}
we show the iteration complexity of the norm-square error to the Nash equilibrium for both the stochastic gradient method 
and the adaptive gradient method. Analogously, in Figure~\ref{fig:complexity_ps_synthetic}, we show the iteration
complexity of the norm-square error to the performatively stable equilibrium for
the stochastic repeated gradient method. 
 The mean and $\pm1$ standard deviation are
depicted with darker  lines and the individual sample trajectories are 
shown using lighter shade lines of the same color indicated in the legend for
the two methods. 
For both Nash-seeking methods (stochastic gradient and adaptive gradient method), we 
use a constant step size $\eta=0.01$ for both
methods for the stochastic gradient update step. Additionally, we use the step size
$\nu=2/(t+6d)$ for
the estimation update in the adaptive gradient method. We see that their iteration complexities are
empirically the same.  For the performatively stable equilibrium seeking
algorithm, we use a constant step-size $\eta=0.1$.

\subsection{Revenue Maximization: Competition in Ride-Share Markets}
\label{sec:rideshare}
The next numerical example we explore is semi-synthetic competition between two ride-share
platforms seeking to maximize their revenue given that the demand they
experience is influenced by their own prices as well as their competitors. We
use data from a prior Kaggle competition to set up the semi-synthetic simulation
environment.\footnote{The data used in this paper is publicly available
    (\url{https://www.kaggle.com/brllrb/uber-and-lyft-dataset-boston-ma}).}

    \paragraph{Game Abstraction.}   The abstraction for the game can be described as follows. Consider a ride-share market with two platforms that
each seek to maximize their revenue by setting the price $x_i$. The vector of
demands
$z_i\in \mb{R}^{m_i}$ containing demand information for $m_i$ locations that
each ride-share platform sees is influenced not only by the prices they set but also
the prices that their competitor sets.  Suppose that platform $i$'s loss is given by
\[\ell_i(x_i,z_i)=-\frac{1}{2}z_i^{\top}x_i+\frac{\lambda_i}{2}\|x_i\|^2\]
where $\lambda_i\geq 0$ is some
regularization parameter, and $x_i\in \mb{R}^{m_i}$ represents the vector of
price differentials from some nominal price for each of the $m_i$ locations.
Observe that this game is separable since the losses $\ell_i$ do not explicitly
depend on $x_{-i}$. Moreover, let us suppose  that the random demand $z_i$
takes the semi-parametric form 
$z_i=\bz_{i}+A_ix_i+A_{-i}x_{-i}$, where $\bz_{i}$ follows some base
distribution $\PP_{i}$ and
the parameters $A_i$ and $A_{-i}$ capture price
elasticities to player $i$'s and its competitor's change in price, respectively;
naturally, the price elasticity for player $i$ to its own price changes is
negative while the price elasticity for player $i$'s demand given changes in its
competitors actions is positive. Namely, we have that $A_i\preceq 0$ and
$A_{-i}\succeq 0$ capturing that an increase in player $i$'s prices results in a
decrease in demand, while an increase in its competitor's prices results in a
increase in demand. Moreover, 
we showed in Example~\ref{example:revenue_max} that the mapping $x\mapsto H_x(y)$
is trivially monotone. Hence, the game between ride-share platforms is strongly
monotone and admits a unique Nash equilibrium. 
Throughout the remainder of this section we set $\lambda_1=\lambda_2=1$.

\begin{figure}[t!]
    \centering \subfloat[][Error to Nash Equilibrium\label{fig:complexity_nash}]{
     \includegraphics[height=0.3\textwidth]{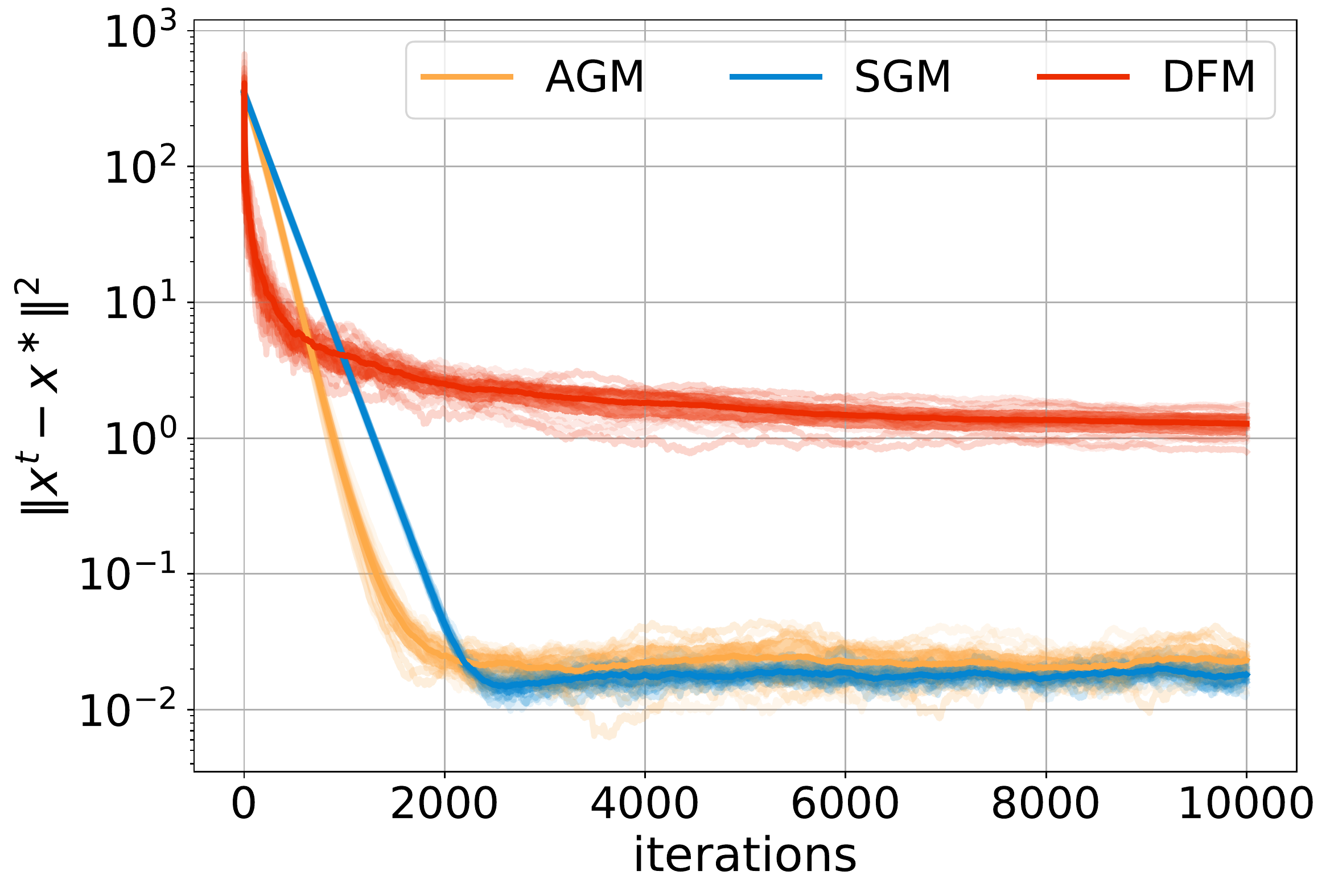}}\hspace*{0.25in}\subfloat[][\label{fig:complexity_ps}Error
        to Performatively Stable Equilibrium]{\includegraphics[height=0.3\textwidth]{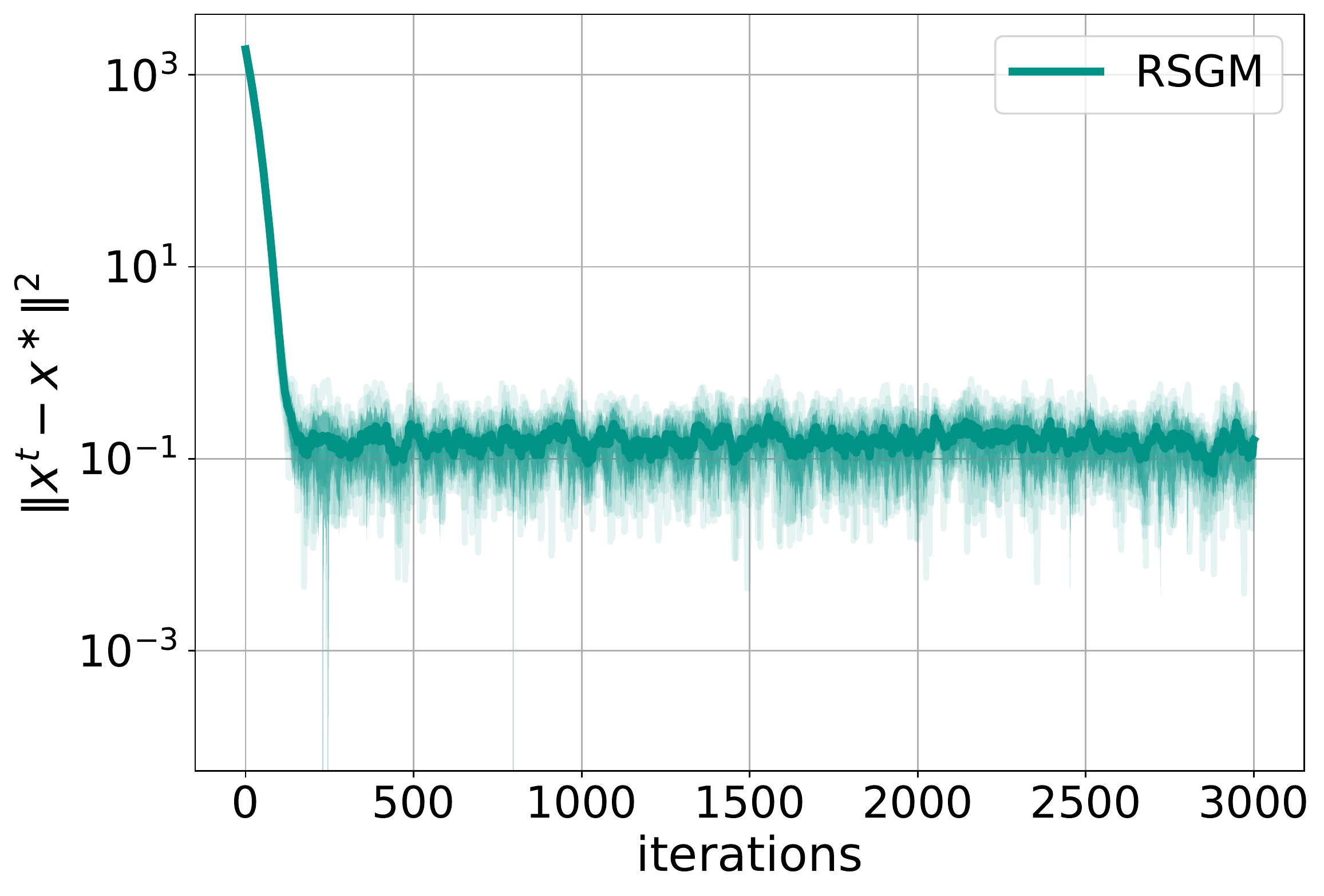}}
        \caption{\textbf{Competition in Ride-Share Markets: Experiment 1.} (a) Iteration complexity for 
        Nash seeking algorithms including derivative free
    gradient method (DFM), stochastic gradient method (SGM) and adaptive
gradient method (AGM). (b) Iteration complexity for the repeated stochastic
gradient method (RSGM), a performatively stable
equilibrium seeking algorithm. The plots demonstrate that the sample
complexities
of AGM and SGM are nearly identical and outperform DFM as expected.  }
    \label{fig:convergence}
\end{figure}
\begin{figure}[t!]
    \centering
   \subfloat[][\label{fig:loss_comp}Loss]{\includegraphics[height=0.25\textwidth]{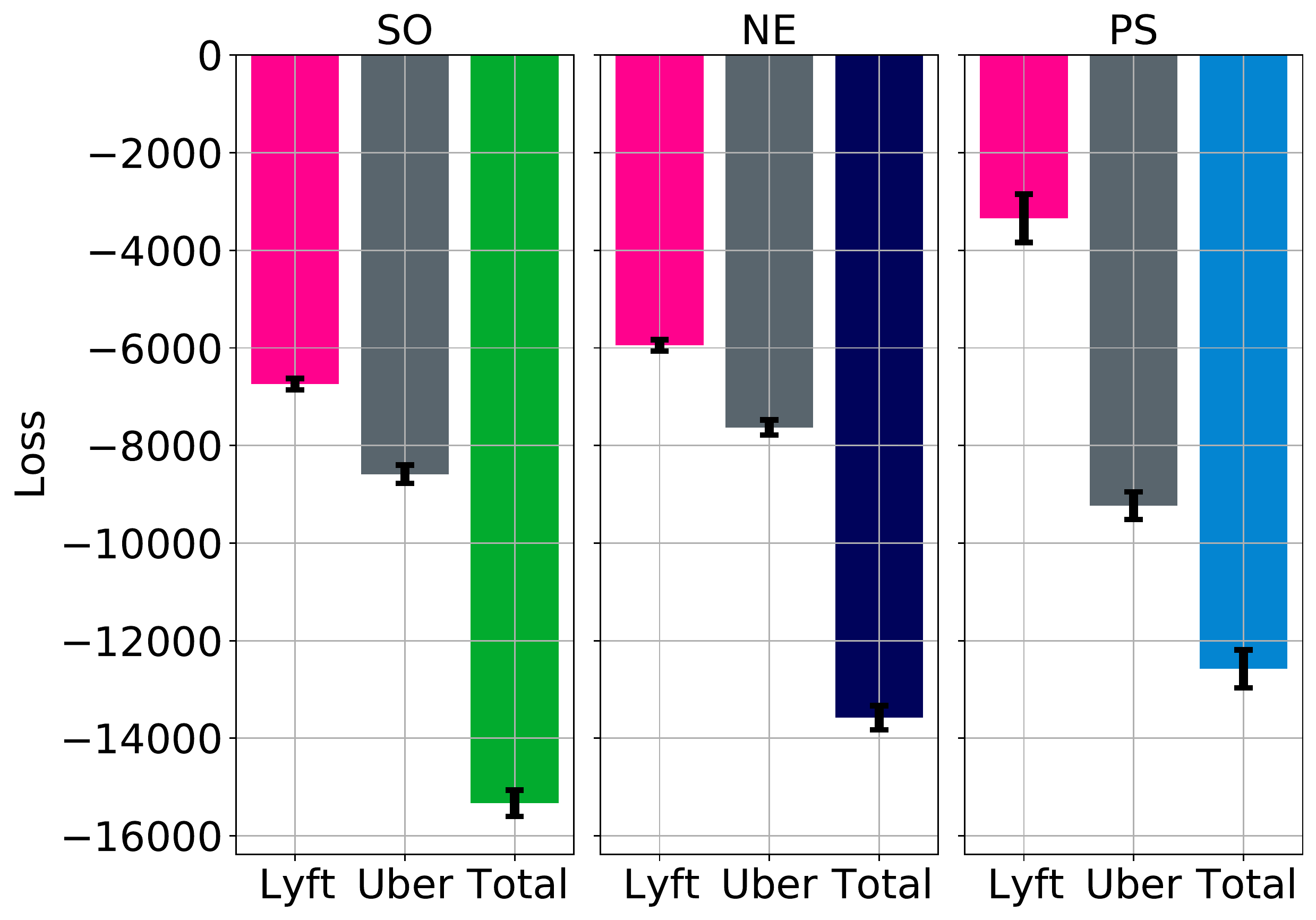}}\hfill\subfloat[][\label{fig:revenue}Revenue]{\includegraphics[height=0.25\textwidth]{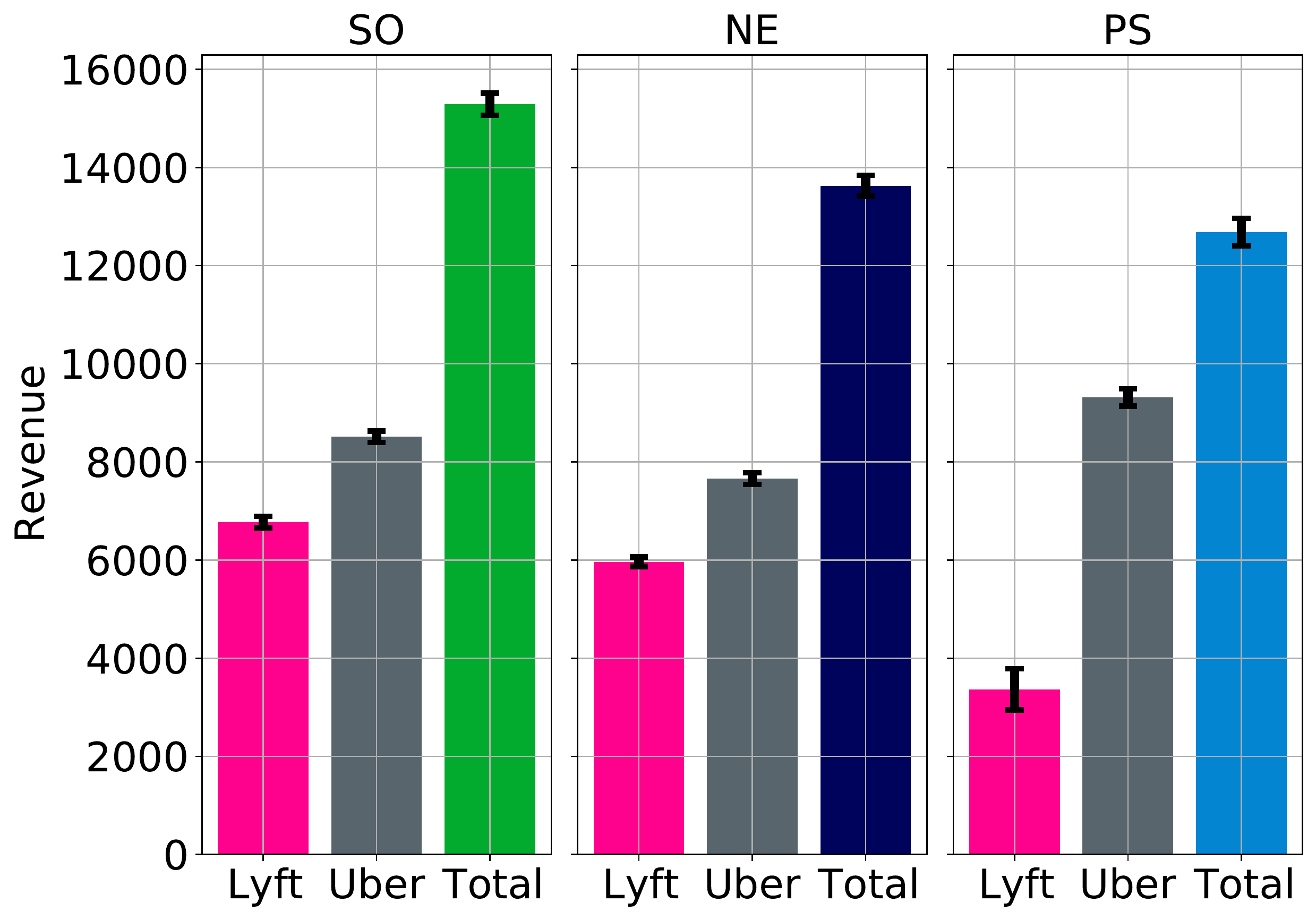}}\hfill\subfloat[][\label{fig:poa}PoA]{\includegraphics[height=0.24\textwidth]{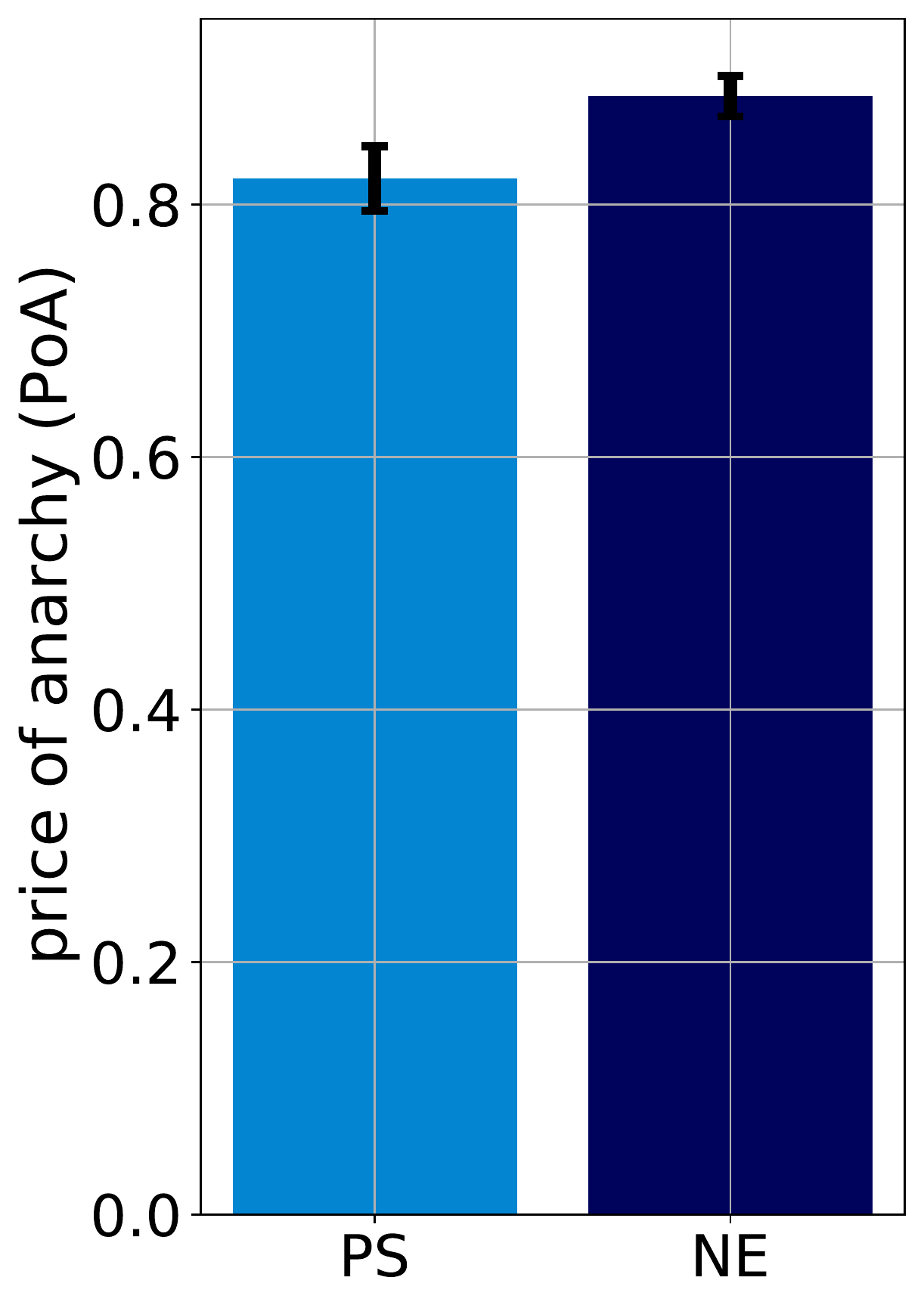}}\hfill

        \caption{\textbf{Competition in Ride-Share Markets: Experiment  2.}
        (a) Loss  and (b) revenue at the social optimum obtained via
stochastic gradient descent on the social cost (sum of players' costs), Nash
equilibrium obtained via the stochastic gradient method, and performatively stable equilibrium obtained via
the
repeated stochastic gradient method.  The overall (sum of both players) loss and
revenue are worse at the performatively stable equilibrium. Note that the loss
is the negative of the revenue plus some small $(\lambda_1=\lambda_2=1)$ regularization term.  (c) Average price of anarchy
(PoA) at the Nash equilibrum versus the performatively stable equilibrium. A
value closer to one is better, and hence the Nash equilibrium (by a small
margin) has a better PoA.}
    \label{fig:comparing}
\end{figure}
\paragraph{Semi-Synthetic Data Construction.} 
There are eleven locations that we consider in our simulation, and each element
in $x_i$ represents the price difference (set by platform $i$) from a nominal
price. We aggregate the rides into bins of \$5 increments; this is done by
taking the raw data and rounding the price to the nearest bin as follows
$5\cdot \lfloor \frac{p}{5} \rfloor$ where $p$ is the price of the ride. 
Then, for each bin $j$ we have a different base empirical distribution $\PP_{i,j}$ for each player
$i\in \{1,2\}$ which is just the collection of rides for that bin. 

For each bin, we estimate these price elasticity matrices $A_i$ and $A_{-i}$
from the data using the heuristic that a 50\% increase in price by any firm
leads to a 75\% decrease in demand. With this heuristic we use the average base
demand for each location and price bin pair to estimate both the diagonal elements of
$A_i$ and $A_{-i}$. In the experiments presented, our semi-synthetic model is
such that there is no correlation between locations; however, in the provided
code base, we have additional experiments that estimate the correlation between
locations and explore the effects of positive and negative correlations on
equilibrium outcomes.\footnote{The data and code used in this paper are publicly
    available
    (\url{https://github.com/ratlifflj/performativepredictiongames}).}
We further note that the results presented in this section are for the \$10 
nominal price bin, however, in the repository of code it is easy to adjust this
parameter to any of the other price bins. The conclusions we draw are similar
across the bins. 

    \paragraph{Experiment 1: Numerical Comparison of Iteration Complexity.} To validate the theory developed in the previous
    sections, we show the iteration complexity of the Nash seeking algorithms
    (Figure~\ref{fig:complexity_nash}),
    and the performatively stable equilibrium seeking algorithms
    (Figure~\ref{fig:complexity_ps}). We run each algorithm from twenty random
    initial conditions, and compute the error between the trajectory of the
    algorithm and the Nash equilibrium (respectively, performatvely stable
    equilibrium). In Figure~\ref{fig:complexity_nash} (respectively,
    Figure~\ref{fig:complexity_ps}), we  show the mean of these error
    trajectories and plus and minus one standard deviation. All the raw
    trajectories are also shown using a less opaque trajectory. For the
    stochastic gradient method, we use a
    constant step size $\eta=5e$-$5$ for the gradient update, and for the
    adaptive gradient method we use the step size schedule $\eta_t=\eta_0/t$ for
    the gradient update and $\nu_t=\nu_0/t$ for the estimation update where
    $\eta_0=5e$-$5$ and $\nu_0=1e$-$4$. For the derivative free method, we use a
    constant query radius $\delta=5$ and step size schedule $\eta_t=2/t$. 
The plots in Figure~\ref{fig:complexity_nash} demonstrate the that empirical sample complexity
of the adaptive gradient method and the stochastic gradient method are nearly
identical, and outperform the derivative free method as expected.
Figure~\ref{fig:complexity_ps} simply illustrates the convergence rate as
predicted by the theory in Section~\ref{sec:algorithms} for the repeated stochastic
gradient method.

    \paragraph{Experiment 2: Social Efficiency of Different Equilibrium
             Concepts.}
As noted in the preceding sections, in the study of equilibrium for games, it is
important to understand the efficiency of different equilibrium concepts. 
The typical benchmark for efficiency is 
the cost at the social optimum. 
The social cost is defined as the sum of all the players individual costs:
\[\mc{S}(x)=\sum_{i=1}^n\LL_i(x).\]
We find the unique socially optimal equilibrium $x^{\tt so}$ by running stochastic gradient
descent on the social cost (with $\eta=0.001$). Let $x^{\tt ne}$ be the Nash equilibrium of the game
$(\LL_1,\LL_2)$ and let $x^{\tt ps}$ be the performatively stable equilibrium of
the game $\mc{G}(x^{\tt ps})$, using the notation from
Section~\ref{sec:algorithms} for the game induced by $x^{\tt ps}$.
 
To compute the Nash equilibrium and the perfomatively stable equilibrium we run
the stochastic gradient method and the repeated stochastic gradient method with
step-size $\eta=0.001$, respectively.
In
Figure~\ref{fig:comparing} we show the loss at each of the 
equilibrium concepts for each player, and the total loss. For the ride-share game, both the Nash equilibrium and the
performatively stable equilibrium are unique, and hence this set is a singleton.

The price anarchy is a
common metric for equilibrium efficiency and is defined as the ratio of the
social cost at the worst case competitive equilibrium concept relative to the social cost
at the social optimum---namely, it is given by
\[{\tt PoA}(x)=\frac{\max_{x\in {\tt Eq}(\mc{G})}\mc{S}(x)}{\mc{S}(x^{\tt
so})},\]
where ${\tt Eq}(\mc{G})$ denotes the set of equilibria for the game
$\mc{G}$. An equilibrium concept is said to be more efficient the closer this
quantity is to one.

Given the stochastic
nature of the problem at hand, we define the empirical price of anarchy as the
ratio of the corresponding empirical social costs, and denote it as
$\widehat{{\tt PoA}}(x)$.
Hence, in comparing equilibrium concepts---i.e., Nash versus performatively
stable equilibrium---the equilibrium with the social cost closest to the social
cost at the social optimum is desirable. 
The empirical price of anarchy for the Nash equilibrium of the ride-share game is $\widehat{{\tt PoA}}(x^{\tt ne})=0.899$
while the empirical price of anarchy for the peformatively stable equilibrium is
$\widehat{{\tt
PoA}}(x^{\tt ps})=0.829$, where we compute the empirical social cost at the
corresponding equilibrium. This is also illustrated in Figure~\ref{fig:poa}. These ratios are
fairly close, indicating that it is an interesting direction of future research
to better understand the analytical properties of efficiency the different equilibrium
concepts. 

In Figure~\ref{fig:revenue}, we also show the revenue for each of the
equilibrium concepts. The revenue shares an analogous story to the loss: the
total revenue at the
Nash equilibrium is closer to the social optimum, but the performatively stable
equilibrium is not far off. 
\begin{figure}[t!]
    \centering
  \subfloat[][
   \label{fig:demand_revenue_total}Price change 
 by
 location.]{\includegraphics[height=0.155\textwidth]{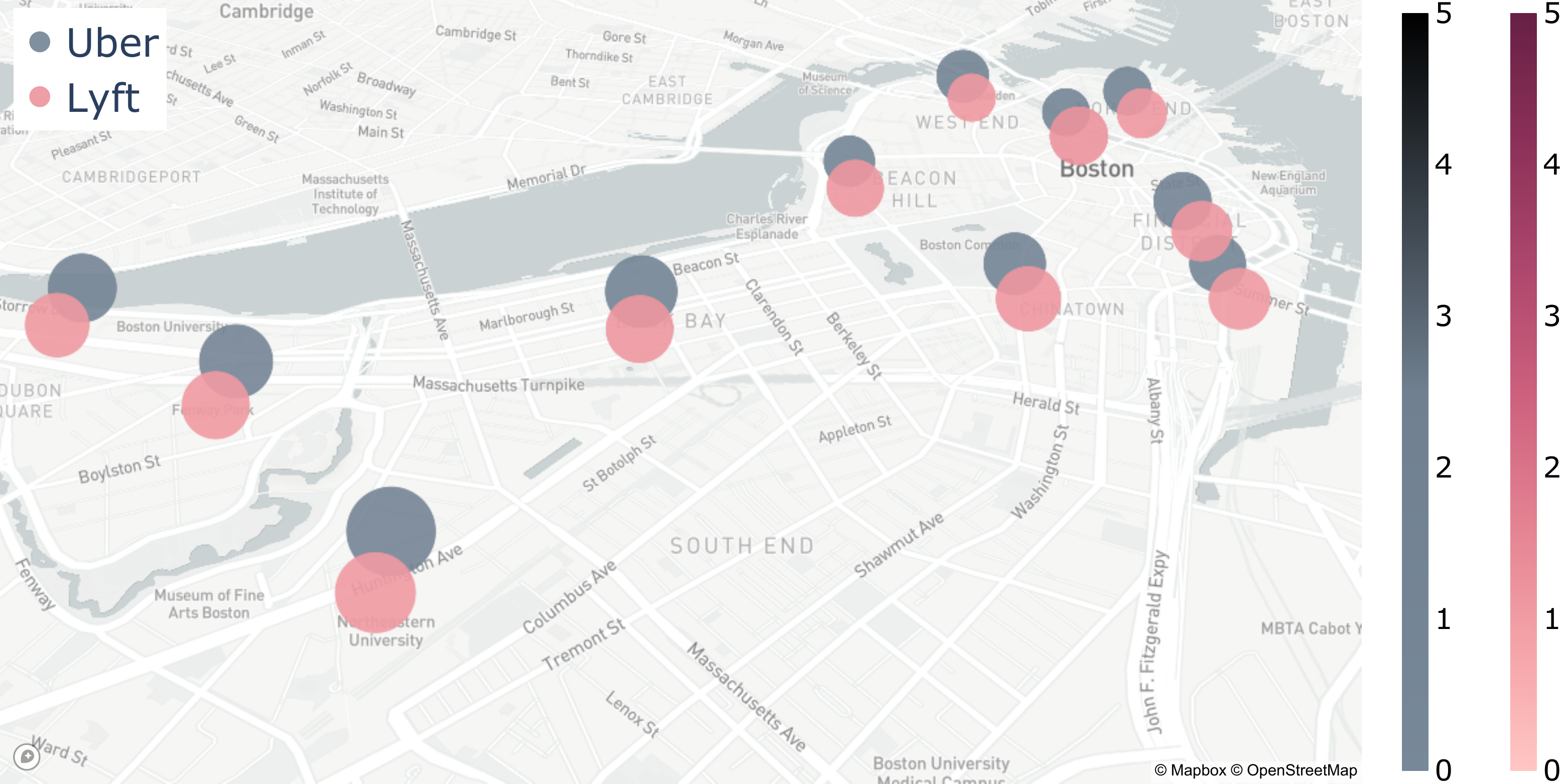}}\hfill\subfloat[][Demand
       change by location.
    ]{\includegraphics[height=0.155\textwidth]{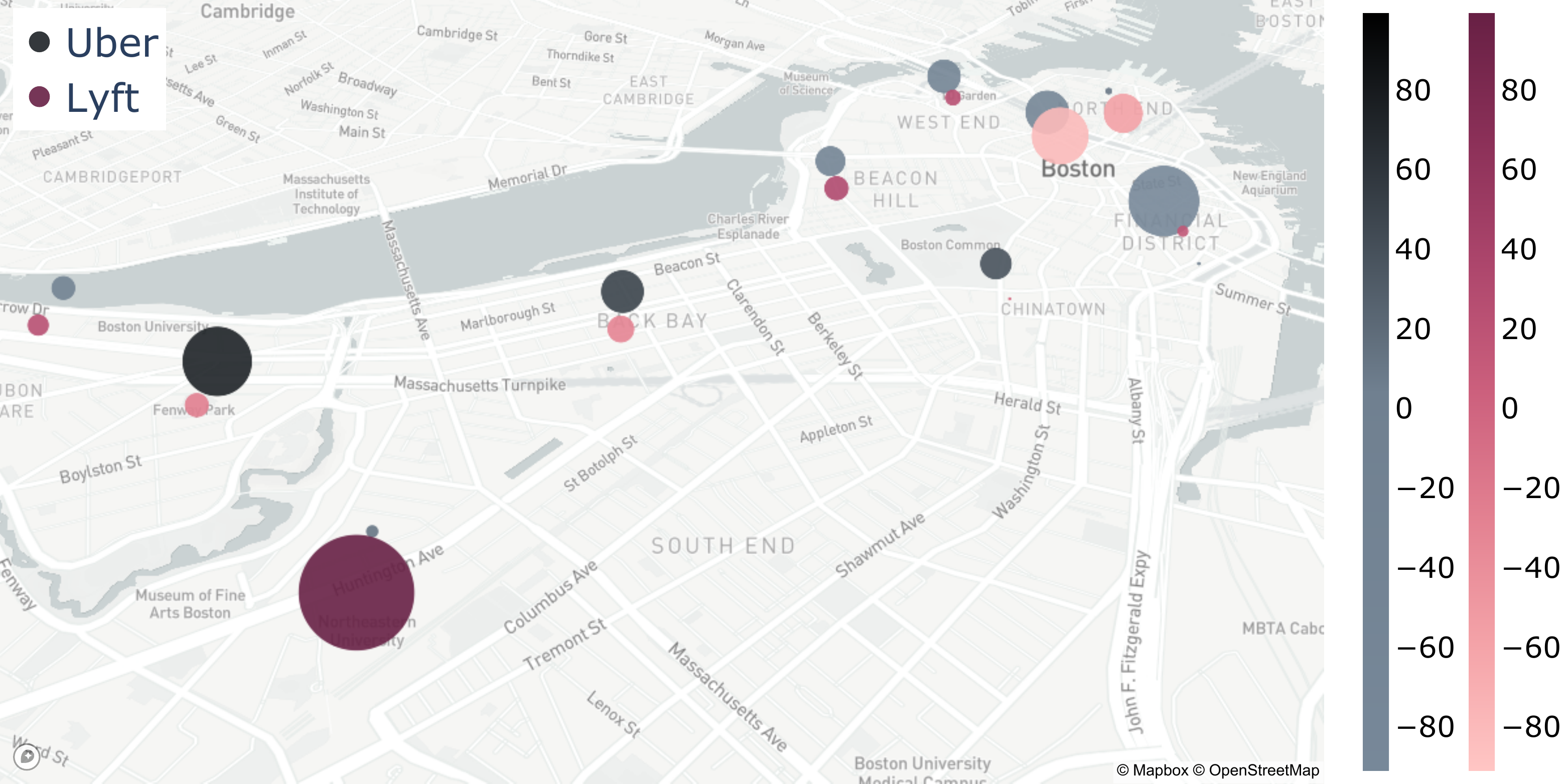}}\hfill\subfloat[][Revenue
    change by location.
]{\includegraphics[height=0.155\textwidth]{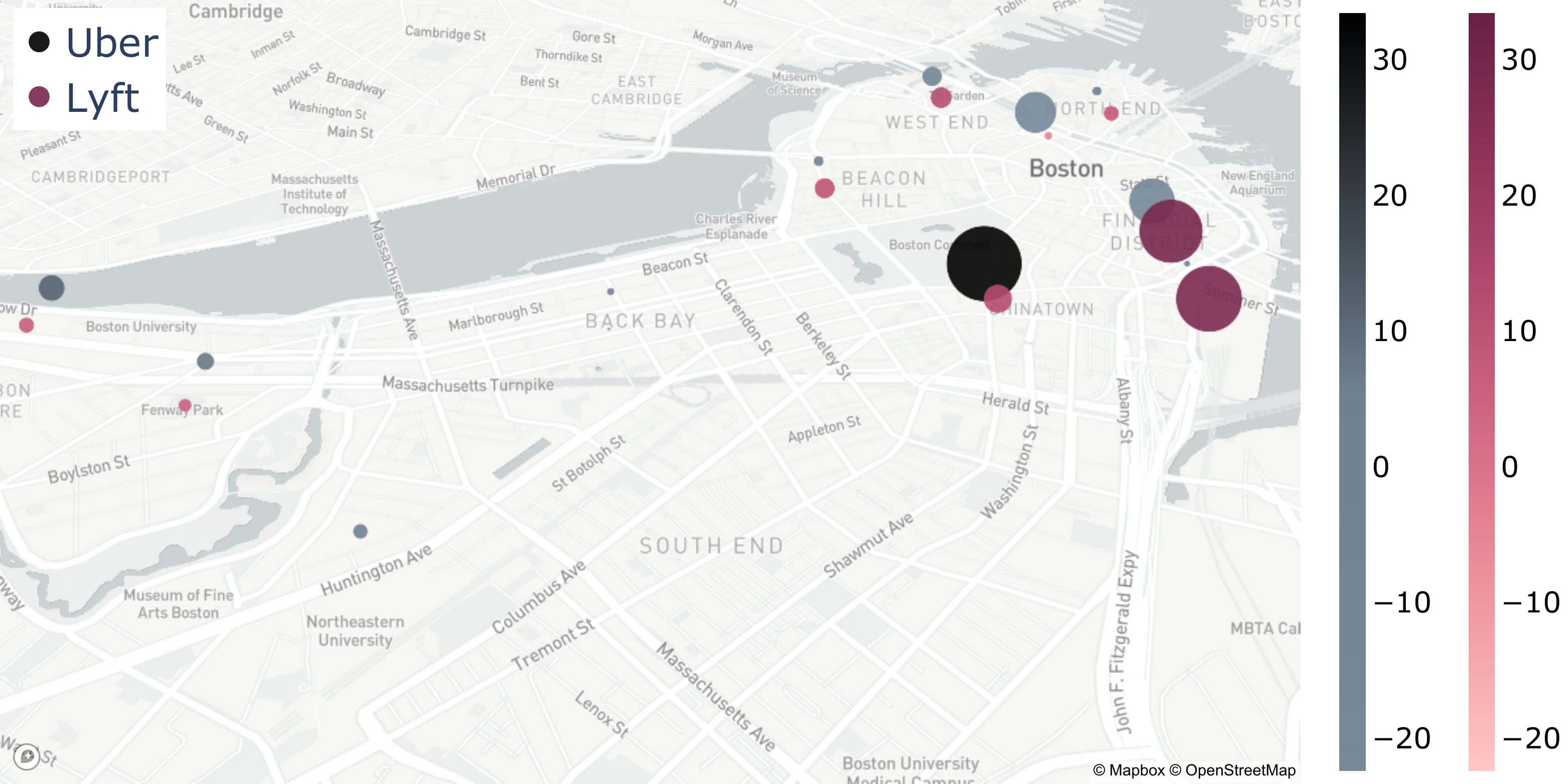}}   
    \caption{\textbf{Competition in Ride-Share Markets:  Experiment 2.} Effects due to myopic
    decision-making. Change in (a) price, (b) demand and (c) revenue from nominal by
    location for
\$10 nominal price bin due to ignoring performative effects: players run stochastic gradient descent, and the image shows
the change in demand (respectively, revenue) when both players model decision-dependence as compared to when
they both do not model decision-dependence. The
size of the circles shows the magnitude of the change, while the color indicates
the raw value. The majority of locations see a \emph{decrease} in demand, but
due to an increase in price at the Nash equilibrium relative to the myopic
outcome, there is an increase in 
revenue for
\emph{both} players.}
    \label{fig:demand_revenue_location}
\end{figure}

\begin{figure}[t!]
    \centering
   \subfloat[][
   \label{fig:both_myopic}Both
   Myopic]{\includegraphics[height=0.2\textwidth]{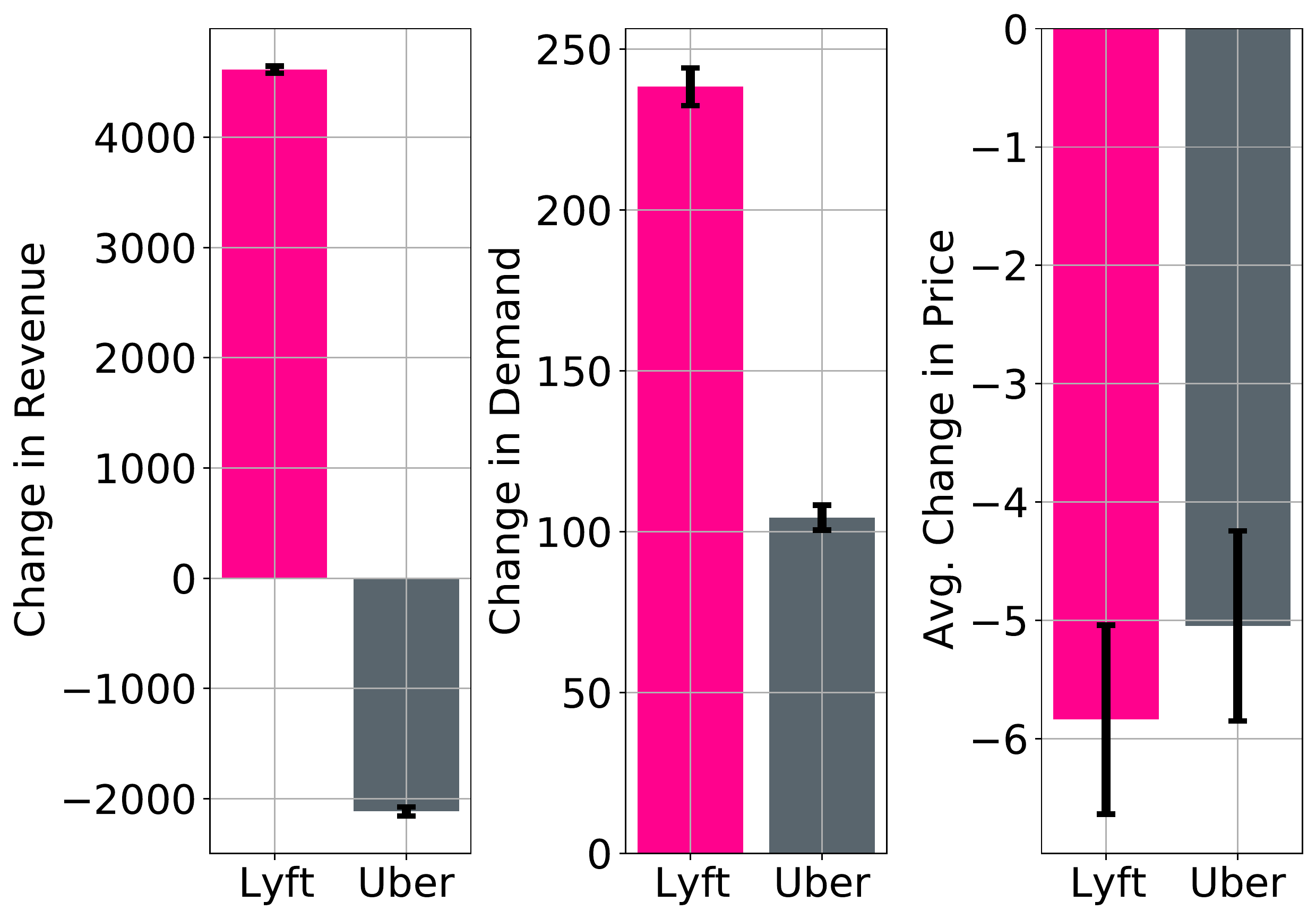}}\hfill\subfloat[][\label{fig:uber_myopic}Uber Myopic Only
   \label{fig:demand_revenue_total}]{\includegraphics[height=0.2\textwidth]{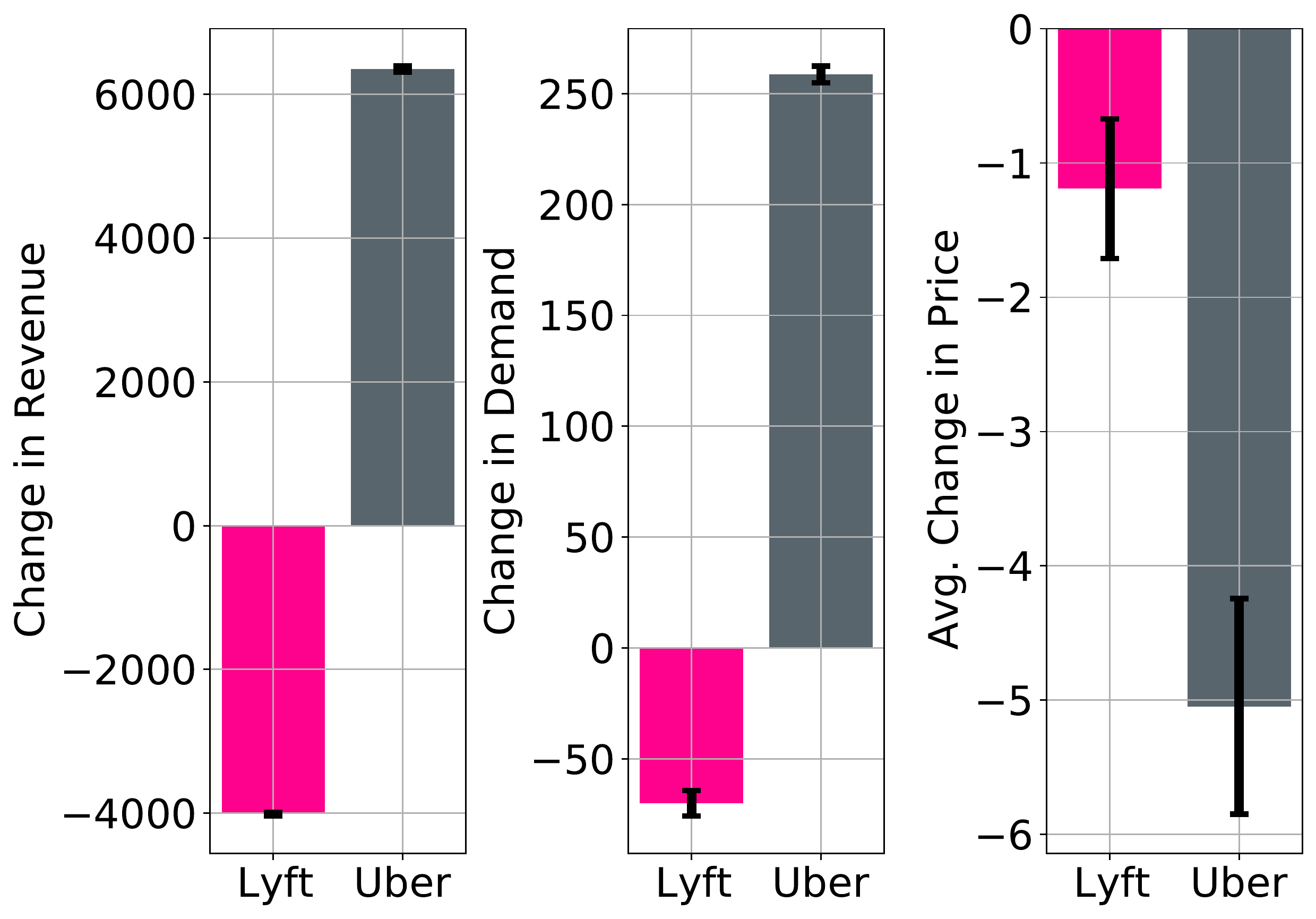}}\hfill\subfloat[][
\label{fig:lyft_myopic}Lyft
        Myopic Only]{\includegraphics[height=0.2\textwidth]{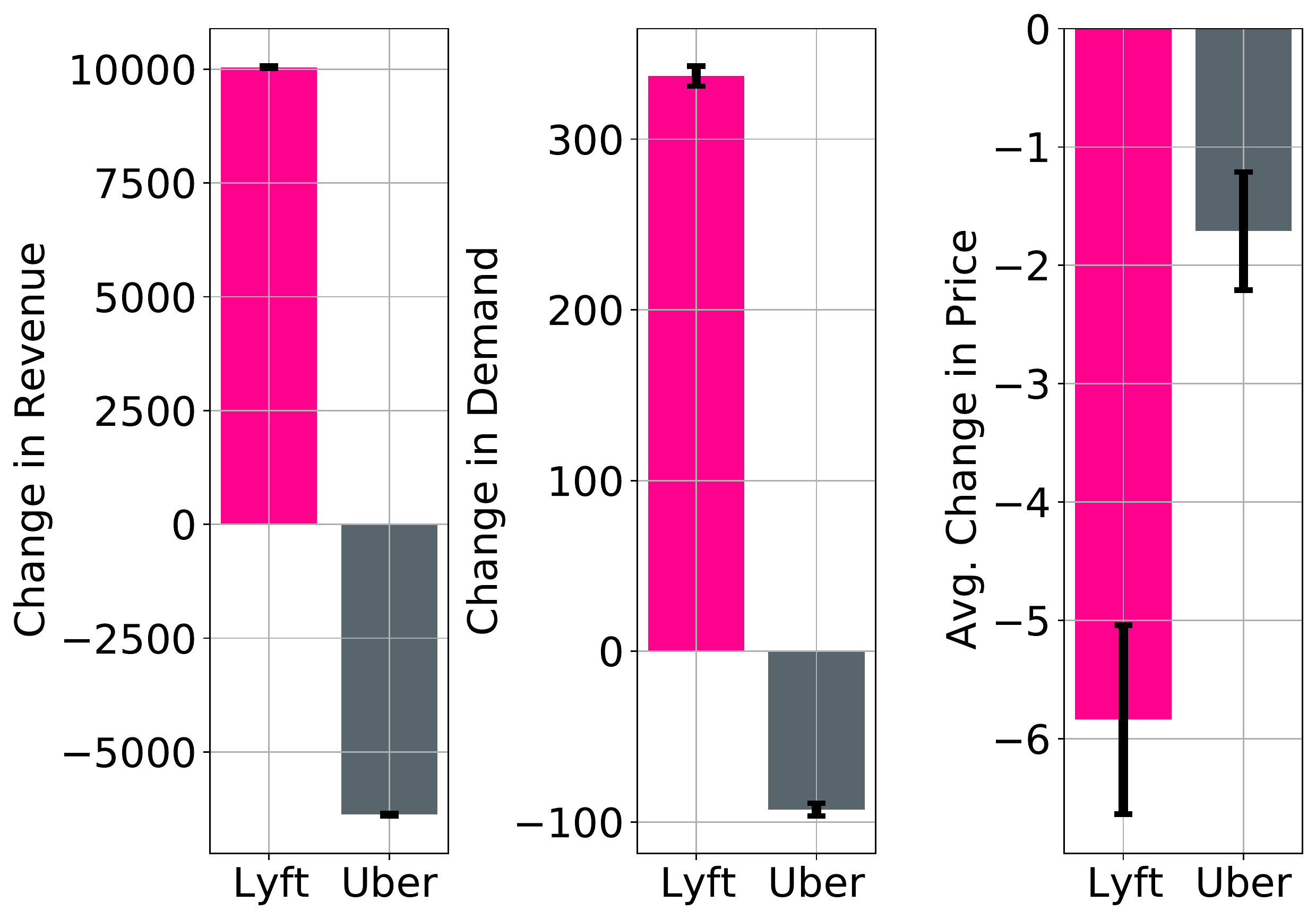}}

\subfloat[][\label{fig:both_partially_myopic}Both Partially Myopic
    ]{\includegraphics[height=0.2\textwidth]{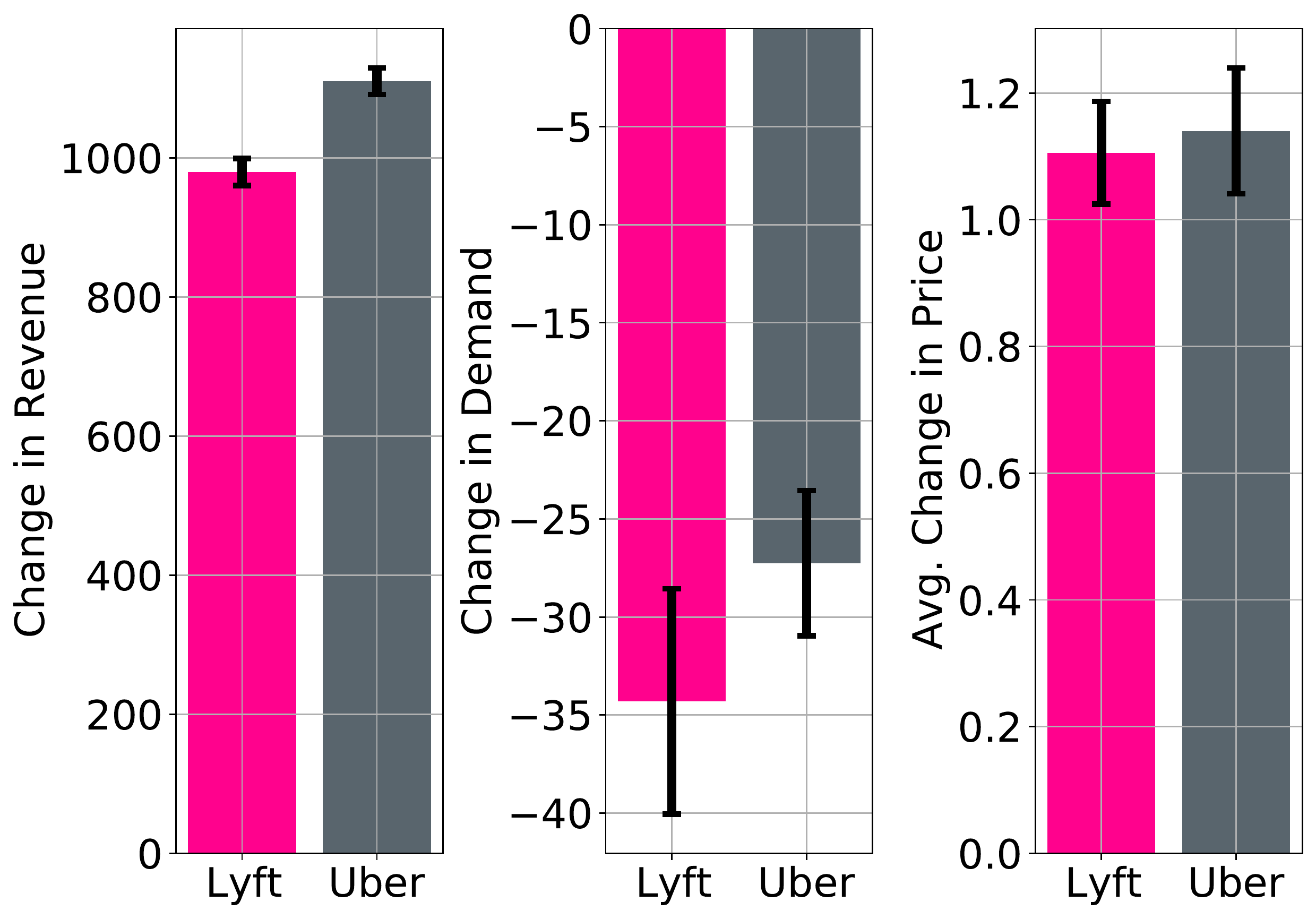}}\hfill\subfloat[][
\label{fig:uber_partially_myopic}Uber Partially Myopic Only]{\includegraphics[height=0.2\textwidth]{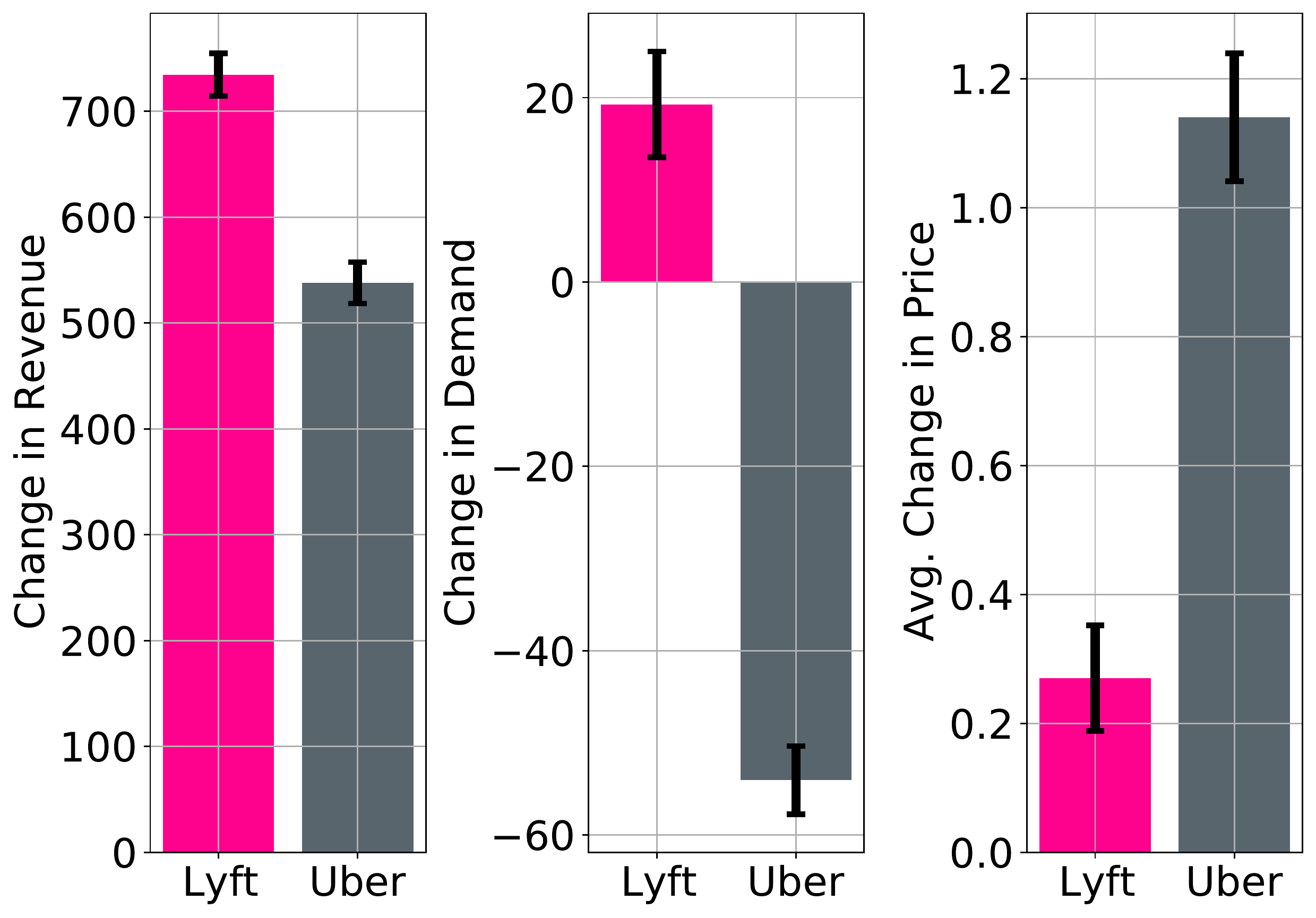}}\hfill\subfloat[][\label{fig:lyft_partially_myopic}Lyft
     Partially Myopic Only
    ]{\includegraphics[height=0.2\textwidth]{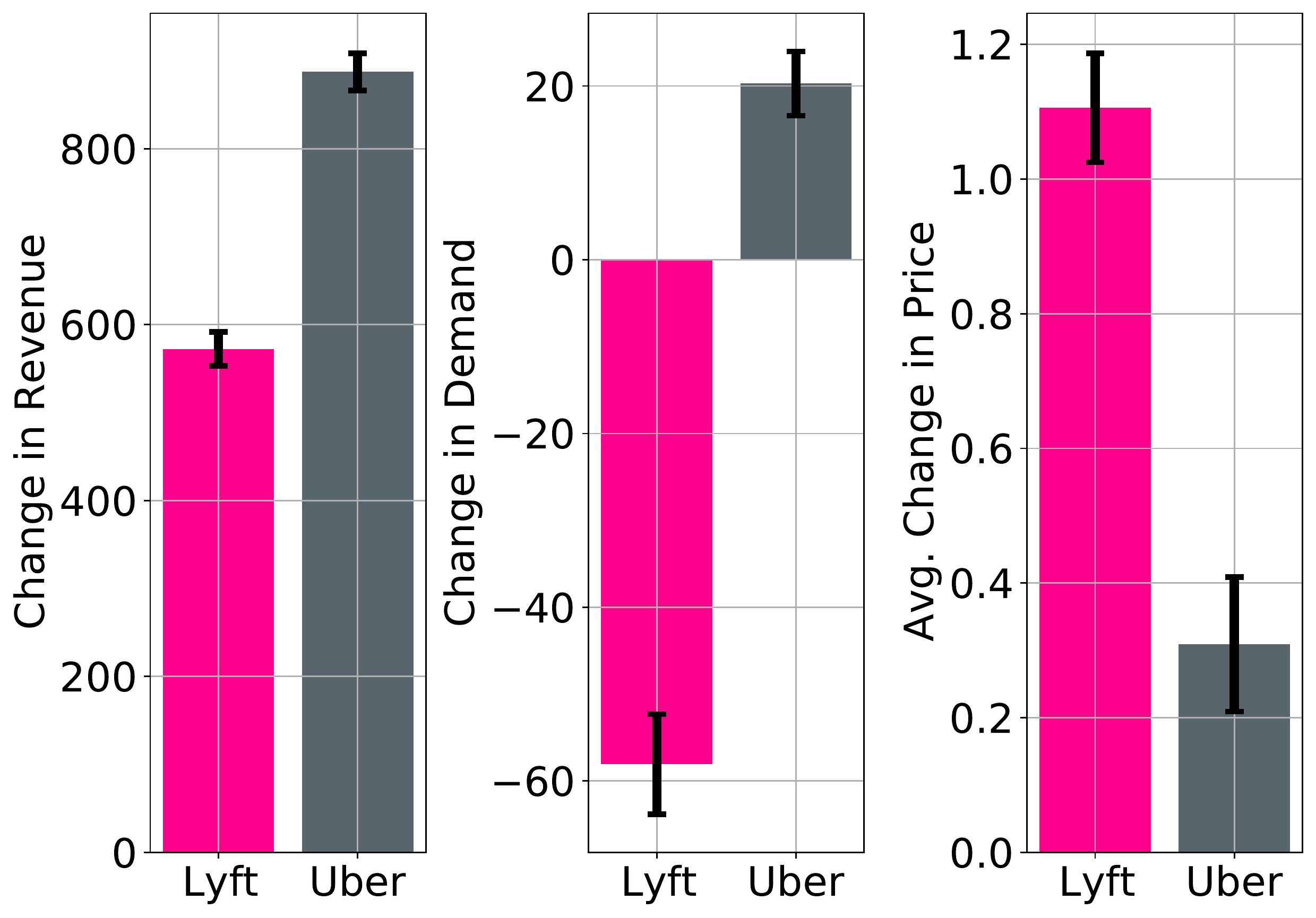}}     
    \caption{\textbf{Competition in Ride-Share Markets: Experiment 3.} Effects of players
    being (a)--(c) myopic
or  (d)--(f) partially myopic relative to Nash (not myopic, and consider competition).
Positive changes in revenue indicate the Nash equilibrium is better for that
player. When a player is myopic, they do not consider any performative effects
in their update---i.e., $g_i^t=\lambda_i  x_i^t-\frac{1}{2}\bz_i^t$---and when a player is partially myopic, they consider
their own performative effects, but not those of their competitor---i.e.,
$g_i^t=-(A_i-\lambda_i I)^\top x_i^t-\frac{1}{2}\bz_i^t$.  In (a)--(c), we
observe that when at least
one player is completely myopic, then at least one player is worse off at the
Nash equilibrium. In (d)--(f) we observe that when
at least one
player is partially myopic, the Nash equilibrium always is better for both
players.  }
    \label{fig:demand_revenue_oneplayer}
\end{figure}
        \paragraph{Experiment 3: Effect of Ignoring Performativity.} 
  We study the impact of players ignoring performative
    effects due to competition in the data distribution.    In Figure~\ref{fig:demand_revenue_oneplayer}, we explore the effects of
    players either being completely myopic---i.e., $g_i^t=\lambda_i x_i^t-\frac{1}{2}\bz_i^t$---or partially myopic---i.e., $g_i^t=-(A_i-\lambda_i
I)^\top x_i^t-\frac{1}{2}\bz_i^t$---on the change in revenue, demand and average
price (across locations)
from nominal at the Nash equilibrium.  Recall that players employing the
stochastic gradient method use the gradient  estimate $g_i^t=-(A_i-\lambda_i I)^\top
    x_i^t-\frac{1}{2}(\bz_i^t+A_{-i}x_{-i}^t)$; we refer to this as the
    non-myopic case since all performative effects are considered. 
Even when the players are myopic or partially myopic, the environment, however, does have these
    performative effects, meaning that $z_i=\bz_i+A_ix_i+A_{-i}x_{-i}$ and hence, the
    myopic player is in this sense ignoring or unaware of the fact that the data
    distribution is reacting to its competition's decisions. To compute the
    equilibrium outcomes we run the stochastic gradient method with a constant
    step-size of $\eta=0.001$. 

    In
Figure~\ref{fig:demand_revenue_oneplayer} (a)--(c), we observe that when at least
one player is completely myopic, then at least one player is worse off at the
Nash equilibrium in the sense that their revenue is lower. Interestingly, the
player that is worse off at the Nash is the non-myopic player. In
Figure~\ref{fig:demand_revenue_oneplayer} (d)--(f), on the
other hand, we observe that when
at least one
player is partially myopic, the Nash equilibrium always is better for both
players in the sense
that their individual revenues are higher at the Nash.

The values in Figure~\ref{fig:demand_revenue_oneplayer} represent the total demand and revenue changes, and average price
change across locations. It is also informative to examine the per-location
changes. Focusing in on the setting considered in
Figure~\ref{fig:demand_revenue_oneplayer} (d), we examine the per-location
price, revenue and demand. We see that the relative change depends on
    the location, however, the majority of locations see a decrease in demand,
    yet an increase in price and hence, revenue. This suggests that modeling performative effects due to
    competition can be beneficial for both players.

\section{Discussion}

The new class of games in this paper motivates interesting future work at the intersection of statistical learning theory and game theory. For instance, it is of interest to extend the present framework to handle more general parametric forms of the distributions $\mathcal{D}_i$.
Many multiplayer performative prediction problems exhibit a hierarchical structure such as a governing body that presides over an institution; hence, a  Stackelberg variant of multiplayer performative prediction is of interest. Along these lines, the multiplayer performative prediction problem is also of interest for mechanism design problems arising in applications such as recommender systems.
For instance, the recommendations that platforms  select at the Nash equilibrium influence the preferences of consumers (data-generators). A mechanism designer (e.g., the government) can place constraints on platforms to prevent them from \textit{manipulating} users' preferences to make their prediction tasks easier.

\bibliographystyle{plainnat}
\bibliography{2022jmlr_ddgames_refs}

\begin{thebibliography}{39}
\providecommand{\natexlab}[1]{#1}
\providecommand{\url}[1]{\texttt{#1}}
\expandafter\ifx\csname urlstyle\endcsname\relax
  \providecommand{\doi}[1]{doi: #1}\else
  \providecommand{\doi}{doi: \begingroup \urlstyle{rm}\Url}\fi

\bibitem[Ahmed(2000)]{ahmed2000strategic}
Shabbir Ahmed.
\newblock \emph{Strategic planning under uncertainty: Stochastic integer
  programming approaches}.
\newblock University of Illinois at Urbana-Champaign, 2000.

\bibitem[Angwin et~al.(2016)Angwin, Larson, Mattu, and
  Kirchner]{angwin2016propub}
Julia Angwin, Jeff Larson, Surya Mattu, and Lauren Kirchner.
\newblock Machine bias.
\newblock \emph{Propublica}, 2016.
\newblock URL
  \url{https://www.propublica.org/article/machine-bias-risk-assessments-in-criminal-sentencing}.

\bibitem[Bartlett et~al.(2019)Bartlett, Morse, Stanton, and
  Wallace]{bartlett2019consumer}
Robert Bartlett, Adair Morse, Richard Stanton, and Nancy Wallace.
\newblock Consumer-lending discrimination in the fintech era.
\newblock Technical report, National Bureau of Economic Research, 2019.

\bibitem[Bechavod et~al.(2020)Bechavod, Ligett, Wu, and
  Ziani]{bechavod2020causal}
Yahav Bechavod, Katrina Ligett, Zhiwei~Steven Wu, and Juba Ziani.
\newblock Causal feature discovery through strategic modification.
\newblock \emph{arXiv preprint arXiv:2002.07024}, 3, 2020.

\bibitem[Bravo et~al.(2018)Bravo, Leslie, and Mertikopoulos]{bravo2018bandit}
Mario Bravo, David~S Leslie, and Panayotis Mertikopoulos.
\newblock Bandit learning in concave $ n $-person games.
\newblock \emph{arXiv preprint arXiv:1810.01925}, 2018.

\bibitem[Brown et~al.(2020)Brown, Hod, and Kalemaj]{brown2020performative}
Gavin Brown, Shlomi Hod, and Iden Kalemaj.
\newblock Performative prediction in a stateful world.
\newblock \emph{arXiv preprint arXiv:2011.03885}, 2020.

\bibitem[Chasnov et~al.(2020)Chasnov, Ratliff, Mazumdar, and
  Burden]{chasnov20a}
Benjamin Chasnov, Lillian Ratliff, Eric Mazumdar, and Samuel Burden.
\newblock Convergence analysis of gradient-based learning in continuous games.
\newblock In Ryan~P. Adams and Vibhav Gogate, editors, \emph{Proceedings of The
  35th Uncertainty in Artificial Intelligence Conference}, volume 115 of
  \emph{Proceedings of Machine Learning Research}, pages 935--944. PMLR, 22--25
  Jul 2020.

\bibitem[Chen et~al.(2015)Chen, Mislove, and Wilson]{chen2015peek}
Le~Chen, Alan Mislove, and Christo Wilson.
\newblock Peeking beneath the hood of uber.
\newblock In \emph{Proceedings of the 2015 Internet Measurement Conference},
  IMC '15, page 495–508, 2015.
\newblock \doi{10.1145/2815675.2815681}.

\bibitem[Courtland(2018)]{courtland2018bias}
R~Courtland.
\newblock Bias detectives: the researchers striving to make algorithms fair.
\newblock \emph{Nature}, pages 357--360, 2018.
\newblock \doi{10.1038/d41586-018-05469-3}.

\bibitem[Cutler et~al.(2021)Cutler, Drusvyatskiy, and
  Harchaoui]{cutler2021stochastic}
Joshua Cutler, Dmitriy Drusvyatskiy, and Zaid Harchaoui.
\newblock Stochastic optimization under time drift: iterate averaging, step
  decay, and high probability guarantees.
\newblock \emph{arXiv preprint arXiv:2108.07356}, 2021.

\bibitem[Diethe et~al.(2019)Diethe, Borchert, Thereska, Balle, and
  Lawrence]{diethe2019continual}
Tom Diethe, Tom Borchert, Eno Thereska, Borja Balle, and Neil Lawrence.
\newblock Continual learning in practice.
\newblock \emph{arXiv preprint arXiv:1903.05202}, 2019.

\bibitem[Dieuleveut et~al.(2017)Dieuleveut, Flammarion, and
  Bach]{dieuleveut2017harder}
Aymeric Dieuleveut, Nicolas Flammarion, and Francis Bach.
\newblock Harder, better, faster, stronger convergence rates for least-squares
  regression.
\newblock \emph{The Journal of Machine Learning Research}, 18\penalty0
  (1):\penalty0 3520--3570, 2017.

\bibitem[Dong et~al.(2018)Dong, Roth, Schutzman, Waggoner, and
  Wu]{dong2018strategic}
Jinshuo Dong, Aaron Roth, Zachary Schutzman, Bo~Waggoner, and Zhiwei~Steven Wu.
\newblock Strategic classification from revealed preferences.
\newblock In \emph{Proceedings of the 2018 ACM Conference on Economics and
  Computation}, pages 55--70, 2018.

\bibitem[Drusvyatskiy and Xiao(2020)]{drusvyatskiy2020stochastic}
Dmitriy Drusvyatskiy and Lin Xiao.
\newblock Stochastic optimization with decision-dependent distributions.
\newblock \emph{arXiv preprint arXiv:2011.11173}, 2020.

\bibitem[Drusvyatskiy et~al.(2021)Drusvyatskiy, Fazel, and
  Ratliff]{drusvyatskiy2021improved}
Dmitriy Drusvyatskiy, Maryam Fazel, and Lillian~J. Ratliff.
\newblock Improved rates for derivative free gradient play in monotone games.
\newblock \emph{arXiv:2111.09456}, 2021.

\bibitem[Dupacov{\'a}(2006)]{dupacova2006optimization}
Jitka Dupacov{\'a}.
\newblock Optimization under exogenous and endogenous uncertainty.
\newblock \emph{University of West Bohemia in Pilsen}, 2006.

\bibitem[Ermoliev(1988)]{ermoliev1988}
Yu. Ermoliev.
\newblock Stochastic quasigradient methods.
\newblock In Yu. Ermoliev and R.~J-B Wets, editors, \emph{Numerical Techniques
  for Stochastic Optimization}, chapter~6, pages 141--185. Springer, 1988.

\bibitem[Gaivoronskii(1978)]{gaivoronskii1978}
A.~A. Gaivoronskii.
\newblock Nonstationary stochastic programming problems.
\newblock \emph{Cybernetics}, 14:\penalty0 575--579, 1978.

\bibitem[Ghadimi and Lan(2013)]{ghadimi2013optimal}
Saeed Ghadimi and Guanghui Lan.
\newblock Optimal stochastic approximation algorithms for strongly convex
  stochastic composite optimization, ii: shrinking procedures and optimal
  algorithms.
\newblock \emph{SIAM Journal on Optimization}, 23\penalty0 (4):\penalty0
  2061--2089, 2013.

\bibitem[Hardt et~al.(2016)Hardt, Megiddo, Papadimitriou, and
  Wootters]{hardt2016strategic}
Moritz Hardt, Nimrod Megiddo, Christos Papadimitriou, and Mary Wootters.
\newblock Strategic classification.
\newblock In \emph{Proceedings of the 2016 ACM conference on innovations in
  theoretical computer science}, pages 111--122, 2016.

\bibitem[Hellemo et~al.(2018)Hellemo, Barton, and
  Tomasgard]{hellemo2018decision}
Lars Hellemo, Paul~I Barton, and Asgeir Tomasgard.
\newblock Decision-dependent probabilities in stochastic programs with
  recourse.
\newblock \emph{Computational Management Science}, 15\penalty0 (3):\penalty0
  369--395, 2018.

\bibitem[Izzo et~al.(2021)Izzo, Ying, and Zou]{izzo2021learn}
Zachary Izzo, Lexing Ying, and James Zou.
\newblock How to learn when data reacts to your model: performative gradient
  descent.
\newblock \emph{arXiv preprint arXiv:2102.07698}, 2021.

\bibitem[Jonsbr{\aa}ten et~al.(1998)Jonsbr{\aa}ten, Wets, and
  Woodruff]{jonsbraaten1998class}
Tore~W Jonsbr{\aa}ten, Roger~JB Wets, and David~L Woodruff.
\newblock A class of stochastic programs withdecision dependent random
  elements.
\newblock \emph{Annals of Operations Research}, 82:\penalty0 83--106, 1998.

\bibitem[Lum and Isaac(2016)]{lum2016predict}
Kristian Lum and William Isaac.
\newblock To predict and serve?
\newblock \emph{Significance}, 13\penalty0 (5):\penalty0 14--19, 2016.

\bibitem[Mendler-D\"{u}nner et~al.(2020)Mendler-D\"{u}nner, Perdomo, Zrnic, and
  Hardt]{mendler2020stochastic}
Celestine Mendler-D\"{u}nner, Juan Perdomo, Tijana Zrnic, and Moritz Hardt.
\newblock Stochastic optimization for performative prediction.
\newblock In H.~Larochelle, M.~Ranzato, R.~Hadsell, M.~F. Balcan, and H.~Lin,
  editors, \emph{Advances in Neural Information Processing Systems}, volume~33,
  pages 4929--4939. Curran Associates, Inc., 2020.

\bibitem[Mertikopoulos and Zhou(2019)]{mertikopoulos2019learning}
Panayotis Mertikopoulos and Zhengyuan Zhou.
\newblock Learning in games with continuous action sets and unknown payoff
  functions.
\newblock \emph{Mathematical Programming}, 173\penalty0 (1):\penalty0 465--507,
  2019.

\bibitem[Miller et~al.(2020)Miller, Milli, and Hardt]{miller2020strategic}
John Miller, Smitha Milli, and Moritz Hardt.
\newblock Strategic classification is causal modeling in disguise.
\newblock In \emph{International Conference on Machine Learning}, pages
  6917--6926. PMLR, 2020.

\bibitem[Miller et~al.(2021)Miller, Perdomo, and Zrnic]{miller2021outside}
John Miller, Juan~C Perdomo, and Tijana Zrnic.
\newblock Outside the echo chamber: Optimizing the performative risk.
\newblock \emph{arXiv preprint arXiv:2102.08570}, 2021.

\bibitem[Perdomo et~al.(2020)Perdomo, Zrnic, Mendler-D{\"u}nner, and
  Hardt]{perdomo2020performative}
Juan Perdomo, Tijana Zrnic, Celestine Mendler-D{\"u}nner, and Moritz Hardt.
\newblock Performative prediction.
\newblock In \emph{International Conference on Machine Learning}, pages
  7599--7609. PMLR, 2020.

\bibitem[Ratliff et~al.(2016)Ratliff, Burden, and
  Sastry]{ratliff2016characterization}
Lillian~J Ratliff, Samuel~A Burden, and S~Shankar Sastry.
\newblock On the characterization of local nash equilibria in continuous games.
\newblock \emph{IEEE transactions on automatic control}, 61\penalty0
  (8):\penalty0 2301--2307, 2016.

\bibitem[Ray et~al.(2022)Ray, Ratliff, Drusvyatskiy, and Fazel]{ray2021dynamic}
Mitas Ray, Lillian~J. Ratliff, Dmitriy Drusvyatskiy, and Maryam Fazel.
\newblock Decision-dependent learning in geometrically decaying environments.
\newblock In \emph{Proceedings of the AAAI International Conference on
  Artificial Intelligence (AAAI)}, 2022.

\bibitem[Rosen(1965)]{rosen1965existence}
J~Ben Rosen.
\newblock Existence and uniqueness of equilibrium points for concave n-person
  games.
\newblock \emph{Econometrica: Journal of the Econometric Society}, pages
  520--534, 1965.

\bibitem[Rubinstein and Shapiro(1993)]{rubinstein1993discrete}
Reuven~Y Rubinstein and Alexander Shapiro.
\newblock \emph{Discrete event systems: Sensitivity analysis and stochastic
  optimization by the score function method}, volume~13.
\newblock Wiley, 1993.

\bibitem[Tatarenko and Kamgarpour(2019)]{tatarenko2019learning}
Tatiana Tatarenko and Maryam Kamgarpour.
\newblock Learning nash equilibria in monotone games.
\newblock In \emph{2019 IEEE 58th Conference on Decision and Control (CDC)},
  pages 3104--3109. IEEE, 2019.

\bibitem[Tatarenko and Kamgarpour(2020)]{tatarenko2020bandit}
Tatiana Tatarenko and Maryam Kamgarpour.
\newblock Bandit online learning of nash equilibria in monotone games.
\newblock \emph{arXiv preprint arXiv:2009.04258}, 2020.

\bibitem[Varaiya and Wets(1988)]{varaiya1988stochastic}
Pravin Varaiya and RJ-B Wets.
\newblock Stochastic dynamic optimization approaches and computation.
\newblock 1988.

\bibitem[Von~Stackelberg(2010)]{von2010market}
Heinrich Von~Stackelberg.
\newblock \emph{Market structure and equilibrium}.
\newblock Springer Science \& Business Media, 2010.

\bibitem[Wood et~al.(2021)Wood, Bianchin, and Dall’Anese]{wood2021online}
Killian Wood, Gianluca Bianchin, and Emiliano Dall’Anese.
\newblock Online projected gradient descent for stochastic optimization with
  decision-dependent distributions.
\newblock \emph{IEEE Control Systems Letters}, 2021.

\bibitem[Wu et~al.(2020)Wu, Dobriban, and Davidson]{wu2020deltagrad}
Yinjun Wu, Edgar Dobriban, and Susan Davidson.
\newblock Deltagrad: Rapid retraining of machine learning models.
\newblock In \emph{International Conference on Machine Learning}, pages
  10355--10366. PMLR, 2020.

\end{thebibliography}

$\ $\\
\appendix

\noindent{\Large \textbf{Appendix}}

\section{Stochastic gradient method with bias}\label{sec:append:generic_res}
In this section, we consider a variational inequality
\begin{equation}\label{eqn:VI}
0\in G(x)+N_{\X}(x),
\end{equation}
where $\X\subset\R^d$ is a closed convex set and $G\colon\R^d\to\R^d$ is an $L$-Lipschitz continuous and $\alpha$-strongly monotone map. We will analyze the stochastic gradient method, which in each iteration performs the update:
\begin{equation}\label{eqn:sgd_bias}
x^{t+1}=\proj_{\X} (x^t-\eta g^t),
\end{equation}
where $\eta>0$ is a fixed stepsize and $g_t$ is a sequence of random variables, which approximates $G(x^t)$. In particular, it will be crucial for us to allow $g^t$ to be a {\em biased} estimator of $G(x^t)$.
Formally, we make the following assumption on the randomness in the process. Throughout, $x^{\star}$ denotes the unique solution of \eqref{eqn:VI}.
\begin{assumption}[Stochastic framework]
Suppose that there exists a filtered probability space $(\Omega, \mathcal{F},\mathbb{F},\mathbb{P})$ with filtration $\mathbb{F}=(\mathcal{F}_t)_{t\geq 0}$ such that $\mathcal{F}_0=\{\emptyset, \Omega\}$. Suppose moreover that $g_t$ is $\mathcal{F}_{t+1}$-measurable and there exist  constants $B,\sigma\geq 0$ and $\mathcal{F}_t$-measurable random variables $C_t,\sigma_t\geq 0$ satisfying the bias/variance bounds
\begin{align*}
{\rm \bf (Bias)}\qquad\qquad \|\mathbb{E}_{t}g^t-G(x^t)]\|&\leq C_t+B\|x^t-x^*\|,\\
{\rm \bf(Variance)}\qquad\qquad \mathbb{E}_t\|g^t-\mathbb{E}_t[g^t]\|^2&\leq \sigma^2_t+D^2\|x^t-x^{\star}\|^2,
\end{align*}
where $\mathbb{E}_t=\E[\cdot\mid \mathcal{F}_t]$ denotes the conditional expectation.
\end{assumption}

The following is a one-step improvement guarantee for the stochastic gradient method in the two conceptually distinct cases $C_t=0$ and $B=0$. In the case $C_t=0$ , the bias $\mathbb{E}_{t}g^t-G(x^t)$ shrinks as the iterates approach $x^{\star}$. The theorem shows that as long as $B/\alpha<1$, with a sufficiently small stepsize $\eta$, the method can converge to an arbitrarily small neighborhood of $x^{\star}$. In the case $B=0$, one can only hope to convergence to a neighborhood of the minimizer whose radius depends on $\{C_t\}_{t\geq 0}$.
\begin{theorem}[One step improvement]\label{thm:sgd}
The following are true.
\begin{itemize}
\item {\bf (Benign bias)} Suppose $C_t\equiv 0$ for all $t$. Set $\rho:=B/\alpha$ and suppose that we are in the regime $\rho\in (0,1)$. Then with any $\eta<\frac{\alpha(1-\rho)}{8 L^2}$, the stochastic gradient method \eqref{eqn:sgd_bias} generates a sequence $x^t$ satisfying
\begin{equation}\label{eqn:ourtarget}
\mathbb{E}_t\|x^{t+1}-x^{\star}\|^2\leq \frac{1+2\eta\alpha \rho +4\eta^2 D^2+2\eta^2\alpha^2 \rho^2}{1+2\eta\alpha(\frac{1+\rho}{2})}\|x^t-x^{\star}\|^2+\frac{4\eta^2\sigma^2_t}{1+2\eta\alpha(\frac{1+\rho}{2})}.
\end{equation}
\item {\bf (Offset bias)} Suppose $B\equiv 0$.
Then with any $\eta\leq \frac{\alpha}{4 L^2}$, the stochastic gradient method \eqref{eqn:sgd_bias} generates a sequence $x^t$ satisfying
\begin{equation}\label{eqn:first_bound}
\mathbb{E}_t\|x^{t+1}-x^{\star}\|^2\leq \frac{1+2\eta^2 D^2}{1+\eta\alpha} \|x^t-x^{\star}\|^2+\frac{2\eta^2\sigma^2_t}{1+\eta\alpha}+\frac{2\eta C_t^2}{\alpha(1+\eta\alpha)}.\end{equation}
Moreover, in the zero bias setting $B\equiv C_t\equiv 0$, the estimate \eqref{eqn:first_bound} holds in the slightly wider parameter regime $\eta\leq \frac{\alpha}{2 L^2}$.
\end{itemize}
\end{theorem}
\begin{proof}
We begin with the first claim. To this end, suppose $C_t\equiv 0$ for all $t$. Set $\rho:=B/\alpha$ and suppose that we are in the regime $\rho\in (0,1)$.
Fix three constants $\Delta_1,\Delta_2,\Delta_3>0$ to be specified later. Noting that $x^{t+1}$ is the minimizer of the $1$-strongly convex function $x\mapsto \frac{1}{2}\|x^t-\eta g^t-x\|^2$ over $\X$, we deduce
\begin{align*}
\frac{1}{2}\|x^{t+1}-x^{\star}\|^2
\leq \frac{1}{2}\|x^t-\eta g^t-x^{\star}\|^2-\frac{1}{2}\|x^t-\eta g^t-x^{t+1}\|^2.
\end{align*}
Expanding the squares on the right hand side and combining terms yields
\begin{align*}
\frac{1}{2}\|x^{t+1}-x^{\star}\|^2&\leq \frac{1}{2}\|x^t-x^{\star}\|^2-\eta\langle   g^t, x^{t+1}-x^{\star}\rangle -\frac{1}{2}\|x^{t+1}-x^t\|^2\\
&=\frac{1}{2}\|x^t-x^{\star}\|^2-\eta\langle   g^t, x^{t}-x^{\star}\rangle -\frac{1}{2}\|x^{t+1}-x^t\|^2-\eta\langle   g^t, x^{t+1}-x^{t}\rangle.
\end{align*}
Setting $\mu^t:=\E_t[g^t]$, we successively compute
\begin{align}
\frac{1}{2}\mathbb{E}_t\|x^{t+1}-x^{\star}\|^2&\leq
 \frac{1}{2}\|x^t-x^{\star}\|^2-\eta\langle    \mathbb{E}_t g^t, x^{t}-x^{\star}\rangle -\frac{1}{2}\mathbb{E}_t\|x^{t+1}-x^t\|^2-\eta \mathbb{E}_t\langle   g^t, x^{t+1}-x^{t}\rangle \notag\\
 &\leq \frac{1}{2}\|x^t-x^{\star}\|^2-\eta\langle    \mu^t, x^{t}-x^{\star}\rangle -\frac{1}{2}\mathbb{E}_t\|x^{t+1}-x^t\|^2-\eta \mathbb{E}_t\langle   g^t, x^{t+1}-x^{t}\rangle\notag\\
&=\frac{1}{2}\|x^t-x^{\star}\|^2-\eta\mathbb{E}_t\langle    G(x^{t+1}), x^{t+1}-x^{\star}\rangle -\frac{1}{2}\mathbb{E}_t\|x^{t+1}-x^t\|^2\notag\\
&+\eta \underbrace{\mathbb{E}_t\langle   g^t-\mu^t, x^{t}-x^{t+1}\rangle}_{P_1}+\eta\underbrace{\mathbb{E}_t\langle    \mu^t-G(x^{t+1}), x^{\star}-x^{t+1}\rangle}_{P_2}]. \label{eqn:startyoung}
\end{align}
Taking into account strong monotonicity of $G$, we deduce $\langle G(x^{t+1}), x^{t+1}-x^{\star}\rangle\geq \alpha\|x^{t+1}-x^{\star}\|^2$ and therefore
\begin{equation}\label{eqn:start_from_here}
\frac{1+2\eta\alpha}{2}\mathbb{E}_t\|x^{t+1}-x^{\star}\|^2\leq \frac{1}{2}\|x^t-x^{\star}\|^2 -\frac{1}{2}\mathbb{E}_t\|x^{t+1}-x^t\|^2+\eta(P_1+P_2).
\end{equation}
Using Young's inequality, we may upper bound $P_1$ and $P_2$ by:
\begin{equation}\label{eqn:ochenyoung}
P_1\leq \frac{\sigma^2_t+D^2\|x^t-x^{\star}\|^2}{2\Delta_1}+\frac{\Delta_1 \mathbb{E}_t\|x^{t+1}-x^t\|^2}{2}.
\end{equation}
Next, we decompose $P_2$ as
\begin{equation}\label{eqn:p2_bound}
P_2=\langle \mu^t-G(x^{t}),x^{\star}-x^{t}\rangle+\mathbb{E}_t\langle \mu^t-G(x^t),x^t-x^{t+1}\rangle+\mathbb{E}_t\langle G(x^t)-G(x^{t+1}),x^{\star}-x^{t+1}\rangle.
\end{equation}
We bound each of the three products in turn. The first bound follows from our assumption on the bias:
\begin{equation}\label{eqn:trick1}
\langle \mu^t-G(x^{t}),x^{\star}-x^{t}\rangle\leq B\|x^t-x^{\star}\|^2.
\end{equation}
The second bound uses Young's inequality and our assumption on the bias:
\begin{equation}\label{eqn:trick1}
\begin{aligned}
\mathbb{E}_t\langle \mu^t-G(x^t),x^t-x^{t+1}\rangle&\leq \frac{\Delta_2\|\mu^t-G(x^t)\|^2}{2}+\frac{\mathbb{E}_t\|x^t-x^{t+1}\|^2}{2\Delta_2}\\
&\leq \frac{\Delta_2B^2\|x^t-x^{\star}\|^2}{2}+\frac{\mathbb{E}_t\|x^t-x^{t+1}\|^2}{2\Delta_2}
\end{aligned}
\end{equation}
The third inequality uses Young's inequality and Lipschitz continuity of $G$:
\begin{equation}\label{eqn:strick}
\begin{aligned}
\mathbb{E}_t\langle G(x^t)-G(x^{t+1}),x^{\star}-x^{t+1}\rangle&\leq \frac{\Delta_3\|G(x^t)-G(x^{t+1})\|^2}{2}+\frac{\mathbb{E}_t\|x^{\star}-x^{t+1}\|^2}{2\Delta_3}\\
&\leq  \frac{\Delta_3L^2\|x^t-x^{t+1}\|^2}{2}+\frac{\mathbb{E}_t\|x^{\star}-x^{t+1}\|^2}{2\Delta_3}
\end{aligned}
\end{equation}
Putting together all the estimates  \eqref{eqn:start_from_here}-\eqref{eqn:strick} yields
\begin{equation}\label{eqn:whereitas}
\begin{aligned}
\frac{1+2\eta\alpha-2\eta\Delta_3^{-1}}{2}\mathbb{E}_t\|x^{t+1}-x^{\star}\|^2&\leq \frac{1+\eta D^2\Delta^{-1}_1+2\eta B +\eta\Delta_2 B^2}{2}\|x^t-x^{\star}\|^2\\
&~~~-\frac{1-\eta\Delta_1-\eta\Delta_2^{-1}-\eta\Delta_3 L^2}{2}\mathbb{E}_t\|x^{t+1}-x^{t}\|^2+\frac{\eta\sigma^2_t}{2\Delta_1}.
\end{aligned}
\end{equation}
Let us now set
$$\Delta_3^{-1}=\frac{(1-\rho)\alpha}{2}, \qquad \Delta_1=\frac{1}{4\eta},\qquad \Delta_2^{-1}:=\eta^{-1}-\Delta_1-\Delta_3 L^2.$$
Notice $\Delta_1\leq \frac{1}{2\eta}-\Delta_3 L^2$ by our assumption that  $\eta\leq\frac{\alpha(1-\rho)}{8 L^2}$. In particular, this implies $\Delta_2^{-1}\geq\frac{1}{2\eta}$.
Notice that our choice of $\Delta_2$ ensures that the
 the fraction multiplying $\mathbb{E}_t\|x^{t+1}-x^{t}\|^2$ in \eqref{eqn:whereitas} is zero. We therefore deduce
 \begin{align*}
\mathbb{E}_t\|x^{t+1}-x^{\star}\|^2&\leq \frac{1+\eta D^2\Delta^{-1}_1+2\eta B +\eta\Delta_2 B^2}{1+2\eta\alpha-2\eta\Delta_3^{-1}}\|x^t-x^{\star}\|^2+\frac{\eta\sigma^2_t}{\Delta_1(1+2\eta\alpha-2\eta\Delta_3^{-1})}\\
&\leq \frac{1+2\eta\alpha \rho +4\eta^2 D^2+2\eta^2\alpha^2 \rho^2}{1+2\eta\alpha(\frac{1+\rho}{2})}\|x^t-x^{\star}\|^2+\frac{4\eta^2\sigma^2_t}{1+2\eta\alpha(\frac{1+\rho}{2})},
\end{align*}
thereby completing the proof of \eqref{eqn:ourtarget}.

We next prove the second claim. To this end, suppose $B=0$. All of the reasoning leading up to and including \eqref{eqn:ochenyoung} is valid. Continuing from this point, using Young's inequality, we upper bound $P_2$ by:
\begin{equation}\label{eqn:ochenyoung2}
\begin{aligned}
P_2\leq \frac{\mathbb{E}_t\|\mu^t-G(x^{t+1})\|^2}{2\Delta_2}+\frac{\Delta_2\mathbb{E}_t\|x^{t+1}-x^{\star}\|^2}{2}.
\end{aligned}
\end{equation}
Next observe
\begin{equation}\label{eqn:square_expand}
\begin{aligned}
\mathbb{E}_t\|\mu^t-G(x^{t+1})\|^2&\leq 2\mathbb{E}_t\|\mu^t-G(x^t)\|^2+2\mathbb{E}_t\|G(x^t)-G(x^{t+1})\|^2,\\
&\leq 2C_t^2+2L^2\|x^t-x^{t+1}\|^2
\end{aligned}
\end{equation}
and therefore
\begin{equation}\label{eqn:final_piece}
P_2\leq \frac{2C_t^2+2L^2\|x^t-x^{t+1}\|^2}{2\Delta_2}+\frac{\Delta_2\mathbb{E}_t\|x^{t+1}-x^{\star}\|^2}{2}
\end{equation}
Putting the estimates \eqref{eqn:start_from_here}, \eqref{eqn:ochenyoung}, and \eqref{eqn:final_piece}  together yields:
\begin{equation}\label{eqn:almost_there2}
\begin{aligned}
\frac{1+\eta(2\alpha-\Delta_2)}{2}\mathbb{E}_t\|x^{t+1}-x^{\star}\|^2&\leq  \frac{1+\eta D^2\Delta_1^{-1} }{2}\|x^t-x^{\star}\|^2 \\
&+\frac{\eta \sigma^2_t}{2\Delta_1}+\frac{2\eta C_t^2 \Delta_2^{-1}}{2}-\frac{1-2\eta L^2 \Delta_2^{-1}-\eta\Delta_1}{2}\mathbb{E}_t\|x^{t+1}-x^t\|^2
\end{aligned}
\end{equation}
Setting $\Delta_2=\alpha$ and $\Delta_1=\eta^{-1}- \frac{2 L^2}{\alpha}$ ensures that the last term on the right is zero. Notice that our assumption $\eta\leq \frac{\alpha}{4 L^2}$ ensures $\Delta_1\geq \frac{1}{2\eta}$.
Rearranging \eqref{eqn:almost_there2} directly yields \eqref{eqn:first_bound}.
In the case $B\equiv C_t\equiv 0$, instead of \eqref{eqn:square_expand} we may simply use the bound $\mathbb{E}_t\|\mu^t-G(x^{t+1})\|^2=\mathbb{E}_t\|G(x^t)-G(x^{t+1})\|\leq L^2\|x^t-x^{t+1}\|^2$. Continuing in the same way as above yields the improved estimate. %
\end{proof}

\section{Technical results on sequences}\label{sec:technical_sequences}
The following lemma is standard and follows from a simple inductive argument.
\begin{lemma}\label{lem:basic_recurs_lemma}
Consider a sequence $D_t\geq 0$ for $t\geq 1$ and constants $t_0\geq 0$, $a>0$ satisfying
\begin{equation}\label{eqn:key_recurs}
D_{t+1}\leq (1-\tfrac{2}{t+t_0})D_t+\tfrac{a}{(t+t_0)^2}
\end{equation}
Then the estimate holds:
\begin{equation}\label{eqn:claim_est}
D_{t}\leq \frac{\max\{(1+t_0)D_1,a\}}{t+t_0}\qquad \forall t\geq 1.
\end{equation}
\end{lemma}

We will need the following more general version of the lemma.

\begin{lemma}\label{lem:lemma_on_seq}
Consider a sequence $D_t\geq 0$ for $t\geq 1$ and constants $t_0\geq 0$, $a,b>0$ satisfying
\begin{equation}\label{eqn:key_recurs}
D_{t+1}\leq (1-\tfrac{2}{t+t_0})D_t+\tfrac{a}{(t+t_0)^2}+\tfrac{b}{(t+t_0)^3}.
\end{equation}
Then the estimate holds:
\begin{equation}\label{eqn:claim_est}
D_{t}\leq \frac{\max\{(1+t_0)D_1/2,a\}}{t+t_0}+\frac{\max\{(1+t_0)^{3/2}D_1/2,b\}}{(t+t_0)^{3/2}}\qquad \forall t\geq 1.
\end{equation}
\end{lemma}
\begin{proof}
Clearly the recursion \eqref{eqn:key_recurs} continues to hold with $a$ and $b$
replaced by the bigger quantities $\max\{(1+t_0)D_1/2,a\}$ and
$\max\{(1+t_0)D_1/2,b\}$, respectively. Therefore abusing notation, we will make this substitution. As a consequence,
the claimed estimate \eqref{eqn:claim_est}  holds automatically for the base case $t=1$. As an inductive assumption, suppose the claim  \eqref{eqn:claim_est} is true for $D_t$. Set $s=t+t_0$.
We then deduce
\begin{align*}
D_{t+1}&\leq \left(1-\frac{2}{s}\right)D_t+\frac{ a}{s^2}+\frac{b}{s^3}\\
&\leq \left(1-\frac{2}{s}\right)\left(\frac{ a}{s}+\frac{ b}{s^{3/2}}\right)+\frac{ a}{s^2}+\frac{b}{s^3}\\
&\leq a\left(\frac{ 1}{s}-\frac{ 1}{s^2}\right)+b\left(\frac{1}{s^{3/2}}-\frac{2}{s^{5/2}}+\frac{1}{s^3}\right).
\end{align*}
Elementary algebraic manipulations show $\frac{ 1}{s}-\frac{ 1}{s^2}\leq \frac{1}{s+1}$. Define the function
$\phi(s)=\frac{1}{s^{3/2}}-\frac{2}{s^{5/2}}+\frac{1}{s^3}-\frac{1}{(1+s)^{3/2}}$. Elementary calculus shows that $\phi$ is increasing on the interval $s\in [1,\infty)$. Since $\phi$ tends to zero as $s$ tends to infinity, it follows that $\phi$ is negative on the interval $[1,\infty)$. We therefore conclude
$D_{t+1}\leq \frac{a}{s+1}+\frac{b}{(1+s)^{3/2}}$ as claimed.
\end{proof}

\section{Online Least Squares}
\label{app:proof_adaptive}

In this appendix section, we record basic and well-known results on estimation in online least squares, following \cite{dieuleveut2017harder}.

\begin{lemma}\label{lem:least_square_step} Fix a probability space $(\Omega, \mathcal{F},\mathbb{P})$ with two sub-$\sigma$-algebras $\mathcal{G}_1\subset \mathcal{G}_2\subset \mathcal{F}$ . Define the function
$$f(B)=\frac{1}{2}\|By-b\|^2,$$
where  $B\colon \Omega\to \R^{m_1\times m_2}$,
$b\colon\Omega\to \R^{m_1}$, and $y\colon \Omega\to \R^{m_2}$ are random variables. 
Suppose moreover that there exist random variables $V\colon\Omega \to \R^{m_1\times m_2}$ and $\sigma\colon\Omega\to\R$ satisfying the following.
\begin{enumerate}
\item $B$, $V$, and $\sigma_1$ are $\mathcal{G}_1$-measurable.
\item $y$ is $\mathcal{G}_2$-measurable.
\item The estimates, $\E[b\mid \mathcal{G}_2]=Vy$ and $\E[\|Vy-b\|^2\mid \mathcal{G}_2]\leq \sigma^2$ ,hold.
\item There exist constants $\lambda_1, \lambda_2,R >0$ satisfying 
$$\lambda_1 I\preceq \mb{E}[y y^\top\mid  \mathcal{G}_1],\qquad \mb{E}[\|y\|^2\mid  \mathcal{G}_1]\leq \lambda_2, \qquad\textrm{and}\qquad \E[\|y\|^2y y^{\top}\mid \mathcal{G}_1]\preceq R^2 \E[y y^{\top}].$$
\end{enumerate}
Then for any constant $\nu\in (0,\frac{2}{R^2})$, the gradient step  $B^+=B-\nu(By-b)y^{\top}$ satisfies the bound:
$$\frac{1}{2}\mathbb{E}[\|B^+-V\|^2_F\mid \mathcal{G}_1]\leq\frac{1-\lambda_1\nu(2-\nu R^2)}{2}\|B-V\|^2_F+\frac{\nu^2\sigma^2\lambda_2}{2}.$$
\end{lemma}
\begin{proof}
Expanding the squared norm yields:
\begin{align*}
\frac{1}{2}\|B^+-V\|^2_F=\frac{1}{2}\|B-V-\nu
(By-b)y^{\top}\|^2_F&=\frac{1}{2}\|B-V\|^2_F-\nu\langle
B-V,(By-b)y^{\top}\rangle\\
&\quad+\frac{\nu^2}{2}\|(By-b)y^{\top}\|^2_F.
\end{align*}
Taking conditional expectations, we conclude
\begin{equation}\label{eqn:online_Least}
\begin{aligned}
\frac{1}{2}\mathbb{E}[\|B^+-V\|^2_F\mid \mathcal{G}_2]&=\frac{1}{2}\|B-V\|^2_F-\nu\langle B-V,(By-\mathbb{E}[b\mid \mathcal{G}_2])y^{\top}\rangle+\frac{\nu^2}{2}\mathbb{E}[\|(By-b)y^{\top}\|^2_F\mid \mathcal{G}_2]\\
&=\frac{1}{2}\|B-V\|^2_F-\nu\|(B-V)y\|_F^2+\frac{\nu^2}{2}\|y\|^2 \mathbb{E}[\|By-b\|^2_F\mid \mathcal{G}_2].
\end{aligned}
\end{equation}
Next, observe 
$$\|By-b\|^2_F=\|(B-V)y\|^2+\|Vy-b\|^2+2\langle By-Vy,Vy-b\rangle.$$
Taking the conditional expectation $\mathbb{E}[\cdot\mid\mathcal{G}_2]$, the last term vanishes, and therefore we deduce $\mathbb{E}[\|By-b\|^2_F\mid \mathcal{F}']\leq \|(B-V)y\|^2+\sigma^2$. Combining this with \eqref{eqn:online_Least}  we compute
\begin{align*}
\frac{1}{2}\mathbb{E}[\|B^+-V\|^2_F\mid \mathcal{G}_2]&\leq \frac{1}{2}\|B-V\|^2_F-\nu\|(B-V)y\|_F^2+\frac{\nu^2}{2}\|y\|^2\|(B-V)y\|^2+\frac{\nu^2\sigma^2}{2}\|y\|^2.
\end{align*}
Taking expectations with respect to $\mathcal{G}_1$ and using the tower rule, we deduce
\begin{align*}
\frac{1}{2}\mathbb{E}[\|B^+-V\|^2_F\mid
\mathcal{G}_1]&\leq\frac{1}{2}\|B-V\|^2_F-\nu\E[\|(B-V)y\|_F^2 \mid
\mathcal{G}_1]+\frac{\nu^2}{2}\E[\|y\|^2\|(B-V)y\|^2\mid \mathcal{G}_1]\\
&\quad+\frac{\nu^2\sigma^2\lambda_2}{2}.
\end{align*}
Observe next
$$\E[\|y\|^2\|(B-V)y\|^2\mid \mathcal{G}_1]=\langle (B-V)(B-V)^{\top}, \mathbb{E}[\|y\|^2y y^{\top}\mid \mathcal{G}_1]]\rangle\leq R^2\mathbb{E}[\|(B-V)y\|_F^2\mid \mathcal{G}_1],$$
and therefore 
\begin{align*}
\frac{1}{2}\mathbb{E}[\|B^+-V\|^2_F\mid \mathcal{G}_1]&\leq\frac{1}{2}\|B-V\|^2_F-(\nu-\frac{\nu^2 R^2}{2})\E[\|(B-V)y\|_F^2 \mid \mathcal{G}_1]+\frac{\nu^2\sigma^2\lambda_2}{2}
\end{align*}
Note that $\nu\geq \frac{\nu^2R^2}{2}$.
Next we estimate
\begin{align*}
\E[\|(B-V)y\|_F^2\mid \mathcal{G}_1]&=\tr ((B-V)^\top(B-V)\E[yy^{\top}\mid \mathcal{G}_1]]\rangle\geq \lambda_1\|B-V\|^2_F.
\end{align*}
This completes the proof.
\end{proof}

\end{document}